\documentclass[12pt]{article}
\usepackage{amsmath}
\usepackage{graphicx}
\usepackage{enumerate}
\usepackage[authoryear]{natbib}
\usepackage{url} %
\usepackage{latexsym}
\usepackage{epsfig}
\usepackage{subfigure}
\usepackage{threeparttable}
\usepackage{amssymb}
\usepackage{mathtools}
\usepackage{amsthm}
\newtheorem{assumption}{Assumption}
\usepackage{array}
\usepackage{caption}
\usepackage{color}
\usepackage{float}
\usepackage[bookmarks=false,hypertexnames=false]{hyperref}\hypersetup{colorlinks=true, linkcolor=blue, filecolor=gray, urlcolor=blue, citecolor=blue}
\usepackage{mathrsfs}
\usepackage{indentfirst}
\usepackage{psfrag}
\usepackage{multirow}
\usepackage{booktabs}
\usepackage{longtable}
\usepackage{bm}
\usepackage{rotating}
\usepackage{bookmark}
\usepackage{algorithm}
\usepackage{algpseudocode}
\usepackage{xr}
\usepackage{chngcntr}

\newtheorem{thm}{Theorem}
\newtheorem{lemma}{Lemma}

\newtheorem{remark}{Remark}

\newtheorem{prop}{Proposition} 
\renewcommand{\baselinestretch}{1}
\def\IR{\rm I \kern-0.20em R}

\renewcommand{\baselinestretch}{1.5}
\begin{document}
\def\spacingset#1{\renewcommand{\baselinestretch}%
{#1}\small\normalsize} \spacingset{1}
\title{\bf Robust Bond Risk Premia Predictability Test in the Quantiles 
}
\author{Xiaosai Liao\hspace{.2cm}\\
Department of Economics, The Chinese University of Hong Kong \\
and \\
Xinjue Li \\
School of Economics, Yunnan University\\
and \\
Qingliang Fan \thanks{{\it Corresponding author:} Q. Fan \textcolor{blue}{(E-mail: michaelqfan@gmail.com).}}\\
Department of Economics, The Chinese University of Hong Kong \\
}
\maketitle
\vspace{-0.5cm}
\begin{abstract} 
Different from existing literature on testing the macro-spanning hypothesis of bond risk premia, which only considers mean regressions, this paper investigates whether the yield curve represented by CP factor \citep{CochranePiazzesi2005} contains all available information about future bond returns in a predictive quantile regression with many other macroeconomic variables. In this study, we introduce the Trend in Debt Holding (TDH) as a novel predictor, testing it alongside established macro indicators such as Trend Inflation (TI) \citep{CieslakPovala2015}, and macro factors from \cite{ludvigson2009macro}. A significant challenge in this study is the invalidity of traditional quantile model inference approaches, given the high persistence of many macro variables involved. Furthermore, the existing methods addressing this issue do not perform well in the marginal test with many highly persistent predictors. Thus, we suggest a robust inference approach, whose size and power performance are shown to be better than existing tests. Using data from 1980--2022, the macro-spanning hypothesis is strongly supported at center quantiles by the empirical finding that the CP factor has predictive power while all other macro variables have negligible predictive power in this case. On the other hand, the evidence against the macro-spanning hypothesis is found at tail quantiles, in which TDH has predictive power at right tail quantiles while TI has predictive power at both tails quantiles. Finally, we show the performance of in-sample and out-of-sample predictions implemented by the proposed method are better than existing methods. 
\end{abstract}
\noindent%
\vfill
\newpage
\spacingset{1.5} %
\section{Introduction}
\label{4}
Treasury bonds are a crucial component in numerous investment portfolios. Therefore, understanding the risk and return dynamics of this asset class is of fundamental importance from an economic perspective. One focal topic in the study of U.S. Treasury bond returns is the famous macro-spanning hypothesis/puzzle on whether the yield curve contains all available information about future bond risk premia. \footnote{See the beginning of Section \ref{Model Framework} for the mathematical representation of the hypothesis. } Numerous studies have explored the predictive power of macro variables on bond risk premiums while controlling for the CP factor \citep{CochranePiazzesi2005} representing the information of the yield curve. These studies include 
\cite{CooperPriestley2008}, \cite{ludvigson2009macro}, 
\citet{Bansal2013}, \cite{Joslin2014}, \cite{GreenwoodVayanos2014}, \cite{CieslakPovala2015}, \cite{Ghysels2018}, \cite{Bauer2018}, \cite{ZhaoZhouguofu2021}.
The literature mentioned above employs mean regression to assess the macro-spanning hypothesis, which suggests that there is no predictability of bond returns beyond the information provided by the yield curve. However, it is uncertain whether these results are consistent across the entire distribution or mainly pertain to the mean return. This paper studies the predictability of bond risk premia in a predictive quantile regression framework. First, the quantile model explores the heterogeneous bond return predictability, which enrichs the study of macro-spanning hypothesis. In the literature, \cite{AndreasenEngsted2021} and \cite{Borup2024} have examined the heterogeneity of bond risk premia predictability, taking into account the business cycle and economic uncertainty. However, the heterogeneous predictability of bond risk premia at various quantile levels remains unexplored. Second, quantile model helps predict bond risk premia in different scenarios, e.g., investors often pay more attention to the tail risk of bonds. 
Third, the estimation and inference in mean regression are significantly affected by outliers, whereas the quantile regression is robust to the outliers.
There are notable challenges to conduct the test of spanning hypothesis in a quantile model. First, \cite{Bauer2018} noted that the predictors commonly employed in the literature exhibit high persistence, leading to size distortion in predictive quantile regression test statistics and potentially causing spurious findings. Second, the conventional test statistics suffer size distortion at tail quantiles due to inaccurate density estimator of quantile regression errors. Third, when investigating the macro-spanning hypothesis, it is common for previous empirical studies to utilize four or more predictors. These predictors typically include the CP factor, two LN factors \citep{ludvigson2009macro}, and at least one additional macro predictor. \footnote{See the details of the variables description in Section \ref{sectionData}.} In summary, testing the spanning hypothesis at the quantiles needs a robust marginal inference approach (meaning testing the predictive power of one variable while controlling for others, as opposed to conducting a joint test) in a predictive quantile regression framework that includes a large number of highly persistent predictors.
Unfortunately, existing inference methods suffer size distortion in marginal tests with many highly persistent predictors. In the literature, \citet{Lee2016} extended the IVX method \citep{PhillipsMagdalinos2009} in mean regression to the quantile regression framework (IVX-QR) to test the predictability at various quantile levels, with a focus on joint tests but suffers severe size distortion at tail quantiles due to inaccuracy of nonparametric density estimator. To avoid size distortions at tail quantiles, \citet{FanLee2019} combined the moving block bootstrap technique and IVX-QR approach for the joint test, which circumvents estimating this density. \citet{LiuYangetal2023} developed a unified predictability test for univariate predictive quantile regression with highly persistent predictors. All three methods above 
are unsuitable for the marginal tests in multivariate models.
\citet{CaiChenLiao2023} (hereafter referred to as CCL2023) constructed the instrumental variable (IV) estimator based on a double-weighted method for predictive QR (DW-QR) to conduct the joint and marginal tests in multivariate models. However, the size performance of CCL2023 is still not very satisfactory at tail quantiles and with many predictors.
\subsection{The Testing Procedure and Contributions}
Our new inference procedure fits the task of testing the macro-spanning hypothesis for bond risk premia. First, we construct a consistent IV estimator under both the null and the alternative hypotheses by a two-step quantile regression. Nonetheless, the test statistics constructed by the IV estimator above continue to experience size distortions with many highly persistent predictors, which is induced by two higher-order terms detailed in Section \ref{section3}. Second, we improve the size and power performance of the test. To address the size distortion issue, we adopt the sample splitting method proposed by \citet{liao2024robust} to eliminate one higher-order term. We select a conservative tuning parameter in constructing instrumental variables to reduce the size distortion induced by another higher-order term. To amend the power loss due to the conservative tuning parameter choice, we enhance the power of the aforementioned test by adding a modified test statistic from the conventional test, which converges to zero at the rate $\sqrt{T}$ under the null hypothesis. This way, we fully utilize the good power performance of the traditional test without the bundling of its size distortion effect. Third, when constructing the final test statistic, we introduce a new and accurate density estimator of the error term, which utilizes a simulation-based estimator instead of the inefficient nonparametric estimator. \footnote{See CCL2023 for more details about the procedure using an irrelevant (but suitably picked) auxiliary variable, such that the resulting test do not depend on regressor persistence under the null.
}
Next, we apply the above procedure for the macro-spanning hypothesis by testing the predictive power of the LN factors, TDH, and TI at various quantiles while controlling for the CP factor. The main empirical findings are summarized as follows. \footnote{In our empirical study, the variable has predictive power if its rejection rate is less than 1\%.} On the one hand, the CP factor has significant predictive power, while all macro variables, including the LN factors, have negligible predictive power at center quantiles, which strongly supports the macro-spanning hypothesis. On the other hand, strong empirical evidence against the macro-spanning hypothesis is found in tail quantiles. In particular, TDH has predictive power for bonds of all maturities at right tail quantiles, while it only has predictive power for 2- and 3-year bonds at left tail quantiles. 
Moreover, for bonds of all maturities, TI has significant predictive power at all quantiles except for center quantiles. Potential economic explanations for the findings are discussed in detail in Section \ref{poexpr}.
Our contributions are threefold. 
\begin{enumerate}
\item We find evidence supporting the macro-spanning hypothesis at center quantiles of future bond risk premia and evidence not found in the literature against the macro-spanning hypothesis at tail quantiles. In particular, we examine the tail risk of future bond risk premia and test whether it could be predicted by the CP factor only. We have two new interesting findings on the predictive power of TDH at right tail quantiles and TI at both tails quantiles.
\item We provide a novel and reliable inference approach whose size and power performance are significantly better than the literature for marginal tests in predictive quantile regression with many highly persistent predictors. As a result, our approach is more suitable to conduct macro-spanning hypothesis than those in existing literature.
\item We demonstrate that our in-sample and out-of-sample predictions outperform existing method across various quantiles, which arise from the new discovery of the prediction power of TDH and TI at tail quantiles. 
\end{enumerate}
The rest of this paper is organized as follows. Section \ref{sectionData} describes the data characteristics of bond risk premia and its predictors. Section \ref{sectionEconometric} introduces our inference approach in the predictive quantile regression with highly persistent predictors. Section \ref{section4} presents the heterogeneous predictability of bond risk premia at various quantiles. Section \ref{appinf1} shows the excellent in-sample and out-of-sample performance compared with CCL2023. Section \ref{section7} concludes the paper. The online appendix includes an algorithm for the inference procedure, additional theoretical and numerical results, and an application on the left and the right tail risk indicators. 
Throughout this paper, the standard notations $\Rightarrow$, $\xrightarrow{d}$, $\xrightarrow{p}$ and $\overset{d}{=}$ are used to represent weak convergence, convergence in distribution and in probability, and equivalence in distribution, respectively. All limits are for $T\rightarrow \infty$ in all limit theories, and $O_p (1)$ is asymptotically bounded while $o_p(1)$ is asymptotically negligible. 
\section{Data Characteristics}\label{sectionData}
In this section, we illustrate the variables that are used to test the macro-spanning hypothesis. Following \cite{CochranePiazzesi2005}, we refer to the difference of the yield-to-maturity (YTM) of the n-year and 1-year maturity of U.S. discount bond as the dependent variable, {\bf{bond risk premia}} rx(n), where $n=2,3,4,5$. 
The five {\bf{predictors}} are CP factor, LN factors (LN1 and LN2), and TI, all from previous literature, and trend in debt holding (TDH), which is new in the literature.
The n-year U.S discount bond data rx(n) is available on Center for Research in Security Prices (CRSP) and the macroeconomic factors are from the paper of \cite{ludvigson2009macro}. We compute the CP factor following \cite{CochranePiazzesi2005}. The consumer price index (CPI) and the debt holding data are from the CEIC database. \footnote{See more details in https://fiscaldata.treasury.gov/americas-finance-guide/national-debt/.} The sample used in this study consists of monthly data from January 1980 to December 2022. We refer to quantile levels close to 0.5 as center quantiles and quantile levels much less (greater) than 0.5 as left (right) tail quantiles. \footnote{We give specific definitions of the left (right) tail quantiles in the illustrative Figure \ref{TDHTIRX}.}
Next, we explain the predictors in detail. First, the CP factor \citep{CochranePiazzesi2005} represents the information of the yield curve, which is the linear combination of the forward rate $F_{i,t-1}$, $i=1,2,3,4,5$. Second, the LN factors, LN1 and LN2 constructed by \cite{ludvigson2009macro} are fitted values of $\frac{1}{4} \sum_{n=2}^5 r x(n)_t$ regressing on macroeconomic factors which are the first eight principal components from a large dateset of 132 macroeconomic indicators. Third, the trend inflation (TI) is proposed by \cite{CieslakPovala2015} to test the predictive power of highly persistent expected inflation dynamics on bond excess returns. Specifically, TI at period (t-1) is equal to $(1-\check{w}) \sum_{i=1}^{t-1} \check{w}^i \check{\pi}_{t-i}$, where $\check{\pi}_{t-i}$ is inflation rate in the core CPI and $\check{w}$ is 0.9 in this paper, $\check{w}^i$ is the $i$th power of $\check{w}$. \footnote{We also set $\check{w}$ to other values such as 0.88 and 0.92, and obtain similar simulation and empirical results. To save space, these results are not shown here but are available upon request. }
Besides the above popular predictors in existing literature, we introduce the trend in debt holding (TDH), defined as the trend in U.S. federal debt held by the public and by various government agencies (inter-governmental holdings). TDH represents the rate at which the federal government's borrowed funds increase in order to address the outstanding balance of expenses over time. It contains the information of the demand shift for the national bonds such as treasury bonds, treasury inflation-protected securities (TIPS) and discount bonds. Specifically, we construct TDH at period (t-1) in the same manner as TI, which is equal to $(1-\check{w}) \sum_{i=1}^{t-1} \check{w}^i \check{d}_{t-i}$ and $\check{d}_{t-i}$ is the increment of debt holding at time $t-i$. 
Figure \ref{TDHTIRX} shows that the demand shift for U.S. treasury bonds represented by TDH frequently precedes fluctuations in bond returns, which indicates the potential predictability. The statistical and economic significance of the TDH predictor is further explained in later sections. 
\vspace{-0.35cm}
\begin{figure}[htbp]
\centering
\includegraphics[width=0.48\linewidth]{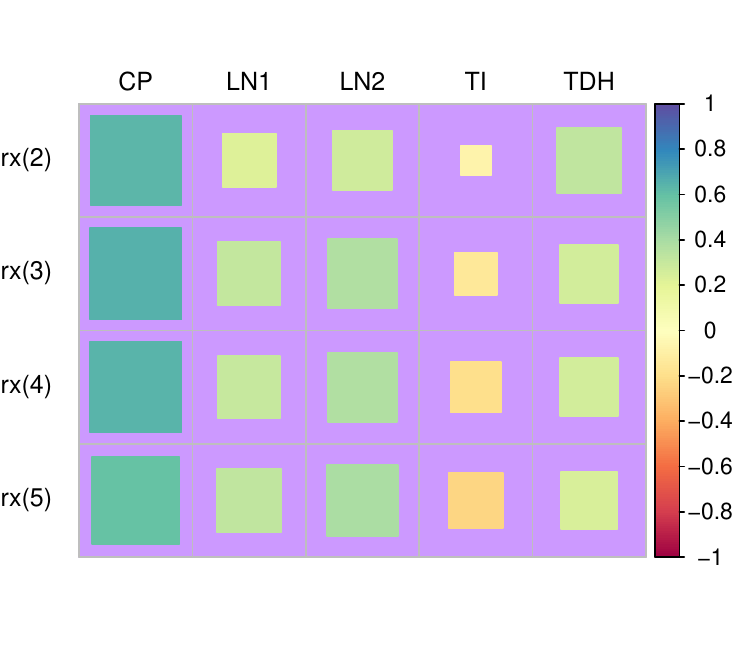}
\vspace{-1cm}
\caption{Correlation Coefficients between rx(n) and One-period Lagged Predictors}
\label{correCoef1}
\vspace{-0.5cm}
\end{figure} 
We first demonstrate the correlation coefficients between bond risk premia rx(n) and one-period lagged predictors in Figure \ref{correCoef1}. The size of the colored square (maximum=1, colors indicate the values of positive/negative correlations) means the absolute value of the correlation coefficient. It shows that the one-period lagged CP factor, which only contains the information of yield curve, has the strongest correlations with bond risk premia while the other macroeconomic predictors have much weaker correlations with bond risk premia. 
The correlation coefficients only reveal linear relationship at the center of the distribution. Next, we demonstrate the relationship between bond risk premia and one-period lagged predictors, specifically TDH and TI, at different quantiles in their time series plot in Figure \ref{TDHTIRX}. For convenience of discussion, we first define the left tail risk period in Figure \ref{TDHTIRX} as the period in which one of bond risk premia rx(n) is less than or equal to its unconditional quantile at 0.05 level, \footnote{This criterion of 0.05 is not essential for our test in the following sections.} for n=2,3,4,5. E.g., Jul. 1980--Sep. 1981 is the left tail risk period by above definition. During periods of left tail risk, the difference between YTM of 1-year and that of n-year bonds (for n=2,3,4,5) is significantly smaller compared to normal conditions. In the extreme instance of the left tail risk period, such as Aug. 2022--Dec. 2022, the inverted yield curve (yields on short-term bonds are higher than those on long-term bonds) occurs. In the left tail quantiles of risk premia for bonds with maturities of two to five years, the relatively high YTM of 1-year bonds suggests that investors are concerned about short-term risks in the bond market. This perception likely stems from current economic indicators or market volatility, which typically leads investors to demand higher yields for assuming additional risk over the near term. Consequently, in such market conditions, investors tend to favor bonds with longer maturities (two to five years) over 1-year bonds. They perceive these longer-term bonds as offering a more attractive balance of risk and return, particularly when short-term forecasts appear unstable. This preference is especially pronounced at the left tail quantiles, where risk sensitivity is heightened. Likewise, we define the right tail risk period as the period in which one of bond risk premia rx(n) is greater than or equal to its unconditional quantile at 0.95 level, for n=2,3,4,5. E.g., Jun. 1982--Dec. 1982, is the right tail risk period. Correspondingly, YTM of n-year bonds (for n=2,3,4,5) is much bigger than that of 1-year bonds in the right tail risk period. 
In the right tail quantiles of YTM of n-year bonds (for n=2,3,4,5), there is an indication that market risks are expected to emerge over the medium term rather than the immediate future. This higher YTM reflects a risk premium that investors demand due to anticipated economic uncertainties or rising interest rates affecting these longer maturities. Consequently, in these scenarios, investors tend to prefer 1-year bonds over those with longer durations (two to five years). This preference is driven by the perceived safety and greater liquidity of shorter-term bonds, especially when facing potential medium-term market fluctuations. 
Some interesting patterns are observable in Figure \ref{TDHTIRX}. First, in the left tail risk periods, e.g., Jul. 1980--Sept. 1981, Dec. 1981--Mar. 1982 and Aug. 2022--Dec. 2022, one-period lagged TI and TDH have an opposite trend of n-year bond risk premia, for n=2,3,4,5. Second, in the right tail risk periods, e.g., Dec. 2001--Sept. 2002 and Oct. 2008--Dec. 2008, one-period lagged TI and TDH have the same ascending trend of n-year bond risk premia, for n=2,3,4,5. This suggests that TDH and TI may have predictive power at tail quantiles. Granted, the correlations of risk premia with predictors shown in Figures \ref{correCoef1} and \ref{TDHTIRX} are insufficient to infer the macro-spanning hypothesis, it motivates us to test the predictability of bond returns using CP factor, LN1, LN2, TDH and TI, not only at center quantiles but also at tail quantiles. 
\begin{figure}[htbp]
\centering
\includegraphics[width=1\linewidth]{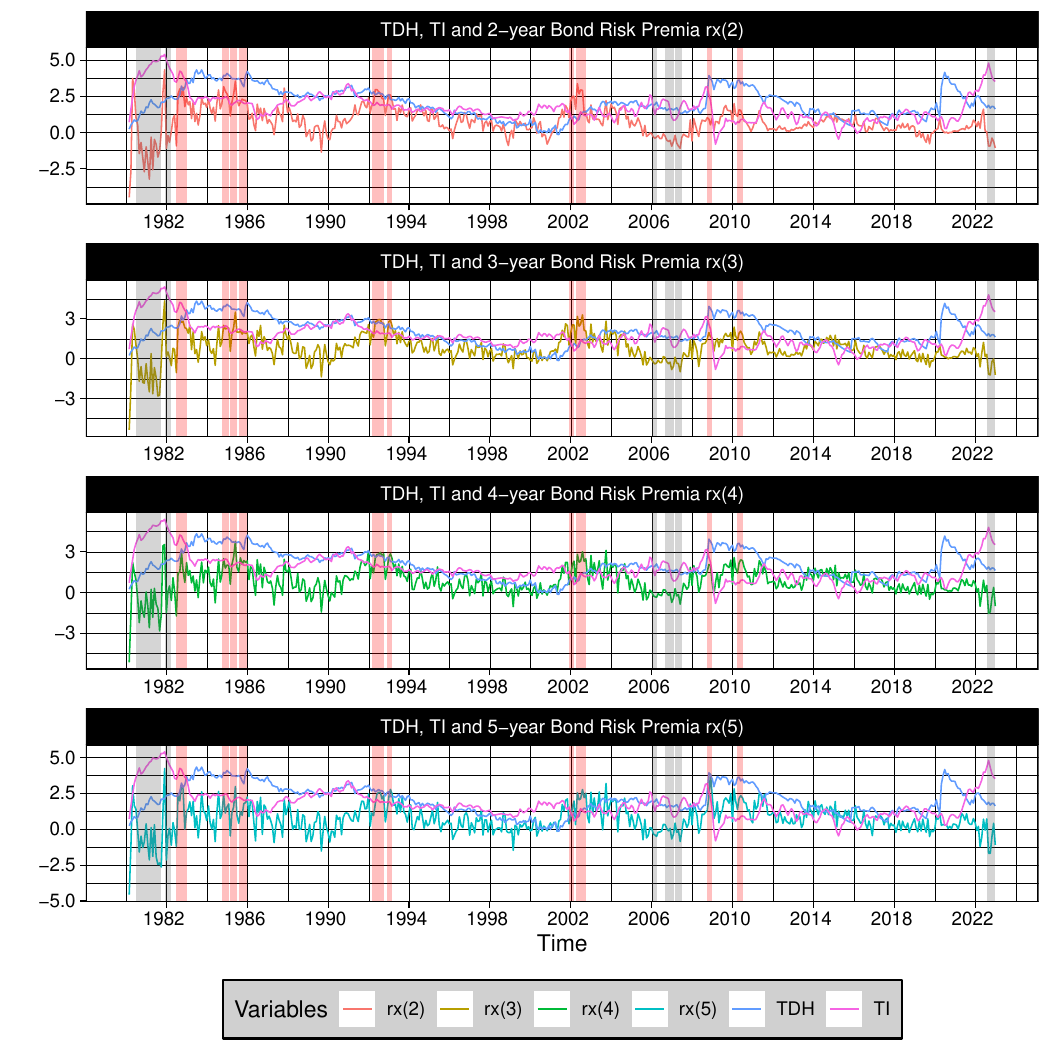}
\caption{TDH, TI and Bond Risk Premia. Gray/Pink area: left/right tail risk period}
\label{TDHTIRX}
\end{figure}
Table \ref{AR1coef} shows that the AR(1) coefficients of CP factor, LN1, LN2, TDH and TI are close to 1, indicating that these predictors are highly persistent. To obtain a reliable test about the macro-spanning hypothesis, which requires marginal test in a multivariate model with many highly persistent predictors, we present the new procedure in the following section. 
\begin{table}[htbp]
\centering
\caption{AR(1) Coefficients of CP, LN1, LN2, TDH and TI}
\begin{tabular}{c|ccccc}
\hline
Predictors & CP & LN1 & LN2 & TDH & TI \\
\hline
AR(1) & 0.974 & 0.881 & 0.937 & 0.994 & 0.994 \\
\hline
\end{tabular}%
\label{AR1coef}%
\end{table}%
\section{Statistical Model}\label{sectionEconometric}
\subsection{Predictive Quantile Regression Model}\label{Model Framework}
Denote the bond risk premia at period t as $y_t$ and its $\tau$th conditional quantile is $Q_{y_{t}}\left(\tau \mid \mathcal{F}_{t-1}\right)$ such that $\operatorname{Pr}\left[y_{t} \leq Q_{y_{t}}\left(\tau \mid \mathcal{F}_{t-1}\right) \mid \mathcal{F}_{t-1}\right]=\tau \in(0,1)$, where $\mathcal{F}_{t-1}$ is the information set available at time $t-1$. For simplicity, a linear conditional quantile of $y_t$ is imposed as follows.
\begin{equation}
\label{eq1}
Q_{y_t}(\tau\,|\mathcal{F}_{t-1})=Q_{y_t}(\tau|x_{t-1})=\mu_{\tau}+ x_{t-1}^\top \beta_{\tau},\; 0<\tau<1, \; t=1,2,\cdots,T,
\end{equation}
where $x_{t-1}=\left(x_{1,t-1},x_{2,t-1} \cdots,x_{K,t-1}\right)^\top$ is the vector of predictors ($K=5$ in this study) including CP factor, LN factors (LN1 and LN2), TDH and TI, representing $\mathcal{F}_{t-1}$. And $\beta_{\tau} = (\beta_{1\tau},\beta_{2\tau},\cdots,\beta_{K\tau})^\top$ is a $K$-dimensional coefficient vector given any $\tau$. The macro-spanning hypothesis essentially means that only the $\beta_{1\tau}$ of the CP factor is nonzero, and the $\beta_{i\tau}$'s of all other macro variables, for $i=2,3,4,5$, are zeroes at all quantile levels $\tau$. We define the quantile measurement error $u_{t\tau}\equiv y_t-Q_{y_t}(\tau|\mathcal{F}_{t-1})$ and the quantile score function $\psi_{\tau} (u_{t\tau})\equiv\tau-1(u_{t\tau}<0)$. 
{Since the AR(1) coefficients of predictors in Table \ref{AR1coef} are close to one (highly persistent), it is common in the literature to model these predictors as follows.} 
\begin{align}\label{mulgtuA1}
x_{i,t} = \rho_i x_{i,t-1} +v_{i,t},
\end{align}
where $x_{i,0}=o_p(\sqrt{T})$, $\rho_i=1+c_i/T^\alpha$,
$t=1,\cdots,T$, $i=1,\cdots,K$, $v_t=(v_{1,t},v_{2,t},\cdots,v_{K,t})^\top$,
$\alpha=0$ or 1, $-2<c_i<0$ and $c=\operatorname{diag}(c_1,c_2,\cdots,c_K)$. We assume no cointegration relationship exist among predictors $x_{t-1}$.
Two types of persistency with different values of $c_i$ and $\alpha$ are considered for theoretical purpose: (1) Strong Dependence [SD]: $\alpha=1$ and $c_i$ is a constant; (2) Weak dependence [WD]: $\alpha=0$ and $|1+c_i|<1$. \footnote{ We also permit $x_{i,t-1}$ with $i=1,2,\cdots,K$ to have mixed types of persistence, that is to say, $\alpha$ could be a function of $i$. In this setting, the theoretical results for the test statistics are almost the same.}
\begin{remark}\label{remaDGP}
DGP of innovations of bond risk premia $y_t$ in (\ref{eq1}) includes the conditional heteroscedasticity case such as GARCH process, which is also described by \cite{LiuYangetal2023}. For instance, $y_t=\mu_{\tau}+ x_{t-1}^\top \beta_{\tau} + u_{t \tau}$, $u_{t \tau}=u_t-\sigma_t Q_\zeta(\tau)$, $u_t=\sigma_t \zeta_t$ and $\sigma_t^2=\mu_\sigma+\sum_{i=1}^q a_i u_{t-i}^2+\sum_{j=1}^r b_j \sigma_{t-j}^2$, 
where $\zeta_t$ is a sequence of independent and identically distributed random vectors with means zero and variances one. $P\left[\zeta_t \leq Q_\zeta(\tau)\right]=\tau$ implies $P\left(u_{t\tau} \leq 0|\mathcal{F}_{t-1}\right)=\tau$, thus $Q_{u_{t\tau}}(\tau\,|\mathcal{F}_{t-1})=0$ and $Q_{y_t}(\tau\,|\mathcal{F}_{t-1}) =\mu_{\tau}+ x_{t-1}^\top \beta_{\tau}$. 
\end{remark}
Following \cite{Lee2016} and CCL2023, we impose the following general weakly dependent structure for innovation $\{v_t\}$ in (\ref{mulgtuA1}).
\begin{assumption}
\label{Assumption A.1}
Assume that $v_t$ follows a linear process given by $ v_{t}=\sum_{j=0}^\infty F_{xj} \varepsilon_{t-j}$, where $\varepsilon_t$ is a martingale difference sequence (MDS) with $E(\varepsilon_t|\mathcal{F}_{t-1})=0$ and $var(\varepsilon_t\varepsilon_t'|\mathcal{F}_{t-1})=\Sigma_{\varepsilon}$ for $\Sigma_{\varepsilon}>0$ and $E\|\varepsilon_t\|^{2+\nu}<\infty$ for some $\nu>0$ and $ E[ \psi_{\tau}(u_{t\tau})\varepsilon_{t}]$ is a constant. Here, $F_{x0}=I_K$, and $\sum_{j=0}^\infty j\|F_{xj}\|<\infty$ and $F_x(1)=\sum_{j=0}^\infty F_{xj}>0$, where $F_x(z)=\sum_{j=0}^\infty F_{xj}z^j$. The covariance matrix of $v_t$ can be expressed as $\Omega_{vv}=\sum_{h=-\infty}^\infty E(v_{t}v_{t-h}^\top)=F_x(1)\Sigma_{\varepsilon} F_x(1)^\top$.
\end{assumption}
\cite{Lee2016} and CCL2023 state the functional central limit theorem (FCLT) for $\{\psi_\tau(u_{t\tau}), v_t\}$ as follows.
\begin{equation}
\label{eq3}
\frac{1}{\sqrt{T}} \sum_{t=1}^{\lfloor rT \rfloor}
\left(\begin{array}{c}
\psi_{\tau}(u_{t\tau})\\
v_{t}
\end{array}
\right)
\Rightarrow
\left(\begin{array}{c}
B_{\psi_\tau}(r)\\
B_v(r)
\end{array}
\right)
=BM \left(\begin{array}{cc}
\tau(1-\tau) & \Sigma_{\psi_\tau v} \\
\Sigma_{\psi_\tau v} & \Omega_{vv}
\end{array}
\right),
\end{equation}
where $[B_{\psi_\tau}(r), B_v(r)]^\top$ is a vector of Brownian motions. Additionally, SD predictors $x_{\lfloor rT \rfloor}/{\sqrt{T}} \Rightarrow J_x^c(r)$ for $0\le r\le 1$, where $J_x^c(r)=\int_0^r e^{(r-s)c}d B_v(s)$ is Ornstein-Uhlenbeck process \citep{Phillips1987}.
Next, we impose some regularity assumptions on the conditional density of $u_{t\tau}$, which are also adopted by \cite{Xiao2009} and CCL2023.
\begin{assumption}\label{Assumption A.2}
(i) The sequence of conditional stationary probability density functions $\{f_{u_{t\tau},t-1}(\cdot)\}$ of $\{u_{t\tau}\}$ given $\mathcal{F}_{t-1}$ evaluated at zero satisfies a moment condition with a non-degenerate mean $f_{u_{\tau}}(0)=E(f_{u_{t\tau},t-1}(0))>0$ and $E(f_{u_{t\tau},t-1}^\vartheta(0))<\infty$ for some $\vartheta>1$. \\
(ii) For each $t$ and $\tau\in(0,1)$, $f'_{u_{t\tau},t-1}(x)$ is bounded with probability one around zero, i.e., $f'_{u_{t\tau},t-1}(\epsilon)<\infty$ and $f_{u_{t\tau},t-1}(\epsilon)<\infty$ almost surely for all $|\epsilon|<\eta$ for some $\eta>0$.
\end{assumption}
The conventional test statistics tend to find spurious predictability with SD predictors and non-zero contemporary correlation between the error term of bond risk premia and predictors. This occurs because their asymptotic distribution contains nuisance parameters that cannot be estimated consistently, which is shown in \cite{Lee2016} and CCL2023. 
\subsection{The Inference Procedure and Theoretical Results}\label{section3}
We improve IVX-QR \citep{Lee2016} to reduce the size distortions arising from the following sources: 1, the inconsistency of IV estimator under alternative hypothesis; 2, density function estimation for the measurement error; 3, the higher-order terms in the test statistic causing the bias. In the following subsection \ref{subsection3.1}, we develop the consistent estimator. In subsection \ref{subsection3.2}, we address the other two issues of size distortions.
\subsubsection{A Two-step Regression for IV Estimator Construction}\label{subsection3.1}
To offer an IV estimator that is consistent under both the null and the alternative hypothesis, we extend the two-step mean regression to quantile regression. We follow \cite{Lee2016} and define the instrumental variable $z_t=(z_{1,t},z_{2,t}\cdots,z_{K,t})^\top$ as follows,
\begin{align}\label{mulivz}
z_{i,t} = \rho_z z_{i,t-1} + \Delta x_{i,t}, \quad i=1,2,\cdots,K,\quad 
t=1,2,\cdots,T,
\end{align}
where $z_{i,0}= x_{i,0} - x_{i,-1}$, $\Delta x_{i,t} = x_{i,t} - x_{i,t-1}$, $\rho_z = 1+ c_z/T^\delta $, $c_z<0$ and $1/2<\delta<1$ (here we set $\delta=0.95$). Next,
the new two-step regression is conducted as follows.
\begin{description}
\item[\textcolor{blue}{Step 1:}] Run the following OLS regression.
\begin{align}\label{firstep}
(\hat{\mu}_x ,\hat{\theta} ) = \arg \; \min_{\mu_x,\theta} \sum_{t=1}^T \left(x_{t-1}- \mu_x - \theta z_{t-1} \right)^\top \left(x_{t-1}- \mu_x - \theta z_{t-1} \right).
\end{align}
By this approach, we decompose predictors $x_{t-1}$ into two orthogonal parts: the fitted value of (\ref{firstep}) $\tilde{x}_{t-1}= \hat{\mu}_x+\hat{\theta} z_{t-1}$, and its residual $\tilde{v}_{t-1} = x_{t-1}-\tilde{x}_{t-1}$.
\item[\textcolor{blue}{Step 2:}] Run the following quantile regression.
\begin{align}\label{firstep2}
\left(\hat\mu_\tau,\hat\beta_\tau^\top,\hat\gamma_\tau^\top\right)^\top =\arg \,\min_{\mu_\tau, \beta_\tau,\gamma_\tau} \sum_{t=1}^T \rho_{\tau}\left(y_{t}-\mu_\tau-\beta_\tau^\top \tilde{x}_{t-1} -\gamma_\tau^\top \tilde{v}_{t-1}\right),
\end{align}
where $\rho_\tau(u)=u[\tau-1(u<0)]$ is referred to as check function in the literature.
\end{description}
The intuition of the above two-step IV estimator is the following. First,
by the equations $x_{t-1}= \tilde{x}_{t-1}+\tilde{v}_{t-1} $ and (\ref{eq1}), it follows that
\begin{align}\label{eqjg3}
Q_{y_t}(\tau\,|\mathcal{F}_{t-1}) ={\mu}_\tau + \tilde{x}_{t-1}^\top \beta_{\tau}+ \tilde{v}_{t-1}^\top \beta_{\tau}.
\end{align}
{It is clear that both $\hat{\beta}_\tau$ and $\hat{\gamma}_\tau$ in (\ref{firstep2}) are consistent estimators of $\beta_\tau$ under both the null and the alternative hypothesis, since they are the estimates of the same true $\beta_{\tau}$ in \eqref{eqjg3}}. Second, the OLS property guarantees that $ \tilde{v}_{t-1}$ is orthogonal to $\tilde{x}_{t-1}$ and the intercept term. The distribution of $\hat{\beta}_\tau$ dependents on $\tilde{x}_{t-1}$ rather than $\tilde{v}_{t-1}$, and thus, it only dependents on $z_{t-1}$ since $\tilde{x}_{t-1}= \hat{\mu}_x+\hat{\theta} z_{t-1}$. Since $z_{t-1}$ is mildly integrated with SD predictors, the asymptotic mixture normal distribution of $\hat{\beta}_\tau$ is guaranteed. 
To obtain the asymptotic distribution of $\hat{\beta}_\tau$, we first establish the Bahadur representation in the following theorem. 
\begin{prop}(Bahadur Representation)
\label{thm1}
Under Assumptions \ref{Assumption A.1} and \ref{Assumption A.2}, it follows that
\begin{small}
\begin{align*}
\left[
\begin{array}{c}
\sqrt{T}\left(\hat\mu_\tau-{\mu}_{\tau} \right) \\
D_T\left(\hat{\beta}_\tau -\beta_\tau\right)\\
\dot{D}_T\left(\hat\gamma_\tau -\beta_\tau \right)
\end{array}
\right] =\frac{1}{f_{u_\tau}(0)}
\left[
\begin{array}{ccc}
1 & \frac{\sum\limits_{t=1}^T \tilde{x}_{t-1}^\top}{\sqrt{T}D_T} & \frac{\sum\limits_{t=1}^T \tilde{v}_{t-1}^\top}{\sqrt{T}\dot{D}_T} \\
\frac{\sum\limits_{t=1}^T \tilde{x}_{t-1}}{\sqrt{T}D_T} & \frac{\sum\limits_{t=1}^T \tilde{x}_{t-1} \tilde{x}_{t-1}^\top }{ D_T^2 } & \frac{\sum\limits_{t=1}^T \tilde{x}_{t-1} \tilde{v}_{t-1}^\top}{ D_T \dot{D}_T } \\
\frac{\sum\limits_{t=1}^T \tilde{v}_{t-1}}{\sqrt{T}\dot{D}_T} & \frac{\sum\limits_{t=1}^T \tilde{v}_{t-1} \tilde{x}_{t-1}^\top }{ D_T \dot{D}_T }& \frac{\sum\limits_{t=1}^T \tilde{v}_{t-1} \tilde{v}_{t-1}^\top}{ \dot{D}_T^2 }
\end{array}
\right]^{-1}
\left[
\begin{array}{c}
\sum\limits_{t=1}^T \frac{\psi_\tau (u_{t\tau})}{\sqrt{T}}\\
\sum\limits_{t=1}^T \frac{\tilde{x}_{t-1}\psi_\tau (u_{t\tau})}{D_T}\\
\sum\limits_{t=1}^T \frac{\tilde{v}_{t-1}\psi_\tau (u_{t\tau})}{\dot{D}_T}
\end{array}
\right] +o_p(1).
\end{align*}
\end{small}
where $D_T= T^{(1+\delta)/2}$ for SD predictors and $D_T= \sqrt{T}$ for WD predictors, $\dot{D}_T= T$ for SD predictors and $D_T= \sqrt{T}$ for WD predictors, and $f_{u_\tau}(0)$ is defined in Assumption \ref{Assumption A.2}.
\end{prop}
Since the first step decomposes predictors $x_{t-1}$ to orthogonal components, i.e., $\sum_{t=1}^T \tilde{x}_{t-1}\tilde{v}_{t-1}^\top =0$ and $\sum_{t=1}^T \tilde{v}_{t-1} =0$, the following theorem holds by Proposition \ref{thm1} and the asymptotic property of instrumental variable $z_{t-1}$ \citep{Kostakisetal2015}.
\begin{thm}
\label{thm2}
Under Assumptions \ref{Assumption A.1} and \ref{Assumption A.2}, the following result holds.
\begin{align}
D_T\left(\hat{\beta}_\tau -\beta_\tau\right)&=\frac{1}{f_{u_\tau}(0)} \left( D_T^{-2} \sum\limits_{t=1}^T \overline{\tilde{x}}_{t-1} \overline{\tilde{x}}_{t-1}^\top \right)^{-1} D_T^{-1} \sum\limits_{t=1}^T \overline{\tilde{x}}_{t-1} \psi_\tau (u_{t\tau}) + o_p(1) \nonumber \\
\label{thm3eg}
&=\frac{1}{f_{u_\tau}(0)} \left( D_T^{-2} \sum\limits_{t=1}^T \bar{z}_{t-1} x_{t-1}^\top \right)^{-1} D_T^{-1} \sum\limits_{t=1}^T \bar{z}_{t-1} \psi_\tau (u_{t\tau}) + o_p(1)\\
& \xrightarrow{d}
\begin{cases}
\operatorname{MN}\left[0_K,\frac{1}{f_{u_\tau}(0)^2 } \Omega_{zx}^{-1}\Omega_{zz}\left( \Omega_{zx}^{-1}\right)^\top\right],\quad SD;\\
\operatorname{N}\left[0_K,\frac{1}{f_{u_\tau}(0)^2 } \Omega_{zx}^{-1}\Omega_{zz}\left( \Omega_{zx}^{-1}\right)^\top\right],\quad WD;
\end{cases}. \nonumber
\end{align}
where $\Omega_{zz}= \tau(1-\tau) \Omega_{vv}/(-2c_z)$ and $\Omega_{zx}=
-c_z^{-1} \operatorname{E}(v_tv_t^\top)-c_z^{-1} \int_0^1 dJ_x^c(r)\, J_x^c(r)\, ^\top$ for SD predictors and $\Omega_{zz}= \tau(1-\tau) \operatorname{E}(x_{t-1}x_{t-1}^\top )$ and $\Omega_{zx}= \operatorname{E}\left( x_{t-1}x_{t-1}^\top\right)$ for WD predictors. And $\bar{z}_{t-1}=z_{t-1}- \frac{1}{T}\sum\limits_{t=1}^{T} z_{t-1}$ and $\overline{\tilde{x}}_{t-1}=\tilde{x}_{t-1}- \frac{1}{T}\sum\limits_{t=1}^{T} \tilde{x}_{t-1}$ and $0_K$ is a $K$-dimensional zero vector.
\end{thm}
Equation (\ref{thm3eg}) in Theorem \ref{thm2} reveals that $\hat{\beta}_\tau$ constructed by the new two-step regression is the instrumental variable estimator in the quantile regression.
It is straightforward to construct the Wald type test statistic $Q_{ivx-qr}$ by self-normalization for the null hypothesis $H_0:R\beta_\tau =r_\tau$ where $R$ is a $J\times K$ predetermined matrix with rank $J$, $r_\tau$ is a predetermined vector with dimension $K$.
\begin{align}\label{dgjh3e1}
Q_{ivx-qr} = \left(R\hat{\beta}_\tau- r_\tau \right)^\top \left[ R \operatorname{\widehat{Avar}}(\hat{\beta}_\tau) R^\top \right]^{-1}\left(R\hat{\beta}_\tau- r_\tau \right),
\end{align}
where $\operatorname{\widehat{Avar}}( \hat{\beta}_\tau)=\frac{1}{\tilde{f}_{u_\tau}(0)^2 } \hat{\Omega}_{zx}^{-1}\hat{\Omega}_{zz}\left( \hat{\Omega}_{zx}^{-1}\right)^\top$, $\hat{\Omega}_{zx}=\sum_{t=1}^T \bar{z}_{t-1} x_{t-1}^\top$ and $\hat{\Omega}_{zz}=\tau(1-\tau) \sum_{t=1}^T \bar{z}_{t-1} \\ \bar{z}_{t-1}^\top$. $\tilde{f}_{u_\tau}(0)$ is some consistent estimator of $f_{u_\tau}(0)$, e.g., the one obtained by the nonparametric method of \cite{Lee2016}. Moreover, we define the t-test statistic for the \emph{one-sided} marginal test with $J=1$ such as $H_0:\beta_{i\tau}=0$, for $i=1,2,...K$.
\begin{align}\label{dgjh3e2}
\check{Q}_{ivx-qr} = \frac{ R\hat{\beta}_\tau- r_\tau }{\left[R \operatorname{\widehat{Avar}}(\hat{\beta}_\tau) R^\top \right]^{1/2} }.
\end{align}
\begin{prop}\label{thm3}
Under Assumptions \ref{Assumption A.1} and \ref{Assumption A.2} and the null hypothesis $H_0:R\beta_\tau=r_\tau$, one can show that the limiting distribution of the t-test statistic $\check{Q}_{ivx-qr}$ with $J=1$ and those of Wald type test statistic $Q_{ivx-qr}$ are the standard normal distribution and the $\chi^2$-distribution with $J$ degrees of freedom, respectively.
\end{prop}
Although the asymptotic distributions of $Q_{ivx-qr}$ and $\check{Q}_{ivx-qr}$ are known, they still suffer from size distortions in finite sample due to the following two reasons. First, both the size distortions of $Q_{ivx-qr}$ and $\check{Q}_{ivx-qr}$ arise due to two higher-order terms $B_T$ and $C_T$ similar to those in mean regression models. \footnote{For more details on the higher-order terms $B_T$ and $C_T$, refer to Proposition \ref{mulpropfdie3} in the subsequent section and the mean regression model counterpart discussed in \cite{liao2024robust}.} Second, the nonparametric estimator $\tilde{f}_{u_\tau}(0)$ is inaccurate at tail quantiles and when the number of macro variables is large. This inaccuracy causes both $Q_{ivx-qr}$ and $\check{Q}_{ivx-qr}$ to suffer from size distortion in these cases.
\subsubsection{Reduce Size Distortions}\label{subsection3.2}
We now address the size distortion issues of the preliminary test statistic $Q_{ivx-qr}$ and $\check{Q}_{ivx-qr}$ in the following aspects. First, we employ a new simulation-based estimator for the density function of quantile regression error $f_{u_\tau}(0)$ to reduce the size distortion at tail quantiles and with a large number of predictors. Second, we apply a sample splitting procedure introduced by \cite{liao2024robust} to eliminate one of the higher-order terms {$C_T$}. Third, we choose a conservative value for the tuning parameter $c_z$ to reduce the size distortion effect of another higher-order term $B_T$ and point out its negative effect on its power performance. The details are shown below.
\subsubsection*{A New Simulation-based Estimator for $f_{u_\tau}(0)$}
The simulation-based approach to estimate $f_{u_\tau} (0)$ is as follows. We repeat the following two steps for $i=1,2,\cdots,M_1$, where $M_1 \in \mathbb{N}^{+}$ is a predetermined positive integer.
\begin{description}
\item[\textcolor{blue}{Step 1:}] We randomly generate i.i.d. multivariate time series $\{\xi_t^{(i)}\}_{t=1}^T$, which follows $\operatorname{N}\left( \operatorname{0_{M_2}}, \operatorname{I_{M_2}} \right)$ and $M_2\in \mathbb{N}^{+}$.
\item[\textcolor{blue}{Step 2:}] By (\ref{eq1}), it follows that
$Q_{y_t}(\tau\,|\mathcal{F}_{t-1}) =\mu_{\tau}+ x_{t-1}^\top \beta_{\tau} + [ \xi_{t-1}^{(i)} ]^\top \operatorname{0_{M_2}}$. 
So we are able to estimate the following quantile regression,
\begin{align}\label{quasiboot}
\left[\hat{\mu}_\tau^s,(\hat{\beta}_\tau^s)^\top,(\hat{l}_\tau^{(i)})^\top \right]^\top =\arg \min_{\mu_\tau, \beta_\tau,l_\tau} \sum_{t=1}^T \rho_{\tau}\left(y_{t}-\mu_\tau-\beta_\tau^\top x_{t-1} - l_\tau^\top \xi_{t-1}^{(i)}\right).
\end{align}
where $\hat{l}_\tau^{(i)}$ is the estimator of $\operatorname{0_{M_2}}$.
\end{description}
Since $\xi_{t-1}^{(i)}$ is randomly generated, it is independent of the main sample, then the following asymptotic distribution of $\xi_{t-1}^{(i)}$ holds. As the sample size $T\rightarrow\infty$,
\begin{align}\label{kdj38f3}
\sqrt{T}\hat{l}_{j,\tau}^{(i)}
&=\frac{1}{f_{u_\tau}(0)} \left\{ \frac{1}{T}\sum\limits_{t=1}^T [\xi_{j,t-1}^{(i)}]^2 \right\}^{-1} \frac{1}{\sqrt{T}} \sum\limits_{t=1}^T \xi_{j,t-1}^{(i)} \psi_\tau (u_{t\tau}) + o_p(1) \\
&=\frac{1}{f_{u_\tau}(0)} \frac{1}{\sqrt{T}} \sum\limits_{t=1}^T \xi_{j,t-1}^{(i)} \psi_\tau (u_{t\tau}) + o_p(1). \nonumber
\end{align}
Equation (\ref{kdj38f3}) and the fact that $\overrightarrow{\xi_{t-1}}\equiv \big(\xi_{1,t-1}^{(1)},\xi_{2,t-1}^{(1)}, \cdots,\xi_{M_2,t-1}^{(1)},\xi_{1,t-1}^{(2)},\xi_{2,t-1}^{(2)}, \cdots,\xi_{M_2,t-1}^{(2)},\\ \cdots\cdots, \xi_{1,t-1}^{(M_1)},\xi_{2,t-1}^{(M_1)}, \cdots,\xi_{M_2,t-1}^{(M_1)}\big)^\top$ follows i.i.d. $\operatorname{N}(0_{M},I_M)$ imply that
\begin{align}\label{dj74e5}
\frac{1}{M} \sum_{i=1}^{M_1}\sum_{j=1}^{M_2} \left[\sqrt{T}\hat{l}_{j,\tau}^{(i)}\right]^2 \bigg| \mathcal{F}_T & \xrightarrow{P} \operatorname{Avar}\left[\hat{l}_{1,\tau}^{(1)}\right]=\frac{1}{f_{u_\tau}(0)^2 } \tau(1-\tau),\\
\frac{1}{M} \sum_{i=1}^{M_1}\sum_{j=1}^{M_2} \left[\sqrt{T}\hat{l}_{j,\tau}^{(i)}\right]^2 & \xrightarrow{P} \operatorname{Avar}\left[\hat{l}_{1,\tau}^{(1)}\right]=\frac{1}{f_{u_\tau}(0)^2 } \tau(1-\tau), \nonumber 
\end{align}
as $M_1\rightarrow \infty$ and $T\rightarrow \infty$ and $M_2$ is constant.
Therefore, the simulation-based estimator of $f_{u_\tau} (0)$ is defined as follows 
\begin{align}\label{dj74e6}
\hat{f}_{u_\tau}(0) \coloneqq \left\{\frac{1}{M} \sum_{i=1}^{M_1}\sum_{j=1}^{M_2} \left[\sqrt{T}\hat{l}_{j,\tau}^{(i)}\right]^2 \right\}^{-1/2} [\tau(1-\tau)]^{1/2} \xrightarrow{P} f_{u_\tau}(0),
\end{align}
as $M_1\rightarrow \infty$ and $T\rightarrow \infty$ and $M_2$ is constant.
Note that (\ref{dj74e6}) holds by the continuous mapping theorem and (\ref{dj74e5}).
In practice, setting $M_2$ to a relatively large number, such as $M_2=50$, is advisable as it allows for a smaller number of loops ($M_1$) to construct $\hat{f}_{u_\tau}(0)$. For example, we set $M_1=100$, which guarantees speedy simulations to accommodate computation efficiency.
{Equations (\ref{dj74e5}) and (\ref{dj74e6}) imply that, unlike the nonparametric estimators, the accuracy of the estimator $\hat{f}_{u_\tau}(0)$ is not affected by the number of predictors and the quantile level $\tau$, since $\hat{f}_{u_\tau}(0)$ does not require the estimator of $u_{t\tau}$.} Additionally, $\operatorname{Var}\Big\{\frac{1}{M} \sum_{i=1}^{M_1}\sum_{j=1}^{M_2} \big[\sqrt{T}\hat{l}_{j,\tau}^{(i)}\big]^2 \Big| \mathcal{F}_T \Big\} = \frac{2}{M} [\tau(1-\tau)]^2 f_{u_\tau}(0)^{-4} +o_p(1)$ shown in equation (\ref{kuy33}) of the appendix imply that the accuracy of $\hat{f}_{u_\tau}(0)$ could be improved by enlarging M. Therefore, applying $\hat{f}_{u_\tau}(0)$ in constructing the test statistics could improve the size performance at tail quantiles and with a large number of predictors.
\subsubsection*{Higher-order Terms}
Another source of size distortion of $Q_{ivx-qr}$ and $\check{Q}_{ivx-qr}$ is the two higher-order terms, $C_T$ and $B_T$, shown in the following Proposition \ref{mulpropfdie3}. 
To demonstrate the concept of higher-order terms, we present them using the univariate case of $\check{Q}_{ivx-qr}$. \footnote{Higher-order terms in the multivariate model closely resemble Proposition 1 by \cite{liao2024robust}, and thus their detailed derivations are omitted for simplicity.} A similar conclusion can be drawn for $Q_{ivx-qr}$ since $Q_{ivx-qr}=\check{Q}_{ivx-qr}^2$ with $J=1$. 
\begin{prop}\label{mulpropfdie3}
Under Assumptions \ref{Assumption A.1} and \ref{Assumption A.2} and the null hypothesis $H_0:\beta_\tau =0$ and $J=K=1$, we have the following results for SD predictors.
\begin{align}
& \check{Q}_{ivx-qr} ={Z_T}+{B_T}+{C_T}+ o_p\left[T^{(\delta- 1)/2}\right]; \nonumber\\
& Z_T = \Omega_{zz}^{-1/2}\frac{1}{T^{1 / 2+\delta / 2}} \sum_{t=1}^T z_{t-1} \psi_\tau (u_{t\tau}) \xrightarrow{P} \operatorname{N}(0,1 );\nonumber \\
\label{dkmul76gh2}
&B_T = \varpi_b \Omega_{zz}^{-1/2}
\frac{1}{T^{1 / 2+\delta / 2}} \sum_{t=1}^T z_{t-1} \psi_\tau (u_{t\tau})\xrightarrow{P} 0;\\
\label{dkmul76gh}
& T^{(1-\delta)/2}\operatorname{E}\left({B_T} \right) \rightarrow
- \rho_{ v\psi} / \sqrt{-2 c_z}; 
\end{align}
\begin{align} 
\label{multpop1th3}
T^{(1 -\delta) / 2}{{C_T}} &=\left[ \tau(1-\tau)\frac{1}{T^{1+\delta}} \sum_{t=1}^T z_{t-1} z_{t-1}^\top \right]^{-1/2}\frac{1}{T^{1 / 2+\delta }} \sum_{t=1}^T z_{t-1} \frac{1}{\sqrt{T}}\sum_{t=1}^T \psi_\tau (u_{t\tau}) \\
&\Rightarrow -(-c_z/2)^{-1/2} \Sigma_{vv}^{-1/2} J_x^c(1)B_{\psi_\tau}(1)[\tau(1-\tau)]^{-1/2} , \nonumber
\end{align}
where $ \rho_{ v \psi} = \Sigma_{vv}^{-1/2} \Sigma_{\psi_\tau v}/ \sqrt{\tau(1-\tau)}$ is the correlation coefficient between $\psi_\tau (u_{t\tau})$ and $v_t$ and $\varpi_b = -\frac{1}{2} \left[\tau(1-\tau)\Omega_{zz}^{-1/2}\frac{1}{T^{1+\delta }} \sum_{t=1}^T z_{t-1} z_{t-1}^\top \Omega_{zz}^{-1/2} - 1\right]$.
\end{prop}
The slow convergence of $\frac{1}{T}\sum_{t=1}^T z_{t-1}$ in $\check{Q}_{ivx-qr}$ induces $C_T$, while the correlation between the numerator and the denominator of $\check{Q}_{ivx-qr}$ leads to $B_T$.
To eliminate $C_T$,
we follow \cite{liao2024robust} to define a new instrumental variable $\tilde{z}_{t-1}$ based on IV $z_{t-1}$ in (\ref{mulivz}). Specifically, define $\tilde{z}_{t-1} = (\operatorname{I_K}-S_a)z_{t-1}$ when
$1\leq t\leq T_0=\lambda\, T$ \footnote{Following \cite{liao2024robust}, we set $\lambda=0.5$ to evenly split the full sample in half.} and $\tilde{z}_{t-1} = (\operatorname{I_K}-S_b)z_{t-1}$ when $T_0+1 \leq t \leq T$, where
$S_a = \frac{1}{T}\sum\nolimits_{t=1}^{T} z_{t-1}\big(\frac{1}{T_0}\sum\nolimits_{t=1}^{T_0} z_{t-1}^\top \big) \big(\frac{1}{T_0}\sum\nolimits_{t=1}^{T_0} z_{t-1}^\top \frac{1}{T_0} \sum\nolimits_{t=1}^{T_0} z_{t-1}\big)^{-1}$ and 
$S_b = \frac{1}{T}\sum\nolimits_{t=1}^{T} z_{t-1}\big(\frac{1}{T-T_0}\sum\nolimits_{t=T_0+1}^{T} z_{t-1}^\top \big) \big(\frac{1}{T-T_0}\sum\nolimits_{t=T_0+1}^{T} z_{t-1}^\top \frac{1}{T-T_0}\sum\nolimits_{t=T_0+1}^{T} z_{t-1}\big)^{-1}$. 
Since the mean of the IV $\tilde{z}_{t-1}$ is exactly zero, 
the source of $C_T$ is eliminated. Hereafter, we apply the IV $\tilde{z}_{t-1}$ to run the two-step regression introduced in (\ref{firstep}) and (\ref{firstep2}) to obtain the consistent IV estimator of $\beta_\tau$ under both the null and the alternative hypothesis. 
Specially, we run regression $ (\hat{\mu}_x^l,\hat{\theta}^l) = \arg \; \min_{\mu_x,\theta} \sum_{t= 1}^T \left(x_{t-1}- \mu_x - \theta \tilde{z}_{t-1} \right)^\top \left(x_{t-1}- \mu_x - \theta \tilde{z}_{t-1} \right)$ to obtain the fitted value $\tilde{x}_{t-1}^l = \hat{\mu}_x^l+\hat{\theta}^l \tilde{z}_{t-1}$ and the residual $\tilde{v}_{t-1}^l = x_{t-1}- \tilde{x}_{t-1}^l$ in the first step, and $ \big[\hat\mu_\tau^l,(\hat\beta_\tau^l)^\top,(\hat\gamma_\tau^l)^\top\big]^\top =\arg \min_{\mu_\tau, \beta_\tau,\gamma_\tau} \sum_{t= 1}^T \rho_{\tau}\left(y_{t}-\mu_\tau-\beta_\tau^\top \tilde{x}_{t-1}^l -\gamma_\tau^\top \tilde{v}_{t-1}^l\right)$ in the second step. Notice $\hat{\beta}_\tau^l$ uses the full sample and will not lose finite sample efficiency.
The following theorem gives the asymptotic distribution of $\hat{\beta}_\tau^l$.
\begin{thm}\label{multh1m}
Under Assumptions \ref{Assumption A.1} and \ref{Assumption A.2}, for SD and WD predictors, it follows that
\begin{align}
D_T\left(\hat{\beta}_\tau^l - \beta\right) &=D_T\left(\hat{\beta}_\tau^{l_0}- \beta\right) + o_p(1) \\
&=\frac{1}{f_{u_\tau}(0)} \left(D_T^{-1} \sum\limits_{t=1}^{T} \tilde{z}_{t-1}x_{t-1}^\top D_T^{-1} \right)^{-1} D_T^{-1} \sum\limits_{t=1}^{T} \tilde{z}_{t-1} \psi_\tau (u_{t\tau})+ o_p(1). \nonumber\\
& \xrightarrow{d} Z_{\beta}\overset{d}{=} \operatorname{MN}\left[0,\operatorname{Avar}(\hat{\beta}_\tau^l) \right] \nonumber
\end{align}
where $\operatorname{Avar}(\hat{\beta}_\tau^l)=\frac{1}{f_{u_\tau}(0)^2}(\Sigma_{zx})^{-1}\Sigma_{zz} \left[\Sigma_{zx}^{-1}\right]^\top$,
$\Sigma_{zz} = \lambda(\operatorname{I_K} - \tilde{S}_a)\Omega_{zz} (\operatorname{I_K} - \tilde{S}_a)^\top +(1-\lambda)(\operatorname{I_K}-\tilde{S}_b)\\
\Omega_{zz} (\operatorname{I_K} - \tilde{S}_b)^\top$, 
$\Sigma_{zx} = -c_z^{-1}(\operatorname{I_K} -\tilde{S}_a) \int_0^{\lambda} dJ_x^c(r)\, J_x^c(r)^\top -c_z^{-1}(\operatorname{I_K} - \tilde{S}_b)\int_{\lambda}^{1} dJ_x^c(r) \, J_x^c(r)^\top -c_z^{-1}[ \operatorname{I_K}\\
- \lambda\tilde{S}_a -(1-\lambda)\tilde{S}_b ]\operatorname{E}(v_tv_t^\top) $
for SD predictors and
$\Sigma_{zx}=\big[\operatorname{I_K} - \lambda\tilde{S}_a -(1-\lambda)\tilde{S}_b \big] \operatorname{E}\left( x_{t-1}x_{t-1}^\top\right)$ for WD predictors. $S_a \Rightarrow \tilde{S}_a $ and $S_b \Rightarrow \tilde{S}_b $, where $\tilde{S}_a = J_x^c(1) J_x^c(\lambda)^\top 
\big[ J_x^c(\lambda)^\top J_x^c(\lambda) \big]^{-1}$ and 
$\tilde{S}_b=J_x^c(1) \left[J_x^c(1) - J_x^c(\lambda) \right]^\top \big\{ \left[J_x^c(1) - J_x^c(\lambda) \right]^\top \left[J_x^c(1) - \\
J_x^c(\lambda) \right] \big\}^{-1}$ for SD predictors and $\tilde{S}_a = B_v(1) B_v(\lambda)^\top \left[ B_v(\lambda)^\top B_v(\lambda) \right]^{-1}$ and $\tilde{S}_b= B_v(1) \left[B_v(1) - B_v(\lambda) \right]^\top 
\{ \left[B_v(1) - B_v(\lambda) \right]^\top \big[B_v(1) \\ - 
B_v(\lambda) \big] \}^{-1}$ for WD predictors. $B_v(r)$ is defined in equation \eqref{eq3}. 
\end{thm}
Then the test statistic ${Q_{l-qr}}$ is constructed for the null hypothesis $H_0:R\beta_\tau=r_\tau$,
\begin{align}\label{walde2tes}
{Q_{l-qr}} \equiv \left(R \hat{\beta}_\tau^l -r_\tau \right)^\top \left[R \operatorname{\widehat{Avar}}(\hat{\beta}_\tau^l ) R^\top \right]^{-1} \left(R \hat{\beta}_\tau^l -r_\tau \right),
\end{align}
where $\operatorname{\widehat{Avar}}(\hat{\beta}_\tau^l )= \hat{f}_{u_\tau}(0)^{-2} ( \sum_{t=1}^T \tilde{z}_{t-1} x_{t-1}^\top )^{-1}[ \tau(1-\tau) \sum_{t=1}^T \tilde{z}_{t-1} \tilde{z}_{t-1}^\top ] ( \sum_{t=1}^T x_{t-1} \tilde{z}_{t-1}^\top )^{-1}$.
Moreover, we construct the t-test statistic ${\check{Q}_{l-qr}}$ when $J=1$ for right-sided test $H_0:\beta_i=0$ vs $H_a^r:\beta_i>0$ and left-sided test $H_0:\beta_i=0$ vs $H_a^l:\beta_i<0$.
\begin{align}\label{defqc1}
{\check{Q}_{l-qr}} \equiv \frac{ R \hat{\beta}_\tau^l -r_\tau }{ \sqrt{ R\operatorname{\widehat{Avar}}(\hat{\beta}_\tau^l )R^\top} }.
\end{align}
The following theorem gives the asymptotic distribution of $\check{Q}_{l-qr}$. 
\begin{prop}\label{t7sf8f1}
Under Assumptions \ref{Assumption A.1} and \ref{Assumption A.2} and the null hypothesis $H_0:R\beta_\tau=r_\tau$, as $M_1\rightarrow\infty$ and $T\rightarrow \infty$, one can show that the limiting distribution of the test statistics ${\check{Q}_{l-qr}}$ and ${Q_{l-qr}}$ are the standard normal and the $\chi^2$-distribution with $J$ degrees of freedom, respectively.
\end{prop}
Next, the local power of the test statistics ${\check{Q}_{l-qr}}$ and ${Q_{l-qr}}$ are shown as follows. 
\begin{prop}\label{mulkeythe2}
Under Assumptions \ref{Assumption A.1} and \ref{Assumption A.2} and the local alternative hypothesis $H_{a}:R\beta_\tau- r_\tau=b_\tau/D_T$ and $b_\tau$ is a constant vector with dimension J, for SD and WD predictors, as $M_1\rightarrow\infty$ and $T\rightarrow \infty$, it follows that
\begin{align}
\check{Q}_{l-qr} &\xrightarrow{d} \check{Q}_{l-qr}^* \overset{d}{=} \operatorname{N}(0,1)+\left[ R\operatorname{Avar}(\hat{\beta}_\tau^l)R^\top \right]^{-1/2} b_\tau,\quad J=1;\\
Q_{l-qr} & \xrightarrow{d} Q_{l-qr}^* \overset{d}{=} \chi_J^2 +2Z_Q + b_\tau^\top \left[ R\operatorname{Avar}(\hat{\beta}_\tau^l)R^\top \right]^{-1} b_\tau,\quad J\geq 1,
\end{align}
where $Z_Q=Z_\beta^\top \left[ R\operatorname{Avar}(\hat{\beta}_\tau^l)R^\top \right]^{-1} b_\tau$.
\end{prop}
By the definition $\Omega_{zz}= \tau(1-\tau) \Omega_{vv}/(-2c_z)$ and Theorem \ref{multh1m}, it follows that
$\operatorname{Avar}(\hat{\beta}_\tau^l)=-c_z \frac{\tau(1-\tau)}{2f_{u_\tau}(0)^2} (\Sigma_{zx}^*)^{-1}\Sigma_{zz}^* \left[(\Sigma_{zx}^*)^{-1}\right]^\top$ is proportional to $-c_z$ with SD predictors, where
$\Sigma_{zz}^* = \lambda(\operatorname{I_K} - \tilde{S}_a)\Omega_{vv}(\operatorname{I_K} - \tilde{S}_a)^\top +(1-\lambda)(\operatorname{I_K}-\tilde{S}_b)\Omega_{vv} (\operatorname{I_K} - \tilde{S}_b)^\top$
and
$\Sigma_{zx}^* = (\operatorname{I_K} -\tilde{S}_a) \int_0^{\lambda} dJ_x^c(r)\, J_x^c(r)^\top
+(\operatorname{I_K}-\tilde{S}_b) \int_{\lambda}^{1} dJ_x^c(r) \, J_x^c(r)^\top + [ \operatorname{I_K}- \lambda\tilde{S}_a -(1-\lambda)\tilde{S}_b ]\operatorname{E}(v_tv_t^\top)$. 
Thus, the larger the magnitude of $|c_z|$, the smaller the absolute values of the test statistics $\check{Q}_{l-qr}$ and $Q_{l-qr}$ under the local alternative hypothesis $H_{a}$, resulting in reduced power for both $\check{Q}_{l-qr}$ and $Q_{l-qr}$. 
On the other hand, Proposition \ref{mulpropp2} in the appendix shows that while the higher-order term $C_T$ vanishes in both $\check{Q}_{l-qr}$ and $Q_{l-qr}$, the higher order term $B_T^l$, as defined in Proposition \ref{mulpropp2} of the appendix persists for the same reasons as $B_T$. Moreover, the smaller the magnitude of $|c_z|$, the more pronounced the size distortions induced by $B_T^l$ becomes. 
Thus, we mitigate the size distortion induced by the higher-order term $B_T$ by selecting a conservative tuning parameter $c_z=-8-2K$. However, Proposition \ref{mulkeythe2} suggests that this conservative setting of $c_z$ improves the size performance at the expense of reduced power. 
\subsubsection{A Linear Combination Test Statistic to Enhance the Power}\label{subsection3.3}
To amend the loss of power due to the conservative choice of turning parameter $c_z$, we construct the following linear combination test statistic $\check{Q}_{m-qr}$ when $J=1$ for one-sided marginal test.
\begin{align}\label{oraqu2np8}
\check{Q}_{m-qr} = \check{Q}_{l-qr} + \frac{1}{\sqrt{T}} \left(\frac{ \check{Q}_{o-qr}}{\check{q}_{0.999}}\right)^{1/(1-\delta)}\operatorname{sign}(\check{Q}_{o-qr}),
\end{align}
where $\check{q}_{0.999}$ is 99.9\% quantile of the standard normal distribution and $1/(1-\delta)=20$ here. And $\check{Q}_{o-qr}$ is the conventional test statistic defined as $\check{Q}_{o-qr} = (R\hat{\beta}_\tau^o - r_\tau) \big[ \tau(1-\tau) R\big(\sum_{t=1}^T \bar{x}_{t-1} \bar{x}_{t-1}^\top\big)^{-1} R^\top \big]^{-1/2} \hat{f}_{u_\tau}(0)$, where $\bar{x}_{t-1}=x_{t-1}-\frac{1}{T}\sum_{t=1}^T x_{t-1}$ and $\big[\hat{\mu}_\tau^o,(\hat{\beta}_\tau^o)^\top \big]^\top =\arg \,\min_{\mu_\tau, \beta_\tau} \sum_{t=1}^T \rho_{\tau}\left(y_{t}-\mu_\tau-\beta_\tau^\top x_{t-1} \right)$.
{We put $\check{q}_{0.999}$ here to ensure that, under the null hypothesis, most realizations of $|\check{Q}_{o-qr}|$ remain below one. Consequently, this facilitates the second term $\frac{1}{\sqrt{T}} \left(\frac{ \check{Q}_{o-qr}}{\check{q}_{0.999}}\right)^{1/(1-\delta)}\operatorname{sign}(\check{Q}_{o-qr})$ in maintaining the good size performance of $\check{Q}_{m-qr}$ in finite sample without exacerbating it.}
Note that 
\begin{align}\label{dk4r2jgj1}
\check{Q}_{o-qr} &= O_p(1),\quad \text{under the}\; H_0: R\beta_\tau- r_\tau=0;\\
\check{Q}_{o-qr} &\xrightarrow{d} \check{Q}_{o-qr}^{*}+\check{Q}_{o-qr}^{**}, \quad \text{under the}\; H_{a}:R\beta_\tau- r_\tau=b_\tau/D_T, \nonumber
\end{align}
where $ \check{Q}_{o-qr}^{*} \overset{d}{=} \beta_\tau^* \big[ \tau(1-\tau) R \big(\int_0^1 \bar{J}_x^c(r)\bar{J}_x^c(r)^\top \big)^{-1} R^\top \big]^{-1/2} $, 
$\beta_\tau^* \overset{d}{=}
\big[\int_0^1 \bar{J}_x^c(r)\bar{J}_x^c(r)^\top \big]^{-1} \int_0^1 \bar{J}_x^c(r)\\
dB_{\psi_\tau}(r)$, and 
$ \check{Q}_{o-qr}^{**} \overset{d}{=}
b_\tau \big[ \tau(1-\tau) R \big(\int_0^1 \bar{J}_x^c(r)\bar{J}_x^c(r)^\top \big)^{-1} R^\top \big]^{-1/2}$ with SD predictors; 
$ \check{Q}_{o-qr}^{*} \overset{d}{=} \operatorname{N} \big[ 0,1\big]$ 
and 
$ \check{Q}_{o-qr}^{**} \overset{d}{=} b_\tau \big[ \tau(1-\tau) R \operatorname{Var}(x_{t-1})^{-1} R^\top \big]^{-1/2} $ for WD predictors. And $\bar{J}_x^c(r)=J_x^c(r)-\int_0^1 J_x^c(r)dr$ with $0\leq r \leq 1$.
Then we have the following theorem on the local power of $\check{Q}_{m-qr}$.
\begin{thm}\label{jg8jgh1}
Under Assumptions \ref{Assumption A.1} and \ref{Assumption A.2} for SD and WD predictors, it follows that
\begin{align}\label{checkoraqu2np8}
\check{Q}_{m-qr} =
\begin{cases}
\check{Q}_{l-qr} + O_p\left(\frac{1}{\sqrt{T}}\right)\xrightarrow{d}\operatorname{N}(0,1),\, \text{under the}\,H_0;\\
\check{Q}_{l-qr}^* +\operatorname{sign}(\check{Q}_{o-qr}^*+\check{Q}_{o-qr}^{**})\left(\frac{ \check{Q}_{o-qr}^*+\check{Q}_{o-qr}^{**}}{\check{q}_{0.999}} \right)^{1/(1-\delta)},\, \text{under the}\,H_a,
\end{cases}
\end{align}
as $M_1\rightarrow\infty$ and $T\rightarrow \infty$, where the local alternative hypothesis is $H_{a}:R\beta_\tau- r_\tau=b_\tau/D_T$ and $b_\tau$ is a constant vector with dimension $J=1$.
\end{thm}
\begin{remark}\label{remar4}
Theorem \ref{jg8jgh1} shows that $\check{Q}_{m-qr}$ not only inherits the good size performance of $\check{Q}_{l-qr}$ but also the good power performance of $\check{Q}_{o-qr}$. 
The two key factors contributing to the good size and the good power performance of $\check{Q}_{m-qr}$ are outlined as follows. First, under the null hypothesis $H_0:R\beta_\tau =r_\tau$, $\check{Q}_{l-qr}$ plays the key part in $\check{Q}_{m-qr}$ while $\check{Q}_{o-qr}$ converges to zero with a rate $\sqrt{T}$. As result, the higher-order term of $\check{Q}_{m-qr}$ is approximately equal to that of of $\check{Q}_{l-qr}$ under the null hypothesis $H_0:R\beta_\tau =r_\tau$ and thus inherits the good size performance of $\check{Q}_{l-qr}$. Second, both $\check{Q}_{l-qr}$ and $ \left(\frac{ \check{Q}_{o-qr}}{\check{q}_{0.999}}\right)^{1/(1-\delta)}\operatorname{sign}(\check{Q}_{o-qr})$ play a role under the local alternative hypothesis. Thus $\check{Q}_{m-qr}$ inherits the good power performance of $\check{Q}_{o-qr}$. In summary, the weighting technique employed in $\check{Q}_{m-qr}$ enables it to achieve both good size and power performance simultaneously.
\end{remark} 
Next, we construct a linear combination test statistic $Q_{m-qr}$ for two-sided tests using the similar procedure to that used for constructing $\check{Q}_{m-qr}$.
\begin{align}\label{nocoraqu2np8}
{Q}_{m-qr} ={Q}_{l-qr} + \frac{1}{\sqrt{T}}\left( \frac{{Q}_{o-qr}}{{q}_{0.999}}\right)^{1/(1-\delta)},
\end{align}
where ${q}_{0.999}$ is 99.9\% quantile of the $\chi_J^2$ distribution 
and
${Q}_{o-qr}$ is the conventional test statistic defined as ${Q}_{o-qr} \equiv (R\hat{\beta}_\tau^o - r_\tau) \big[ \tau(1-\tau) R\big(\sum_{t=1}^T \bar{x}_{t-1} \bar{x}_{t-1}^\top\big)^{-1} R^\top \big]^{-1} (R\hat{\beta}_\tau^o - r_\tau)^\top \hat{f}_{u_\tau}(0)^{-2}$.
Note that 
\begin{align}\label{2dk4r2jgj1}
{Q}_{o-qr} &= O_p(1),\quad \text{under the}\; H_0:R\beta_\tau- r_\tau=0;\\
{Q}_{o-qr} &\xrightarrow{d}{Q}_{o-qr}^{*}+{Q}_{o-qr}^{**}+{Q}_{o-qr}^{***}+\left({Q}_{o-qr}^{***}\right)^\top, \quad \text{under the}\; H_{a}:R\beta_\tau- r_\tau=b_\tau/D_T; \nonumber
\end{align}
where 
${Q}_{o-qr}^{*} \overset{d}{=} \beta_\tau^* \big[ \tau(1-\tau) R \big(\int_0^1 \bar{J}_x^c(r)\bar{J}_x^c(r)^\top \big)^{-1} R^\top \big]^{-1}\big( \beta_\tau^*\big)^\top$, 
${Q}_{o-qr}^{**} \overset{d}{=} b_\tau \big[ \tau(1-\tau) R \big(\int_0^1 \bar{J}_x^c(r)$ \quad $ \bar{J}_x^c(r)^\top \big)^{-1} R^\top \big]^{-1} b_\tau ^\top$ 
and 
${Q}_{o-qr}^{***} \overset{d}{=}\beta_\tau^* \big[ \tau(1-\tau) R \big(\int_0^1 \bar{J}_x^c(r)\bar{J}_x^c(r)^\top \big)^{-1} R^\top \big]^{-1} b_\tau^\top$ with SD predictors;
${Q}_{o-qr}^{*} \overset{d}{=} \chi_J^2$, 
${Q}_{o-qr}^{**} \overset{d}{=} 
b_\tau \big[ \tau(1-\tau) R \operatorname{Var}(x_{t-1})^{-1} R^\top \big]^{-1}b_\tau ^\top$ 
and 
${Q}_{o-qr}^{***} \overset{d}{=} 
\beta_\tau^* \big[ \tau(1-\tau) R \operatorname{Var}(x_{t-1})^{-1} R^\top \big]^{-1} b_\tau^\top$ for WD predictors.
We can then establish the following theorem, which is analogous to Theorem \ref{jg8jgh1}.
\begin{thm}\label{jg8jghh1}
Under Assumptions \ref{Assumption A.1} and \ref{Assumption A.2} for SD and WD predictors, it follows that
\begin{align}\label{noco1np8}
{Q}_{m-qr} =
\begin{cases}
{Q}_{l-qr} + O_p\left(\frac{1}{\sqrt{T}}\right)\xrightarrow{d}\chi_J^2,\, \text{under the}\,H_0;\\
{Q}_{l-qr}^* +\left[\frac{{Q}_{o-qr}^*+{Q}_{o-qr}^{**}+{Q}_{o-qr}^{***}+ ({Q}_{o-qr}^{**})^\top}{{q}_{0.999}} \right]^{1/(1-\delta)},\, \text{under the}\,H_a,
\end{cases}
\end{align} 
as $M_1\rightarrow\infty$ and $T\rightarrow \infty$.
\end{thm}
As indicated in Remark \ref{remar4}, Theorem \ref{jg8jghh1} shows that ${Q}_{m-qr}$ not only retains the favorable size performance of ${Q}_{l-qr}$ but also adopts the good power performance of ${Q}_{o-qr}$. 
\section{Robust Inference for Unspanned Predictability}\label{section4}
In this section we test the macro-spanning hypothesis using the inference approach outlined in Section \ref{sectionEconometric}. We apply the proposed test statistic $Q_{m-qr}$ to test the predictive power of macro variables for n-year bond risk premia rn(x), with n= 2,3,4,5, controlling for CP factor at quantile level $\tau_j\in (0.01,0.02,\cdots,0.99)$, where $j =1,2,\cdots,J$ and $J=99$. \footnote{The same choice of $\tau_j$ and $J$ is seen in many other empirical studies such as \cite{RiskyOil2024}.}
We report the testing results in Section \ref{inferres} and give the potential economic explanation for our test results in Section \ref{poexpr}. We compare our results with those of CCL2023, the only feasible method in existing literature that is designed to do marginal test in the quantile predictive model with many highly persistent predictors.
\subsection{Inference Results}\label{inferres}
The results for the two-sided marginal test are reported in Figure \ref{logPvalueTwoside}. \footnote{See appendix of subsection \ref{onesidetest1} for results of the one-sided test statistic $\check{Q}_{m-qr}$.} The horizontal axis is $\tau_j$ and the vertical axis is $\ln{[\max(0.001,P^*)]}$, where $P^*$ is p-values of $Q_{m-qr}$ and CCL2023. We take the logarithm of the p-value instead of the original p-value to show the statistical significance more clearly in Figure \ref{logPvalueTwoside}. The dashed horizontal line at $-4.605$ corresponds to the significant level, which is set at $0.01$. In this way, the conclusion of predictability is more reliable with small type I error, and the corresponding out-of-sample performance is shown to be good. Additionally, we take the maximum value of p-value and 0.001 to avoid $\ln{(P^*)}$ being too small to disrupt the scale of that figure. The macro variable exhibits significant predictive power when it is within the gray region of Figure \ref{logPvalueTwoside}. 
\begin{figure}[htbp]
\centering
\includegraphics[width=1.1\linewidth]{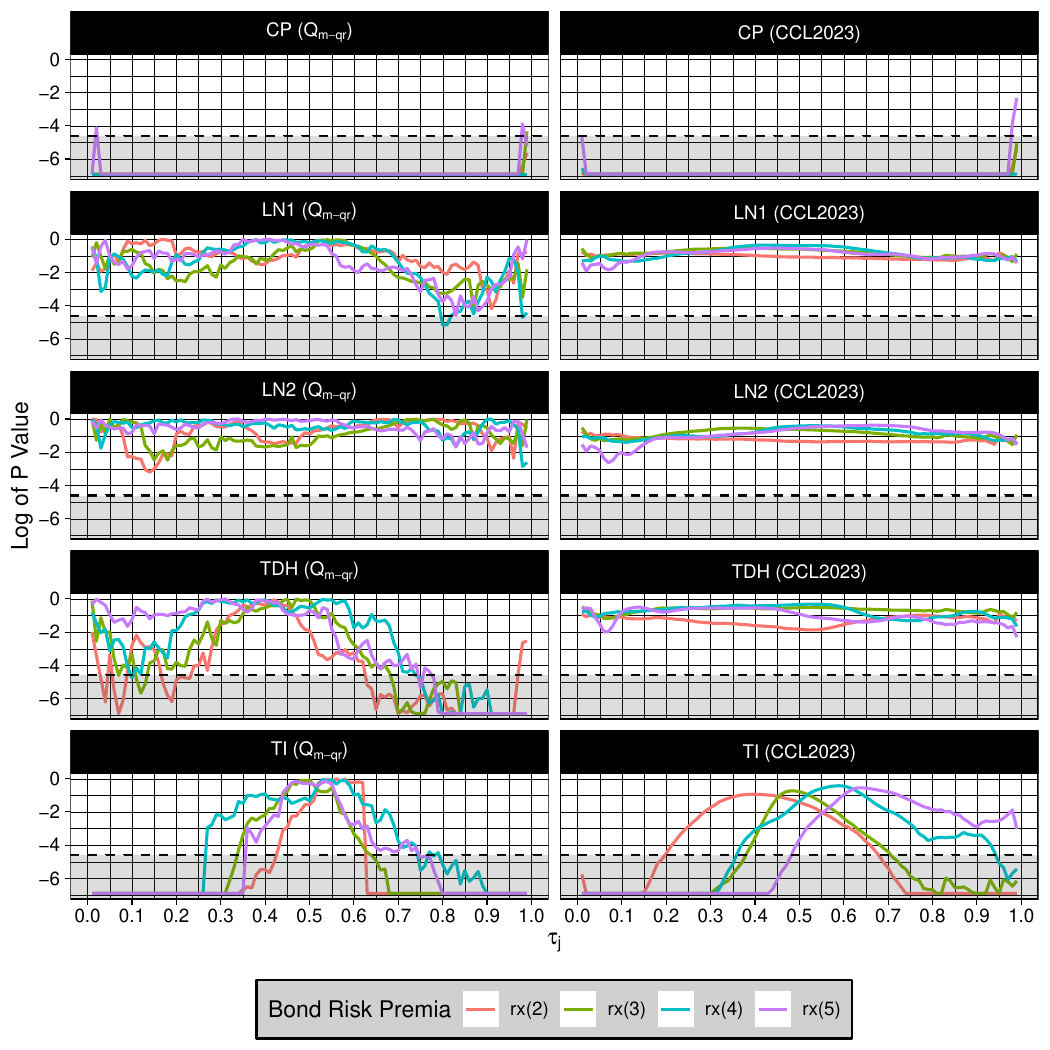}
\caption{Log of p-value of $Q_{m-qr}$ and CCL2023 for Two-sided Test}
\label{logPvalueTwoside}
\end{figure}
The following conclusions could be drawn from the empirical results. First, at center quantiles, the CP factor has positive predictive power while all macro variables have negligible predictive power. This observation supports the macro-spanning hypothesis and aligns with the empirical results reported in studies by \cite{Ghysels2018} and \cite{Bauer2018} using mean regressions. Secondly, the CP factor consistently exhibits positive predictive power across all quantiles, while LN1 and LN2 show no predictive power at any quantile, strongly supporting the spanning hypothesis. Third, TDH demonstrates predictive power at the right tail quantiles for bonds of all maturities and at the left tail quantiles for 2- and 3-year bonds. Fourth, TI exhibits predictive power for bonds of all maturities across all quantiles, except for the median quantiles. The size distortion due to the highly persistence of macro variable is significantly mitigated by the proposed two-sided marginal ($Q_{m-qr}$) test. This study effectively eliminates the spurious predictive power of LN1 and LN2 across all quantiles, and of TDH and TI at the center quantiles. 
{Although CCL2023 makes the same conclusion for most cases}, our findings diverge from those of CCL2023 regarding the predictive power of TDH and TI at tail quantiles. First, we find that TDH has significant predictive power at right tail quantiles for bonds of all maturities, a finding not reported by CCL2023. This capability of TDH is particularly valuable for predicting tail risks associated with bond risk premia. Second, we find that TI has predictive power at both tails quantiles for bonds of all maturities, while CCL2023 fails to find the predictive power of TI at right tail quantiles for 4- and 5-year bonds. Third, we find the significant predictive power of TDH at left tail quantiles for 2- and 3-year bonds, which CCL2023 does not find. Considering the simulations presented in the appendix, which demonstrate that our size and power performance surpass those of CCL2023 across all quantiles, the most plausible explanation for the discrepancies between our findings and those of CCL2023 is that CCL2023 likely commits type II errors (failing to detect true predictive power) at tail quantiles, whereas our analysis does not. 
The predictability at tail quantiles (high-risk region) is practically beneficial, e.g., for constructing the tail risk indicator in the appendix Section \ref{tairisin1}.
\subsection{Economic Explanations of Testing Results}\label{poexpr}
The macro-spanning hypothesis states that the yield curve contains all available information about future bond risk premia, which provides economic interpretation about the empirical findings at the center quantiles. The CP factor contains the information of the yield curve so it has significant predictive power while LN1 and LN2 have negligible predictive power at all quantiles. Moreover, the spanning hypothesis also explains the test results that all other variables, including TDH and TI, have no predictive power at center quantiles. 
The most interesting empirical finding of this paper is the significant predictive power of TDH at right tail quantiles and the significant predictive power of TI at both tails quantiles for bonds of all maturities. The predictive power of TDH at left tail quantiles for 2- and 3-year bonds are also found. Subsequently, we offer potential economic explanations for these findings. 
Since investors prefer 1-year bonds over n-year bonds (n=2,3,4,5) at right tail quantiles, \footnote{See more discussions of this point in Section \ref{sectionData}.} this preference results in a more pronounced increase in demand for 1-year bonds compared to n-year bonds when the TDH rises. Consequently, as TDH increases, the price of 1-year bonds experiences a sharper increase relative to the prices of n-year bonds (n=2,3,4,5). This sharper price increase leads to a more substantial decrease in the YTM of 1-year bonds compared to that of n-year bonds. Given that the risk premium of n-year bonds (n=2, 3, 4, 5) is defined as the difference between the YTM of n-year and 1-year bonds, it follows that the risk premia for n-year bonds will increase if TDH rises at the right tail quantiles. 
In other words, TDH has positive predictive power for bonds of all maturities at right tail quantiles. 
On the other hand, when investors exhibit a preference for bonds with maturities of two and three years over 1-year bonds at the left tail quantiles, the demand for these 2- and 3-year bonds significantly increases. This heightened demand occurs particularly if the TDH rises. Consequently, the prices of 2- and 3-year bonds rise more sharply compared to 1-year bonds under the same conditions. This increase in price leads to a more pronounced decrease in the YTM of 2- and 3-year bonds than that observed in 1-year bonds. As a result, the bond risk premia, which represent the difference between the returns of 2- and 3-year bonds and 1-year bonds, decrease if the TDH increases at the left tail quantiles. This effect highlights the sensitivity of bond risk premia to shifts in investor preference and market conditions at the left tail quantiles. In other words, TDH has negative predictive power for both 2- and 3-year bonds at left tail quantiles. 
Besides that, the demand for n-year (n=2,3,4,5) and 1-year bonds may rise comparably if TDH increases within the center quantiles. Additionally, at the left tail quantiles of the risk premia for bonds with longer maturities (n=4,5 years), the risk associated with 1-year bonds is high, while bonds with maturities of four and five years continue to face elevated risks due to their longer durations. Consequently, if the TDH increases at the left tail quantiles, the demand for both 1-year bonds and longer-term bonds (n=4,5 years) may rise comparably. As a result, the prices and YTM for both 1-year and longer-term bonds (n=4,5 years) may move in the same direction and by approximately the same magnitude, leading to relatively stable bond risk premia. That is, TDH has no predictive power at center quantiles for bonds of all maturities and left tail quantiles for 4- and 5-year bonds.
Next, we discuss the predictive power of TI found at all quantiles except center quantiles. If TI increases, the real return of investors decrease, leading to a reduced demand for all 1- to 5-year bonds. However, the demand for bonds with maturities of two to five years decreases more sharply than that for 1-year bonds when TI increases since investors prefer 1-year bonds at right tail quantiles. \footnote{See more discussions of this point in Section \ref{sectionData}.} As a result, the prices of n-year (n=2,3,4,5) bonds decrease more sharply than that of a 1-year bond when TI increases. Thus, the YTM of these longer-term bonds rise more significantly than those of 1-year bonds. Given that the risk premia for bonds with maturities of two to five years are defined as the difference between their YTM and that of 1-year bonds, the risk premia for these longer-term bonds (n=2, 3, 4, 5) increase as TI rises. That is, TI has positive predictive power for all bonds at right tail quantiles. Following the same logic, TI has negative predictive power for all bonds at left tail quantiles. 
\section{In-sample and Out-of-sample Performance Evaluation}\label{appinf1}
In this section, we use the test results of $Q_{m-qr}$ in Section \ref{section4} to obtain the in-sample and out-of-sample predictions and compare the performance with those using the test of CCL2023. We also use the out-of-sample predictive power to introduce a new tail risk indicator of the U.S. bond market in the appendix Section \ref{tairisin1}. 
The quantile-weighted continuous ranked probability score (qw-CRPS) criterion introduced by 
\cite{Gneiting2011} is useful to evaluate the in-sample and out-of sample prediction performance of our method at left tail quantiles, center quantiles, right tail quantiles and both tails quantiles. Specifically, the qw-CRPS is a scoring rule that takes into account information from the entire predictive density, but it also permits researchers to focus more on certain areas of interest. For example, evaluating how effectively a model predicts downside risks could be of independent interest, for this task qw-CRPS puts more emphasis on the area in the right tail quantiles.
{The in-sample prediction of condition quantile of bond risk premia $\hat{y}_{t\tau_j}$ is obtained as follows.} 
First, using the two-sided test results of $Q_{m-qr}$ in Section \ref{section4}, we select predictors $\dot{x}_{t-1}^j=(x_{i_{j1},t-1},x_{i_{j2},t-1},\cdots,x_{i_{jI},t-1},)^\top$ in CP factor, LN factors, TI and TDH if the p-values of $x_{i_{j1},t-1},x_{i_{j2},t-1},\cdots,x_{i_{jI},t-1}$ are less than 1\%, where $1\leq i_{j1} < i_{j2} <\cdots<i_{jI} \leq K$. That is, $\dot{x}_{t-1}^j$ has predictive power at quantile $\tau_j$. Second, we run the following quantile regressions
$\big(\hat{\dot\mu}_{\tau_j}, \hat{\dot\beta}_{\tau_j}^\top \big)^\top =\arg \,\min_{\mu_{\tau_j}, \beta_{\tau_j}} \sum_{t=1}^T \rho_{\tau_j} \, (y_{t}-\mu_{\tau_j}-\dot\beta_{\tau_j}^\top \dot{x}_{t-1}^j )$
at quantile $\tau_j$, where $j=1,2,\cdots,J${, which are defined in the beginning of Section \ref{section4}}.
Finally, we obtain in-sample prediction value $\hat{y}_{t\tau_j}=\hat{\dot\mu}_{\tau_j}+\hat{\dot\beta}_{\tau_j}^\top \dot{x}_{t-1}^j$ and the in-sample error estimator $\hat{u}_{t\tau_j}=y_t-\hat{y}_{t\tau_j}$.
Then the qw-CRPS, which is calculated by applying weights to $\hat{y}_{t\tau}$ for $J-1$ different quantiles, is computed as follows 
\begin{align}\label{cer1ev}
\operatorname{qw}-{\operatorname{CRPS}_t} =\frac{2}{J-1} \sum_{j=1}^{J-1} W\left(\tau_j\right) \rho_{\tau_j}(\hat{u}_{t\tau_j}),\; \operatorname{qw}_c = \frac{1}{T} \sum_{t=1}^T \operatorname{qw}-{\operatorname{CRPS}_t}
\end{align}
where $W\left(\tau_j\right)$ is a deterministic weighting function. The considered weighting functions are: $W\left(\tau_j\right)=\tau_j\left(1-\tau_j\right)$, which gives more weight to the center of the predictive distribution; $W\left(\tau_j\right)=\left(1-\tau_j\right)^2$ (and $W\left(\tau_j\right)=\tau_j^2$), which assigns greater weight to the left (right) tail to emphasize downside (upside) risks; and $W\left(\tau_j\right)=\left(2 \tau_j-1\right)^2$ which places emphasis on the prediction performance at both tails. 
Using the qw-CRPS criteria defined in (\ref{cer1ev}), we obtain the following in-sample performance of our proposed approach and CCL2023 in Table \ref{inSamplePerf}. \footnote{{In-sample performance of CCL2023 is obtained by the same way of $Q_{m-qr}$.}} It shows that the in-sample performance of the proposed test $Q_{m-qr}$ is much better than CCL2023 at center quantiles, left tail quantiles, right tail quantiles and both tails quantiles, since all of its $\operatorname{qw}_c$ values are much less than those of CCL2023. The excellent in-sample performance of $Q_{m-qr}$ arises from its success of finding the significant predictive power of TDH and TI at right tail quantiles. This result is also shown in Table \ref{inSamplePerf}, that the difference between $\operatorname{qw}_c$ of $Q_{m-qr}$ and CCL2023 when the weight $W(\tau_j)$ emphasizing in right tail quantiles is much larger than those in left tail quantiles and center quantiles. %
\setlength{\tabcolsep}{8pt}
\renewcommand{\arraystretch}{0.6}
\captionsetup{
width=0.8\textwidth,
}
\begin{longtable}{l|cccc|cccc}
\caption{In-sample Performance $\operatorname{qw}_c$ (in $10^{-3}$) in 1980-2022} \label{inSamplePerf}\\
\hline
Method & \multicolumn{4}{c|}{CCL2023} & \multicolumn{4}{c}{$Q_{m-qr}$} \\
\hline
Emphasis & rx(2) & rx(3) & rx(4) & rx(5) & rx(2) & rx(3) & rx(4) & rx(5) \\
\hline
\endfirsthead
\hline
\multicolumn{9}{c}{{ \tablename\ \thetable{} -- continued from previous page}} \\
\hline
\endhead
\hline
\multicolumn{9}{r}{{Continued on next page}} \\
\hline
\endfoot
\hline
\endlastfoot
Center & 1.0 & 1.8 & 2.6 & 3.1 & 0.6 & 1.1 & 1.5 & 1.9 \\
Left Tail & 1.1 & 2.0 & 2.8 & 3.5 & 0.9 & 1.6 & 2.2 & 2.8 \\
Right Tail & 2.7 & 5.0 & 7.3 & 8.6 & 1.0 & 1.8 & 2.5 & 3.2 \\
Both Tails & 1.9 & 3.5 & 5.0 & 5.9 & 0.7 & 1.2 & 1.7 & 2.1 \\
\end{longtable} 
Next, we present the out-of-sample prediction results. To evaluate the out-of-sample prediction performance (in the out-of-sample period from $T_m$ to $T$, and we set $T_m$ to be 1992/01 and $T$ to be 2022/12), we undertake the following steps. To predict the conditional quantile of bond risk premia at period $\tilde{T}$ and quantile level $\tau_j$, where $j=1,2,\cdots,J$ and $T_m \le \tilde{T}\le T$, we first use the sample during period $1-(\tilde{T}-1)$ to run the following quantile regression to obtain the real time coefficient estimator $ 
\big[\hat{\dot\mu}_{\tau_j}(\tilde{T}), \hat{\dot\beta}_{\tau_j}(\tilde{T})^\top \big]^\top =\arg \,\min_{\mu_{\tau_j}, \beta_{\tau_j}} \sum_{t=1}^{\tilde{T}-1} \rho_{\tau_j} (y_{t}-\mu_{\tau_j}-\dot\beta_{\tau_j}^\top \dot{x}_{t-1}^j )$, where $\dot{x}_{t-1}^j$ is the one used in (\ref{cer1ev}). 
Then the out-of-sample predicted value of bond risk premia at quantiles $\tau_j$ is $\hat{y}_{\tau_j}^o(\tilde{T})=\hat{\dot\mu}_{\tau_j}(\tilde{T})+ \hat{\dot\beta}_{\tau_j}(\tilde{T})^\top \dot{x}_{\tilde{T}-1}^j$ and the out-of-sample quantile regression error is $\hat{u}_{\tau_j}^o (\tilde{T})=y_{\tilde{T}}-\hat{y}_{\tau_j}^o (\tilde{T})$. The out-of-sample performance in the period from $T_0$ to T is evaluated by qw-CRPS. 
Specifically, the out-of-sample evaluation criteria are 
$\operatorname{qw}-{\operatorname{CRPS}_{\tilde{T}}}=\frac{2}{J-1} \sum_{j=1}^{J-1} W\left(\tau_j\right) \rho_{\tau_j}[\hat{u}_{\tau_j}^o (\tilde{T})]$, and $ \operatorname{qw}_c(T_m) = \frac{1}{T-T_m+1} \sum_{\tilde{T}=T_m}^T \operatorname{qw}-{\operatorname{CRPS}_{\tilde{T}}}$.
By this criteria, we obtain the out-of-sample performance in the period 1992--2022, which is shown in Table \ref{outSamplePerf}. 
\setlength{\tabcolsep}{8pt}
\renewcommand{\arraystretch}{0.6}
\captionsetup{
width=0.8\textwidth,
}
\begin{longtable}{l|cccc|cccc}
\caption{Out-of-sample Performance $\operatorname{qw}_c(T_m)$ (in $10^{-3}$) in 1992--2022} \label{outSamplePerf}\\
\hline
Method & \multicolumn{4}{c|}{CCL2023} & \multicolumn{4}{c}{$Q_{m-qr}$} \\
\hline
\endfirsthead
\hline
\multicolumn{9}{c}{{ \tablename\ \thetable{} -- continued from previous page}} \\
\hline
\endhead
\hline
\multicolumn{9}{r}{{Continued on next page}} \\
\hline
\endfoot
\hline
\endlastfoot
Emphasis & rx(2) & rx(3) & rx(4) & rx(5) & rx(2) & rx(3) & rx(4) & rx(5) \\
\hline
Center & 0.5 & 0.8 & 1.2 & 1.6 & 0.5 & 0.7 & 1.0 & 1.5 \\
Left Tail & 0.9 & 1.2 & 2.0 & 2.3 & 0.9 & 1.1 & 1.6 & 2.1 \\
Right Tail & 1.5 & 2.5 & 3.5 & 4.8 & 1.0 & 1.7 & 2.3 & 3.3 \\
Both Tails & 1.4 & 2.1 & 3.2 & 3.9 & 0.9 & 1.4 & 1.8 & 2.5 \\
\end{longtable} 
The out-of-sample performance of $Q_{m-qr}$ is better than those of CCL2023, as $\operatorname{qw}_c(T_0)$ of $Q_{m-qr}$ is much less than those of CCL2023 in all quantiles except for rx(2) at center quantiles and left tail quantiles where they are equivalent. The excellent out-of-sample performance of $Q_{m-qr}$ also arises from its success of finding the significant predictive power of TDH and TI at right tail quantiles. The difference between $\operatorname{qw}_c(T_0)$ of $Q_{m-qr}$ and CCL2023 is much larger at right tail quantiles than those at left tail quantiles, center quantiles and both tails quantiles.
\section{Conclusion}\label{section7} 
We examine the predictability of bond risk premia using predictive quantile regression with many highly persistent predictors, employing a reliable inference procedure tailored for this analysis. Our results lend support to the macro-spanning hypothesis at center quantiles, while evidence at tail quantiles contradicts this hypothesis. This dichotomy helps clarify inconsistencies observed in previous studies that utilized mean regressions. Specifically, all predictors, with the exception of the CP factor, exhibit negligible predictive power at the center quantile. Additionally, TDH demonstrates predictive power at the right tail quantiles, and TI shows predictive power at both left and right tail quantiles. Furthermore, we demonstrate that both in-sample and out-of-sample prediction performance using our proposed inference method significantly outperforms existing methods.
\bibliographystyle{apalike}
\bibliography{Bibliography}

\begin{thebibliography}{}

\bibitem[Andreasen et~al., 2021]{AndreasenEngsted2021}
Andreasen, M.~M., Engsted, T., Møller, S.~V., and Sander, M. (2021).
\newblock {The Yield Spread and Bond Return Predictability in Expansion and
  Recessions}.
\newblock {\em Review of Economic Studies}, 34(6):2773–2812.

\bibitem[Bansal and Shaliastovich, 2013]{Bansal2013}
Bansal, R. and Shaliastovich, I. (2013).
\newblock {A Long-Run Risks Explanation of Predictability Puzzles in Bond and
  Currency Markets}.
\newblock {\em Review of Financial Studies}, 26(1):1--33.

\bibitem[Bauer and Hamilton, 2018]{Bauer2018}
Bauer, M.~D. and Hamilton, J.~D. (2018).
\newblock {Robust Bond Risk Premia}.
\newblock {\em Review of Financial Studies}, 31(2):399--448.

\bibitem[Baumeister et~al., 2024]{RiskyOil2024}
Baumeister, C., Huber, F., and Marcellino, M. (2024).
\newblock Risky oil: It's all in the tails.
\newblock Working Paper 32524, National Bureau of Economic Research.

\bibitem[Borup et~al., 2024]{Borup2024}
Borup, D., Eriksen, J.~N., Kjær, M.~M., and Thyrsgaard, M. (2024).
\newblock {Predicting Bond Return Predictability.}
\newblock {\em Management Science}, 70(2):931--951.

\bibitem[Cai et~al., 2023]{CaiChenLiao2023}
Cai, Z., Chen, H., and Liao, X. (2023).
\newblock A new robust inference for predictive quantile regression.
\newblock {\em Journal of Econometrics}, 234(1):227--250.

\bibitem[Cieslak and Povala, 2015]{CieslakPovala2015}
Cieslak, A. and Povala, P. (2015).
\newblock {Expected Returns in Treasury Bonds}.
\newblock {\em Review of Financial Studies}, 28(10):2859–2901.

\bibitem[Cochrane and Piazzesi, 2005]{CochranePiazzesi2005}
Cochrane, J.~H. and Piazzesi, M. (2005).
\newblock {Bond Risk Premia}.
\newblock {\em American Economic Review}, 95(1):138–160.

\bibitem[Cooper and Priestley, 2008]{CooperPriestley2008}
Cooper, I. and Priestley, R. (2008).
\newblock {Time-Varying Risk Premiums and the Output Gap}.
\newblock {\em Review of Financial Studies}, 22(7):2801–2833.

\bibitem[Fan and Lee, 2019]{FanLee2019}
Fan, R. and Lee, J. (2019).
\newblock {Predictive Quantile Regressions under Persistence and Conditional
  Heteroskedasticity}.
\newblock {\em Journal of Econometrics}, 213(1):261--280.

\bibitem[Ghysels et~al., 2018]{Ghysels2018}
Ghysels, E., Horan, C., and Moench, E. (2018).
\newblock {Forecasting through the Rearview Mirror: Data Revisions and Bond
  Return Predictability}.
\newblock {\em Review of Financial Studies}, 31(2):678–714.

\bibitem[Gneiting and Ranjan, 2011]{Gneiting2011}
Gneiting, T. and Ranjan, R. (2011).
\newblock {Comparing Density Forecasts Using Threshold and Quantile-weighted
  Scoring Rules}.
\newblock {\em Journal of Business and Economic Statistics}, 29(3):411--422.

\bibitem[Greenwood and Vayanos, 2014]{GreenwoodVayanos2014}
Greenwood, R. and Vayanos, D. (2014).
\newblock {Bond Supply and Excess Bond Returns}.
\newblock {\em Review of Financial Studies}, 27(3):663–713.

\bibitem[Joslin et~al., 2014]{Joslin2014}
Joslin, S., Priebsch, M., and Singleton, K.~J. (2014).
\newblock {Risk Premiums in Dynamic Term Structure Models with Unspanned Macro
  Risks}.
\newblock {\em Journal of Finance}, 69(3):1197–1233.

\bibitem[Kostakis et~al., 2015]{Kostakisetal2015}
Kostakis, A., Magdalinos, T., and Stamatogiannis, M. (2015).
\newblock {Robust Econometric Inference for Stock Return Predictability}.
\newblock {\em Review of Financial Studies}, 28(5):1506--1553.

\bibitem[Lee, 2016]{Lee2016}
Lee, J. (2016).
\newblock {Predictive Quantile Regression with Persistent Covariates: IVX-QR
  Approach}.
\newblock {\em Journal of Econometrics}, 192(1):105--118.

\bibitem[Liao et~al., 2024]{liao2024robust}
Liao, X., Li, X., and Fan, Q. (2024).
\newblock {Robust Inference for Multiple Predictive Regressions with an
  Application on Bond Risk Premia}.
\newblock {\em arXiv preprint arXiv:2401.01064}.

\bibitem[Liu et~al., 2023]{LiuYangetal2023}
Liu, X., Long, W., Peng, L., and Yang, B. (2023).
\newblock {A Unified Inference for Predictive Quantile Regression}.
\newblock {\em Journal of the American Statistical Association},
  119(546):1526–1540.

\bibitem[Ludvigson and Ng, 2009]{ludvigson2009macro}
Ludvigson, S.~C. and Ng, S. (2009).
\newblock Macro factors in bond risk premia.
\newblock {\em Review of Financial Studies}, 22(12):5027--5067.

\bibitem[Phillips, 1987]{Phillips1987}
Phillips, P. C.~B. (1987).
\newblock {Towards a Unified Asymptotic Theory for Autoregression}.
\newblock {\em Biometrika}, 74(3):535--547.

\bibitem[Phillips and Magdalinos, 2009]{PhillipsMagdalinos2009}
Phillips, P. C.~B. and Magdalinos, T. (2009).
\newblock {Econometric Inference in the Vicinity of Unity}.
\newblock CoFieWorking Paper.

\bibitem[Xiao, 2009]{Xiao2009}
Xiao, Z. (2009).
\newblock {Quantile Cointegrating Regression}.
\newblock {\em Journal of Econometrics}, 150(2):248--260.

\bibitem[Zhao et~al., 2021]{ZhaoZhouguofu2021}
Zhao, F., Zhou, G., and Zhu, X. (2021).
\newblock {Unspanned Global Macro Risks in Bond Returns}.
\newblock {\em Management Science}, 67(12):7825–7843.

\end{thebibliography}
\newpage
\setcounter{page}{1}
\counterwithin{figure}{section}
\counterwithin{table}{section}
\counterwithin{thm}{section}
\counterwithin{prop}{section}
\counterwithin{equation}{section}
{\Large \bf 
\begin{center}
Supplementary Material to ``Robust Bond Risk Premia Predictability Test in the Quantiles''
\end{center}
}
In this supplementary material, we provide the following parts:
Appendix \ref{app:A} gives the details of computation algorithm to construct our test statistics used in the main text. Appendix \ref{section5} collects all the results for Monte Carlo simulations. The numerical results show that it is necessary to use the new method introduced in the main text Section \ref{sectionEconometric} for the macro-spanning hypothesis test. Appendix \ref{app:C} includes additional asymptotic properties for the higher-order terms that appear in the test statistics construction, whose results we applied in the main text Section \ref{sectionEconometric}. Appendix \ref{app:D} provides additional empirical results. A new tail risk indicator is also introduced there, which further showcases the usefulness of our test. Appendix \ref{app:E} presents selected proofs for main theorems. 
\begin{description}
\item[Explanation about subscript:] 
For any matrix $A$, $A^{1/2}$ is a matrix defined as $A^{1/2} = Q_L \operatorname{diag}( \lambda_1^{1/2}, \lambda_2^{1/2},\cdots,\lambda_K^{1/2})Q_L^\top$, in which $Q_L$ is the matrix whose $i$th column is the eigenvectors of $A$, and $\lambda_i$, $i=1,2,\cdots,K$ are the eigenvalues of $A$. It implies that $A^{1/2}A^{1/2}=A$. Also, we define $A^{-1/2}=(A^{1/2})^{-1}$.
\end{description}
\appendix
\renewcommand{\thesection}{\Alph{section}}
\renewcommand{\thesubsection}{\Alph{section}.\arabic{subsection}}
\section{Algorithm for Construction of Proposed Test Statistics}\label{app:A}
The new test procedure employed in the testing of macro-spanning hypothesis is summarized in Algorithm \ref{algorithm1} below. Equation numbers refer to those in the main text.
\begin{algorithm}[H]
\footnotesize
\caption{\label{algorithm1}The instructions for the construction of the test statistics $\check{Q}_{m-qr}$ and ${Q}_{m-qr}$}
\hspace*{0.01in}
\begin{algorithmic}[1]
\State Construct the IV $z_{t-1}$ by equation (\ref{mulivz}).
\State To eliminate the higher-order term $C_T$, we construct the IV $\tilde{z}_{t-1} = (\operatorname{I_K}-S_a)z_{t-1}$ when
$1\leq t\leq T_0$ while $\tilde{z}_{t-1} = (\operatorname{I_K}-S_b)z_{t-1}$ when $T_0+1 \leq t \leq T$.
\State Construct the IV estimator $\hat{\beta}_\tau^l$ by a two-step regression, which is consistent under both the null and the alternative hypothesis. Specially, we run regression $ (\hat{\mu}_x^l,\hat{\theta}^l) = \arg \; \min_{\mu_x,\theta} \sum_{t= 1}^T \left(x_{t-1}- \mu_x - \theta \tilde{z}_{t-1} \right)^\top \left(x_{t-1}- \mu_x - \theta \tilde{z}_{t-1} \right)$ to obtain the fitted value $\tilde{x}_{t-1}^l = \hat{\mu}_x^l+\hat{\theta}^l \tilde{z}_{t-1}$ and the residual $\tilde{v}_{t-1}^l = x_{t-1}- \tilde{x}_{t-1}^l$ in the first step while $ \big[\hat\mu_\tau^l,(\hat\beta_\tau^l)^\top,(\hat\gamma_\tau^l)^\top\big]^\top =\arg \min_{\mu_\tau, \beta_\tau,\gamma_\tau} \sum_{t= 1}^T \rho_{\tau}\left(y_{t}-\mu_\tau-\beta_\tau^\top \tilde{x}_{t-1}^l -\gamma_\tau^\top \tilde{v}_{t-1}^l\right)$ in the second step.
\State Using the IV estimator $\hat{\beta}_\tau^l$ to construct the test statistics $Q_{l-qr}$ and $\check{Q}_{l-qr}$ in equations (\ref{walde2tes}) and (\ref{defqc1}), which is valid for marginal test with many highly persistent predictors.
\State Improve size performance at tail quantiles and with many predictors: 
\begin{itemize}
\item Use the simulation-based estimator $\hat{f}_{u_\tau}(0)$ for density $f_{u_\tau}(0)$ in (\ref{dj74e6}), which is more accurate than that of CCL2023.
\item Set the conservative tuning parameter $c_z=-8-2K$ to reduce the size distortion induced by the higher-order term $B_T$.
\end{itemize}
\State Amend the loss of power due to the conservative tuning parameter $c_z=-8-2K$ while keeping a good size performance: Construct ${Q}_{m-qr} ={Q}_{l-qr} + \frac{1}{\sqrt{T}}\left( \frac{{Q}_{o-qr}}{{q}_{0.999}}\right)^{1/(1-\delta)}$ in (\ref{oraqu2np8}) for two-sided test.
\If{Focus on the one-sided marginal test $H_0:\beta_i=0$ vs $H_a^r:\beta_i>0$ and $H_0:\beta_i=0$ vs $H_a^l:\beta_i<0$ }
\State Construct the t-test statistic $\check{Q}_{m-qr} = \check{Q}_{l-qr} + \frac{1}{\sqrt{T}} \left(\frac{ \check{Q}_{o-qr}}{\check{q}_{0.999}}\right)^{1/(1-\delta)}\operatorname{sign}(\check{Q}_{o-qr})$.
\EndIf
\end{algorithmic}
\end{algorithm}
\section{Monte Carlo Simulations}\label{section5}
To show the finite sample performance of the proposed inference procedure in predictive quantile regression with many persistent predictors, we conduct two Monte Carlo experiments.
The first experiment considers a data generating process (DGP) with a univariate predictor, while the second experiment is devoted to multivariate predictive model. For the first experiment, we compare the proposed test ${Q}_{m-qr}$ with the tests in \cite{Lee2016}, \cite{FanLee2019}, CCL2023 and \cite{LiuYangetal2023} within the context of the two-sided test. We also show the test result of the proposed test statistic $\check{Q}_{m-qr}$ for one-sided test. For the second experiment, we compare the size performance of the proposed test $Q_{m-qr}$ and CCL2023 for two-sided test with multiple predictors. Moreover, we show the power performance of $Q_{m-qr}$ and $\check{Q}_{m-qr}$ and the size performance of $\check{Q}_{m-qr}$ in multivariate models. Additionally, we show the result of the joint test by \cite{Lee2016} in multivariate models.
In all experiments, the DGP of $x_{t-1}$ is equation (\ref{mulgtuA1}) while the DGP of $y_t$ is the same as equation (1) of \cite{LiuYangetal2023}, which is described in Remark \ref{remaDGP}.
We set $\mu_\tau=1$ and $\mu_\sigma=1$. The DGP of the error $v_{i,t}$ is as follows: $v_{i,t}=\gamma_i \zeta_t + \check{v}_{i,t}$, where $(\zeta_t,\check{v}_{i,t})^\top \;\sim \; i.i.d. \, \operatorname{N}(0_K,\operatorname{I_K})$.
Thus the contemporaneous correlation coefficient between $\zeta_t$ and $v_{i,t}$ is
$\gamma_i (1+\gamma_i^2)^{-1/2}$, which is the source of the size distortion shown in Proposition \ref{mulpropfdie3}.
We report the simulation results for a i.i.d. model with $\sigma_t=1$ for $t=1,2,\cdots,T$ and the sample size $T=750$ with the nominal size $5\%$. \footnote{The absolute value of the correlation coefficient between $u_{t\tau}$ and $v_{i,t}$ in i.i.d. model is significantly greater than that in the GARCH model. This difference leads to a larger size distortion in all approaches when compared to their counterparts in the GARCH model. Consequently, it is sufficient to demonstrate the advantages of the proposed test statistics using only the i.i.d. model results, rather than including those from the GARCH model. To maintain brevity and convey the main message, we have omitted the results for the GARCH and ARCH models with sample sizes $T=250$, $T=500$, and $T=750$, as well as for the i.i.d. model with sample sizes $T=250$ and $T=500$. Detailed results for these models are available upon request.} Simulation is repeated $5000$ times for each setting.
Hereafter, \cite{Lee2016}, \cite{FanLee2019} and \cite{LiuYangetal2023} are referred to as L2016, FL2019 and LLPY2023, respectively.
\textbf{Experiment A1:}\quad In experiment A1, we conduct the test in univariate model with $\gamma_i=-3$ such that $\operatorname{corr}(u_t,v_t)=-0.95$. First, the size performance (in \%) of two-sided test in univariate model of L2016, FL2019, CCL2023 and LLPY2023 and the proposed test statistic ${Q}_{m-qr}$ for two-sided test $H_0:\beta_\tau=0$ are shown in Table \ref{sizeunil}, in which $\tau=0.05,0.25,0.5,0.75,0.95$ and predictors are SD ($c=0$), SD ($c=-5$), SD ($c=-15$) and WD ($c=-0.05$). Second, the power performances of L2016, FL2019, CCL2023 and LLPY2023 and the proposed test statistic ${Q}_{m-qr}$ for two-sided test $H_0:\beta_\tau=0$ vs $H_a:\beta_\tau\neq 0$ are shown in Figure \ref{poweruni1}, in which $\tau=0.5$ and predictors are SD ($c=0$), SD ($c=-5$), SD ($c=-15$) and WD ($c=-0.05$). Third, the size performance (in \%) and the power performance of the proposed test statistic $\check{Q}_{m-qr}$ for one-sided test $H_0:\beta_\tau=0$ vs $H_a^r:\beta_\tau> 0$ in univariate model are shown in Table \ref{sizeuni2}, in which $\tau=0.05,0.25,0.5,0.75,0.95$ and predictors are SD ($c=0$), SD ($c=-5$), SD ($c=-15$) and WD ($c=-0.05$).
We have the following findings from Tables \ref{sizeunil} and \ref{sizeuni2} and Figure \ref{poweruni1}. First, the size performance of the proposed test statistic ${Q}_{m-qr}$ surpasses that of L2016, FL2019 and CCL2023 and is slightly better than LLPY2023. The proposed test statistic ${Q}_{m-qr}$ is free of size distortion even at tail quantiles. Meanwhile, LLPY2023 tends to slightly over-reject when $\tau=0.05$ and under-reject when $\tau=0.95$. Second, the power performance of the proposed test statistic ${Q}_{m-qr}$ is better than L2016, FL2019 and CCL2023 while is comparable to LLPY2023. In scenarios SD ($c=0$) and SD ($c=-5$), the power of the proposed test statistic ${Q}_{m-qr}$ is lower than LLPY2023 when $\beta_\tau$ is small. Conversely, ${Q}_{m-qr}$ exhibits greater power than LLPY2023 when $\beta_\tau$ is large. For case SD ($c=-15$), the power of the proposed test statistic ${Q}_{m-qr}$ is the same as LLPY2023. For case WD ($c=-0.05$), the power of the proposed test statistic ${Q}_{m-qr}$ is greater than that of LLPY2023. Third, the size and the power performance of the proposed test statistic $\check{Q}_{m-qr}$ is very well for one-sided test.
In summary, the performance of the proposed test $\check{Q}_{m-qr}$ is comparable to LLPY2023 and superior to L2016, FL2019 and CCL2023 for two-sided test. Meanwhile, the proposed test $\check{Q}_{m-qr}$ performs well for one-sided test, which is not considered in the literature.
\begin{center}
\setlength{\tabcolsep}{8pt}
\renewcommand{\arraystretch}{0.6}
\captionsetup{
width=0.8\textwidth,
}
\begin{longtable}{c|c|cccc}
\caption{Size Performance (\%) of Two-sided Test in Univariate Model }\label{sizeunil}\\
\hline
$\tau$ & Method & SD (c=0) & SD (c=-5) & SD (c=-15) & WD (c=-0.05) \\
\hline
\endfirsthead
\hline
\multicolumn{6}{c}{{ \tablename\ \thetable{} -- continued from previous page}} \\
\hline
$\tau$ & Method & SD (c=0) & SD (c=-5) & SD (c=-15) & WD (c=-0.05) \\
\hline
\endhead
\hline
\multicolumn{6}{r}{{Continued on next page}} \\
\hline
\endfoot
\hline
\endlastfoot
\multirow{5}[0]{*}{0.05} & L2016 & 7.6 & 7.4 & 5.9 & 6.7 \\
& FL2019 & 6.8 & 6.5 & 6.0 & 6.1 \\
& CCL2023 & 7.1 & 7.5 & 6.7 & 6.9 \\
& LLPY2023 & 6.2 & 4.1 & 5.0 & 4.9 \\
& ${Q}_{m-qr}$ & 5.1 & 4.7 & 5.0 & 4.6 \\
\hline
\multirow{5}[0]{*}{0.25} & L2016 & 6.8 & 5.8 & 4.9 & 4.4 \\
& FL2019 & 8.1 & 6.4 & 6.9 & 5.2 \\
& CCL2023 & 5.1 & 5.3 & 4.4 & 5.5 \\
& LLPY2023 & 6.5 & 4.4 & 5.1 & 4.1 \\
& ${Q}_{m-qr}$ & 4.5 & 3.9 & 4.6 & 4.4 \\
\hline
\multirow{5}[0]{*}{0.5} & L2016 & 5.6 & 4.3 & 3.7 & 4.6 \\
& FL2019 & 7.9 & 6.5 & 6.0 & 6.8 \\
& CCL2023 & 4.8 & 4.9 & 4.7 & 5.0 \\
& LLPY2023 & 5.0 & 4.1 & 4.7 & 5.0 \\
& ${Q}_{m-qr}$ & 5.1 & 4.5 & 4.8 & 4.0 \\
\hline
\multirow{5}[0]{*}{0.75} & L2016 & 6.4 & 4.7 & 4.8 & 5.4 \\
& FL2019 & 8.0 & 6.6 & 6.2 & 6.5 \\
& CCL2023 & 5.1 & 4.6 & 4.8 & 5.2 \\
& LLPY2023 & 3.9 & 4.1 & 5.0 & 5.2 \\
& ${Q}_{m-qr}$ & 5.1 & 4.2 & 4.7 & 4.1 \\
\hline
\multirow{5}[0]{*}{0.95} & L2016 & 9.7 & 7.5 & 5.7 & 4.9 \\
& FL2019 & 8.4 & 7.3 & 6.2 & 5.2 \\
& CCL2023 & 6.5 & 6.9 & 6.0 & 6.7 \\
& LLPY2023 & 3.5 & 4.0 & 3.3 & 2.9 \\
& ${Q}_{m-qr}$ & 5.2 & 4.2 & 5.3 & 4.7 \\
\end{longtable}
\end{center}
Note that the black dash line in Figure \ref{poweruni1} represent 5\% nominal size. 
\begin{figure}[H]
\centering
\includegraphics[width= \textwidth]{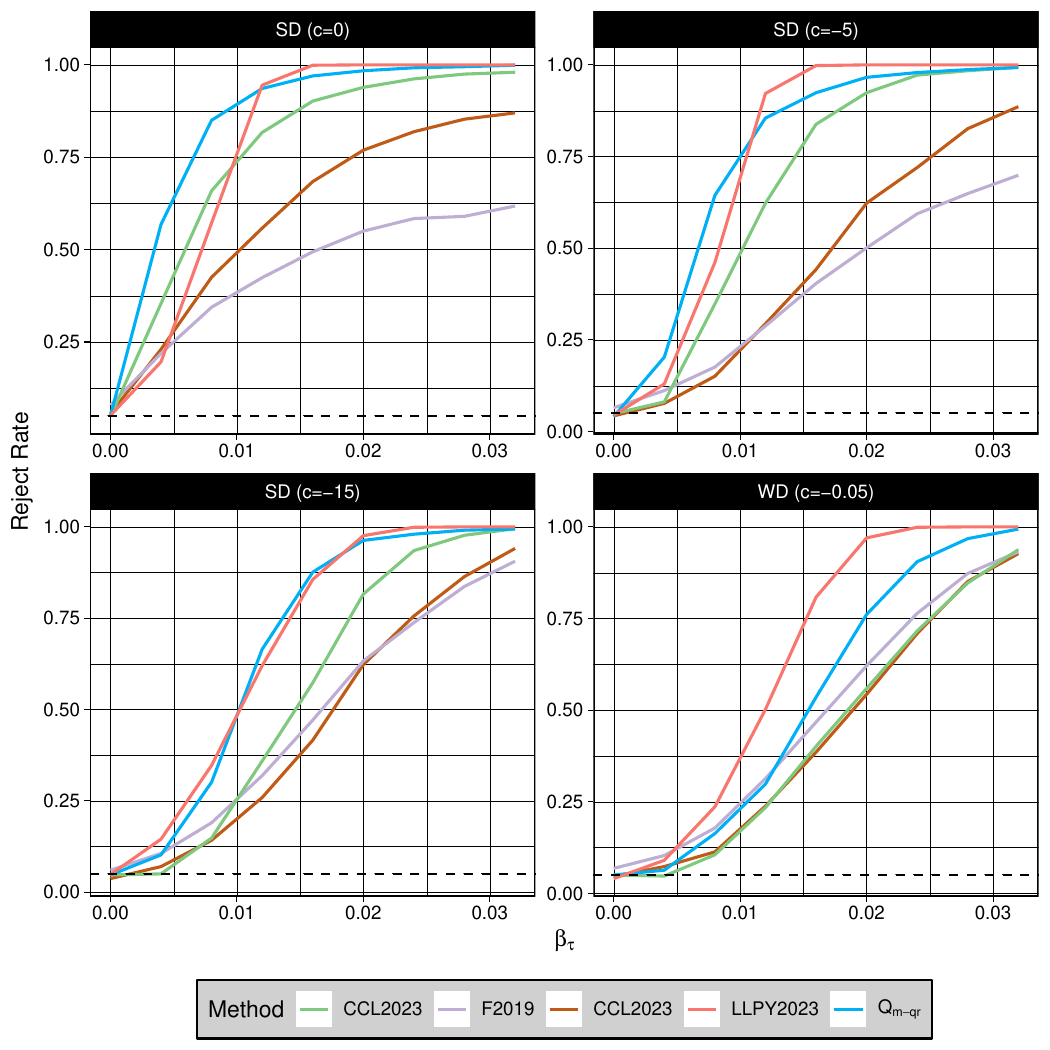}\\
\caption{Power Performance in the Univariate Model with $\tau=0.5$}\label{poweruni1}
\end{figure}
\begin{center}
\setlength{\tabcolsep}{8pt}
\renewcommand{\arraystretch}{0.6}
\captionsetup{
width=0.8\textwidth,
}
\begin{longtable}{c|l|ccccccc}
\caption{Right-sided Test Results (\%) of $\check{Q}_{m-qr}$ in the Univariate Model}\label{sizeuni2}\\
\hline
Persistence & $\beta_\tau$ & 0 & 0.004 & 0.008 & 0.012 & 0.016 & 0.02 & 0.024 \\
\hline
\endfirsthead
\hline
\multicolumn{9}{c}{{ \tablename\ \thetable{} -- continued from previous page}} \\
\hline
Persistence & $\beta_\tau$ & 0 & 0.004 & 0.008 & 0.012 & 0.016 & 0.02 & 0.024 \\
\hline
\endhead
\hline
\multicolumn{9}{r}{{Continued on next page}} \\
\hline
\endfoot
\hline
\endlastfoot
\multirow{5}[0]{*}{SD (c=0)} & $\tau$=0.05 & 5.7 & 28.5 & 74.1 & 94.3 & 99.4 & 99.9 & 100.0 \\
& $\tau$=0.25 & 6.5 & 63.4 & 98.9 & 100.0 & 100.0 & 100.0 & 100.0 \\
& $\tau$=0.5 & 6.7 & 74.1 & 99.7 & 100.0 & 100.0 & 100.0 & 100.0 \\
& $\tau$=0.75 & 6.2 & 63.0 & 98.8 & 100.0 & 100.0 & 100.0 & 100.0 \\
& $\tau$=0.95 & 5.2 & 27.9 & 73.0 & 94.5 & 98.9 & 99.9 & 100.0 \\
\hline
& $\tau$=0.05 & 4.7 & 23.2 & 70.5 & 93.4 & 99.1 & 99.9 & 100.0 \\
& $\tau$=0.25 & 4.6 & 59.5 & 98.6 & 100.0 & 100.0 & 100.0 & 100.0 \\
SD (c=-5) & $\tau$=0.5 & 4.9 & 70.5 & 99.7 & 100.0 & 100.0 & 100.0 & 100.0 \\
& $\tau$=0.75 & 4.6 & 59.3 & 98.7 & 100.0 & 100.0 & 100.0 & 100.0 \\
& $\tau$=0.95 & 4.4 & 22.7 & 69.5 & 93.3 & 98.8 & 99.9 & 100.0 \\
\hline
\multirow{5}[0]{*}{SD (c=-15)} & $\tau$=0.05 & 5.5 & 16.2 & 40.7 & 74.9 & 93.4 & 98.9 & 99.9 \\
& $\tau$=0.25 & 5.5 & 26.6 & 85.6 & 99.6 & 100.0 & 100.0 & 100.0 \\
& $\tau$=0.5 & 4.8 & 31.5 & 94.2 & 100.0 & 100.0 & 100.0 & 100.0 \\
& $\tau$=0.75 & 4.9 & 26.5 & 85.0 & 99.7 & 100.0 & 100.0 & 100.0 \\
& $\tau$=0.95 & 5.5 & 15.9 & 40.8 & 75.4 & 93.6 & 98.7 & 99.9 \\
\hline
\multirow{5}[0]{*}{WD (c=-0.05)} & $\tau$=0.05 & 4.8 & 10.0 & 32.3 & 70.0 & 91.8 & 98.5 & 99.8 \\
& $\tau$=0.25 & 4.2 & 18.5 & 82.2 & 99.5 & 100.0 & 100.0 & 100.0 \\
& $\tau$=0.5 & 4.1 & 21.6 & 92.4 & 100.0 & 100.0 & 100.0 & 100.0 \\
& $\tau$=0.75 & 4.6 & 17.9 & 82.0 & 99.6 & 100.0 & 100.0 & 100.0 \\
& $\tau$=0.95 & 4.5 & 10.1 & 33.0 & 70.4 & 91.8 & 98.4 & 99.9 \\
\end{longtable}
\end{center}
\noindent
\textbf{Experiment A2:}\quad First, Table \ref{CCLsmul2023} compares the size performance of the proposed test statistic $Q_{m-qr}$ and CCL2023 with K=8 for both the joint test $H_0:\beta_{\tau}=0$ and the two-sided marginal tests with $H_0:\beta_{i\tau}=0$, for $i=1,2,\cdots,8$. Second, the power performances of two-sided test results of $Q_{m-qr}$ are shown in Table \ref{mult_qm1}, in which K=8 and $\beta_{1\tau}=\beta_{2\tau}=\cdots=\beta_{K\tau}=\tilde{\beta}_\tau$ are set to 0, 0.02, 0.04, 0.06, 0.08 and 0.1, respectively. Third, we compare the size performances of L2016, CCL2023 and $Q_{m-qr}$ for the joint test $H_0:\beta_{\tau}=0$ with $\operatorname{K}=2,3,4,5,6,7,8$ in Table \ref{tdfbel}. Fourth, the size and the power performances of $\check{Q}_{m-qr}$ for right-sided test in Table \ref{mult_qmche}, in which K=8 and $\beta_{1\tau}=\beta_{2\tau}=\cdots=\beta_{K\tau}=\tilde{\beta}_\tau$ are set to 0, 0.02, 0.04, 0.06, 0.08 and 0.1, respectively. \footnote{We do not show the power performance \cite{FanLee2019} in experiment 2 since it does not offer the details of key procedure to construct empirical CDF in multivariate models. Also, We do not show the power performance \cite{LiuYangetal2023} in experiment 2 since the stationarity of predictors, i.e. the value of $\alpha$, must be known in its procedure but we skip the test of $\alpha$ here.} 
In these tables, we set $\tau=0.05,0.25,0.5,0.75,0.95$ and $(\rho_1,\rho_2,\cdots,\rho_{K})^\top=(\bm{\rho})_K$, where
$$\bm{\rho}= (0.996,0.993,1,0.987,0.967,0.95,0.9,0.98)^\top.$$ And $(\gamma_1,\gamma_2,\cdots,\gamma_K)^\top=(\Gamma)_K$, $\Gamma=(-3,2,1,3,1,-0.833,0.667,0.5)^\top$ and thus the contemporaneous correlation coefficient $\left[\gamma_1 (1+\gamma_1^2)^{-1/2},\cdots,\gamma_K (1+\gamma_K^2)^{-1/2} \right]^\top$ between $u_t$ and $v_t$ are $(\tilde{\Gamma})_K$ and
$ \tilde{\Gamma}=(-0.949,0.894,0.707,0.949,0.707,0.64,0.555,0.447)^\top$.
Several findings can be observed from Tables \ref{CCLsmul2023}, \ref{mult_qm1}, \ref{tdfbel} and \ref{mult_qmche}. First, for two-sided tests $H_0:\beta_{\tau}=0$ and $H_0:\beta_{i\tau}=0$, $i=1,2,\cdots,8$, with $K=8$, the proposed test statistic $Q_{m-qr}$ is free of size distortion, unlike the CCL2023, which exhibits size distortions at quantile levels $\tau=0.05,0.25,0.5,0.75,0.95$. Second, the power of the proposed test statistic $Q_{m-qr}$ is satisfactory for two-sided tests $H_0:\beta_{\tau}=0$ and $H_0:\beta_{i\tau}=0$, $i=1,2,\cdots,8$, with $K=8$. Third, at quatile levels $\tau=0.25, 0.5,0.75$, the size performance of L2016 and CCL2023 is acceptable when the number of predictors is less than 4 and 5, respectively, for joint test $H_0:\beta_\tau=0$. However, L2016 and CCL2023 suffer size distortion for other cases, such as the case when the number of predictors is less than 4 and 5 respectively and at tail quantiles $\tau=0.05,0.95$. This issue becomes more pronounced as the number of predictors increases. Fourth, the proposed test statistic $Q_{m-qr}$ outperforms L2016 and CCL2023 for joint test $H_0:\beta_\tau=0$, since it is free of size distortion for all cases. Fifth, the proposed test statistic $\check{Q}_{m-qr}$ perform well in terms of size and power for one-sided test, which is not considered in literature. 
To sum up, the size and power of the proposed test are comparable to those found in the existing literature for univariate models, and it significantly outperforms previous works in multivariate contexts.
\begin{center}
\setlength{\tabcolsep}{8pt}
\renewcommand{\arraystretch}{0.6}
\captionsetup{
width=0.8\textwidth,
}
\begin{longtable}{c|ccccccccc}
\caption{Size Performance (\%) of Two-sided Test for $H_0:\beta_{\tau}=0$ and $H_0:\beta_{i\tau}=0$ with K=8}\label{CCLsmul2023}\\
\hline
\multicolumn{10}{c}{Panel A: CCL2023}\\
\hline
$\tau$ & $H_0:\beta_{\tau}=0$ & i=1 & i=2 & i=3 & i=4 & i=5 & i=6 & i=7 & i=8 \\
\endfirsthead
\hline
\multicolumn{10}{c}{{ \tablename\ \thetable{} -- continued from previous page}} \\
\hline
$\tau$ & $H_0:\beta_{\tau}=0$ & i=1 & i=2 & i=3 & i=4 & i=5 & i=6 & i=7 & i=8 \\
\hline
\endhead
\hline
\multicolumn{8}{r}{{Continued on next page}} \\
\hline
\endfoot
\hline
\endlastfoot
\hline
$0.05$ & 26.2 & 12.8 & 11.9 & 12.2 & 12.2 & 11.6 & 12.1 & 11.6 & 11.8 \\
$0.25$ & 9.5 & 7.2 & 7.4 & 7.3 & 7.5 & 7.5 & 7.6 & 7.9 & 7.0 \\
$0.5$ & 8.6 & 7.6 & 7.5 & 8.2 & 7.4 & 8.0 & 7.5 & 7.4 & 7.6 \\
$0.75$ & 9.5 & 7.9 & 8.5 & 8.4 & 8.0 & 8.1 & 8.7 & 8.0 & 7.9 \\
$0.95$ & 26.6 & 13.0 & 12.0 & 12.2 & 13.2 & 12.2 & 11.9 & 12.4 & 12.3 \\
\hline
\multicolumn{10}{c}{Panel B: $Q_{m-qr}$ } \\
\hline
$\tau$ & $H_0:\beta_{\tau}=0$ & i=1 & i=2 & i=3 & i=4 & i=5 & i=6 & i=7 & i=8 \\
\hline
$0.05$ & 4.4 & 4.8 & 4.8 & 5.2 & 5.1 & 4.7 & 4.6 & 4.8 & 4.5 \\
$0.25$ & 4.2 & 4.9 & 5.2 & 4.6 & 6.0 & 5.5 & 4.8 & 4.5 & 4.7 \\
$0.5$ & 5.0 & 5.7 & 5.3 & 5.5 & 6.4 & 5.4 & 4.9 & 4.9 & 4.9 \\
$0.75$ & 4.4 & 5.4 & 5.0 & 5.5 & 6.0 & 5.0 & 4.9 & 4.5 & 4.7 \\
$0.95$ & 4.8 & 5.3 & 4.9 & 5.0 & 5.6 & 4.9 & 4.7 & 4.4 & 4.8 \\
\end{longtable}
\vspace{-0.5cm}
\begin{tablenotes}
\footnotesize
\item[1] Note: For different i, the null hypothesis for marginal test is $H_0:\beta_{i\tau}=0$.
\end{tablenotes}
\end{center}
\begin{center}
\setlength{\tabcolsep}{8pt}
\renewcommand{\arraystretch}{0.6}
\captionsetup{
width=0.8\textwidth,
}
\begin{longtable}{c|c|cccccc}
\caption{Power and Size Performance of Two-sided Test Results (\%) of $Q_{m-qr}$ in Multivariate Model with K=8.} \label{mult_qm1} \\
\hline
Null Hypothesis & $\tilde{\beta}_\tau$ & 0 & 0.02 & 0.04 & 0.06 & 0.08 & 0.1 \\
\hline
\endfirsthead
\hline
\multicolumn{8}{c}{{ \tablename\ \thetable{} -- continued from previous page}} \\
\hline
Null Hypothesis & $\tilde{\beta}_\tau$ & 0 & 0.02 & 0.04 & 0.06 & 0.08 & 0.1 \\
\hline
\endhead
\hline
\multicolumn{8}{r}{{Continued on next page}} \\
\hline
\endfoot
\hline
\endlastfoot
& $\tau$=0.05 & 4.4 & 97.9 & 100.0 & 100.0 & 100.0 & 100.0 \\
& $\tau$=0.25 & 4.2 & 100.0 & 100.0 & 100.0 & 100.0 & 100.0 \\
$H_0:\beta_{\tau}=0$ & $\tau$=0.5 & 5.0 & 100.0 & 100.0 & 100.0 & 100.0 & 100.0 \\
& $\tau$=0.75 & 4.4 & 99.9 & 100.0 & 100.0 & 100.0 & 100.0 \\
& $\tau$=0.95 & 4.8 & 98.0 & 100.0 & 100.0 & 100.0 & 100.0 \\
\hline
& $\tau$=0.05 & 4.8 & 26.1 & 74.8 & 95.1 & 99.5 & 100.0 \\
& $\tau$=0.25 & 4.9 & 54.9 & 96.1 & 99.9 & 100.0 & 100.0 \\
$H_0:\beta_{1\tau}=0$ & $\tau$=0.5 & 5.7 & 62.2 & 97.6 & 100.0 & 100.0 & 100.0 \\
& $\tau$=0.75 & 5.4 & 55.1 & 96.1 & 99.9 & 100.0 & 100.0 \\
& $\tau$=0.95 & 5.3 & 26.6 & 74.5 & 95.5 & 99.5 & 99.9 \\
\hline
& $\tau$=0.05 & 4.8 & 14.3 & 45.9 & 79.3 & 94.6 & 98.9 \\
& $\tau$=0.25 & 5.2 & 28.2 & 78.6 & 97.3 & 99.9 & 100.0 \\
$H_0:\beta_{2\tau}=0$& $\tau$=0.5 & 5.3 & 33.0 & 83.5 & 98.4 & 99.9 & 100.0 \\
& $\tau$=0.75 & 5.0 & 28.0 & 78.5 & 97.2 & 99.8 & 100.0 \\
& $\tau$=0.95 & 4.9 & 14.0 & 45.9 & 79.3 & 94.7 & 99.0 \\
\hline
& $\tau$=0.05 & 5.2 & 12.9 & 41.4 & 73.1 & 91.0 & 97.8 \\
& $\tau$=0.25 & 4.6 & 25.7 & 74.1 & 95.0 & 99.5 & 100.0 \\
$H_0:\beta_{3\tau}=0$ & $\tau$=0.5 & 5.5 & 30.0 & 79.6 & 97.0 & 99.7 & 100.0 \\
& $\tau$=0.75 & 5.5 & 25.5 & 73.0 & 95.2 & 99.6 & 100.0 \\
& $\tau$=0.95 & 5.0 & 12.9 & 41.6 & 73.1 & 91.6 & 97.9 \\
\hline
& $\tau$=0.05 & 5.1 & 52.0 & 88.8 & 98.4 & 99.7 & 100.0 \\
& $\tau$=0.25 & 6.0 & 74.5 & 97.9 & 99.9 & 100.0 & 100.0 \\
$H_0:\beta_{4\tau}=0$ & $\tau$=0.5 & 6.4 & 78.6 & 98.5 & 99.9 & 100.0 & 100.0 \\
& $\tau$=0.75 & 6.0 & 74.3 & 97.7 & 99.9 & 100.0 & 100.0 \\
& $\tau$=0.95 & 5.6 & 51.8 & 88.7 & 98.3 & 99.8 & 100.0 \\
\hline
& $\tau$=0.05 & 4.7 & 11.3 & 34.9 & 65.4 & 87.5 & 96.3 \\
& $\tau$=0.25 & 5.5 & 20.6 & 66.3 & 93.5 & 99.2 & 100.0 \\
$H_0:\beta_{5\tau}=0$& $\tau$=0.5 & 5.4 & 25.2 & 73.9 & 95.9 & 99.7 & 100.0 \\
& $\tau$=0.75 & 5.0 & 22.1 & 66.6 & 93.1 & 99.1 & 100.0 \\
& $\tau$=0.95 & 4.9 & 11.9 & 34.8 & 66.1 & 87.6 & 96.3 \\
\hline
& $\tau$=0.05 & 4.6 & 10.0 & 25.7 & 50.9 & 76.2 & 90.8 \\
& $\tau$=0.25 & 4.8 & 17.0 & 52.4 & 84.7 & 97.1 & 99.6 \\
$H_0:\beta_{6\tau}=0$ & $\tau$=0.5 & 4.9 & 19.2 & 59.0 & 89.4 & 98.7 & 99.8 \\
& $\tau$=0.75 & 4.9 & 16.2 & 52.0 & 85.8 & 97.0 & 99.6 \\
& $\tau$=0.95 & 4.7 & 9.6 & 26.3 & 51.3 & 76.2 & 90.6 \\
\hline
& $\tau$=0.05 & 4.8 & 7.3 & 17.1 & 32.8 & 52.1 & 70.6 \\
& $\tau$=0.25 & 4.5 & 11.3 & 33.5 & 61.9 & 85.1 & 96.0 \\
$H_0:\beta_{7\tau}=0$& $\tau$=0.5 & 4.9 & 12.8 & 37.7 & 68.7 & 90.3 & 98.0 \\
& $\tau$=0.75 & 4.5 & 11.1 & 32.4 & 61.7 & 85.1 & 95.8 \\
& $\tau$=0.95 & 4.4 & 7.8 & 17.0 & 32.4 & 51.5 & 70.5 \\
\hline
& $\tau$=0.05 & 4.5 & 8.1 & 18.1 & 35.5 & 56.9 & 76.0 \\
& $\tau$=0.25 & 4.7 & 12.3 & 36.9 & 68.2 & 89.5 & 97.7 \\
$H_0:\beta_{8\tau}=0$ & $\tau$=0.5 & 4.9 & 13.9 & 42.7 & 74.7 & 93.1 & 98.9 \\
& $\tau$=0.75 & 4.7 & 12.0 & 36.8 & 68.8 & 89.5 & 97.6 \\
& $\tau$=0.95 & 4.8 & 8.3 & 17.5 & 36.0 & 56.2 & 76.8 \\
\end{longtable}
\vspace{-0.5cm}
\begin{tablenotes}
\footnotesize
\item[1] $\beta_{1\tau}=\beta_{2\tau}=\cdots=\beta_{K\tau}=\tilde{\beta}_\tau$.
\end{tablenotes}
\end{center}
\begin{center}
\setlength{\tabcolsep}{8pt}
\renewcommand{\arraystretch}{0.6}
\captionsetup{
width=0.8\textwidth,
}
\begin{longtable}{c|ccccccc}
\caption{Size Performance of L2016, CCL2023 and $Q_{m-qr}$ for Joint Test $H_0:\beta_{\tau}=0$} \label{tdfbel}\\
\hline
\multicolumn{8}{c}{Panel A: L2016 }\\
\hline
$\tau$ & K=2 & K=3 & K=4 & K=5 & K=6 & K=7 & K=8 \\ 
\endfirsthead
\hline
\multicolumn{8}{c}{{ \tablename\ \thetable{} -- continued from previous page}} \\
\hline
$\tau$ & K=2 & K=3 & K=4 & K=5 & K=6 & K=7 & K=8 \\ 
\hline
\endhead
\hline
\multicolumn{8}{r}{{Continued on next page}} \\
\hline
\endfoot
\hline
\endlastfoot
\hline
$0.05$ & 8.0 & 10.6 & 18.6 & 22.6 & 25.6 & 29.4 & 34.2 \\
$0.25$ & 5.3 & 5.9 & 13.0 & 13.2 & 15.3 & 15.1 & 16.8 \\
$0.5$ & 4.6 & 5.3 & 12.2 & 13.1 & 14.0 & 14.0 & 14.6 \\
$0.75$ & 4.8 & 6.2 & 12.7 & 13.5 & 15.3 & 15.7 & 15.1 \\
$0.95$ & 8.4 & 10.7 & 18.5 & 22.9 & 26.1 & 28.7 & 33.1 \\
\hline
\multicolumn{8}{c}{Panel B: CCL2023 }\\
\hline
$\tau$ & K=2 & K=3 & K=4 & K=5 & K=6 & K=7 & K=8 \\
\hline
$0.05$ & 8.1 & 11.3 & 11.8 & 15.8 & 18.6 & 21.3 & 25.7 \\
$0.25$ & 5.0 & 7.2 & 6.8 & 6.8 & 8.1 & 8.1 & 10.2 \\
$0.5$ & 4.3 & 6.2 & 5.2 & 7.0 & 7.4 & 7.4 & 8.8 \\
$0.75$ & 5.2 & 6.9 & 4.4 & 8.5 & 7.0 & 9.9 & 9.9 \\
$0.95$ & 8.9 & 10.2 & 12.2 & 16.1 & 19.6 & 23.5 & 26.9 \\
\hline
\multicolumn{8}{c}{Panel C: $Q_{m-qr}$ }\\
\hline
$\tau$ & K=2 & K=3 & K=4 & K=5 & K=6 & K=7 & K=8 \\
\hline
$0.05$ & 4.3 & 4.6 & 4.6 & 5.4 & 5.2 & 4.9 & 4.6 \\
$0.25$ & 4.7 & 5.1 & 5.3 & 5.2 & 5.0 & 4.0 & 5.2 \\
$0.5$ & 4.9 & 4.8 & 5.1 & 5.5 & 5.3 & 5.0 & 5.1 \\
$0.75$ & 4.7 & 4.8 & 4.4 & 4.5 & 4.5 & 5.0 & 4.9 \\
$0.95$ & 4.8 & 5.0 & 4.8 & 4.5 & 4.2 & 4.8 & 4.7 \\
\end{longtable}
\end{center}
\begin{center}
\setlength{\tabcolsep}{8pt}
\renewcommand{\arraystretch}{0.6}
\begin{longtable}{c|c|cccccc}
\caption{Right-sided Test Results (\%) of $\check{Q}_{m-qr}$ in Multivariate Model with K=8.} \label{mult_qmche} \\
\hline
Null Hypothesis & $\tilde{\beta}_\tau$ & 0 & 0.02 & 0.04 & 0.06 & 0.08 & 0.1 \\
\hline
\endfirsthead
\hline
\multicolumn{8}{c}{{ \tablename\ \thetable{} -- continued from previous page}} \\
\hline
Null Hypothesis & $\tilde{\beta}_\tau$ & 0 & 0.02 & 0.04 & 0.06 & 0.08 & 0.1 \\
\hline
\endhead
\hline
\multicolumn{8}{r}{{Continued on next page}} \\
\hline
\endfoot
\hline
\endlastfoot
\multirow{5}[0]{*}{$H_0:\beta_{1\tau}=0$} & $\tau$=0.05 & 5.4 & 33.7 & 79.3 & 96.3 & 99.6 & 100.0 \\
& $\tau$=0.25 & 5.3 & 61.9 & 96.8 & 99.9 & 100.0 & 100.0 \\
& $\tau$=0.5 & 6.0 & 67.7 & 98.2 & 100.0 & 100.0 & 100.0 \\
& $\tau$=0.75 & 6.2 & 61.8 & 97.0 & 99.9 & 100.0 & 100.0 \\
& $\tau$=0.95 & 5.6 & 34.2 & 79.0 & 96.4 & 99.7 & 100.0 \\
\hline
\multirow{5}[0]{*}{$H_0:\beta_{2\tau}=0$} & $\tau$=0.05 & 4.6 & 20.8 & 54.6 & 83.9 & 95.9 & 99.1 \\
& $\tau$=0.25 & 4.5 & 36.5 & 83.3 & 98.1 & 99.9 & 100.0 \\
& $\tau$=0.5 & 4.7 & 41.3 & 87.1 & 98.8 & 99.9 & 100.0 \\
& $\tau$=0.75 & 4.5 & 36.6 & 83.2 & 98.0 & 99.8 & 100.0 \\
& $\tau$=0.95 & 4.3 & 20.7 & 54.9 & 83.7 & 96.1 & 99.3 \\
\hline
\multirow{5}[0]{*}{$H_0:\beta_{3\tau}=0$} & $\tau$=0.05 & 4.9 & 19.2 & 49.2 & 77.9 & 92.8 & 98.4 \\
& $\tau$=0.25 & 4.6 & 33.5 & 78.7 & 96.3 & 99.7 & 100.0 \\
& $\tau$=0.5 & 5.1 & 37.9 & 83.9 & 97.6 & 99.8 & 100.0 \\
& $\tau$=0.75 & 5.2 & 33.2 & 78.2 & 96.3 & 99.7 & 100.0 \\
& $\tau$=0.95 & 5.1 & 19.3 & 49.9 & 78.3 & 93.5 & 98.4 \\
\hline
\multirow{5}[0]{*}{$H_0:\beta_{4\tau}=0$} & $\tau$=0.05 & 4.3 & 56.6 & 90.7 & 98.6 & 99.7 & 100.0 \\
& $\tau$=0.25 & 3.9 & 77.9 & 98.3 & 99.9 & 100.0 & 100.0 \\
& $\tau$=0.5 & 4.1 & 81.4 & 98.8 & 100.0 & 100.0 & 100.0 \\
& $\tau$=0.75 & 4.2 & 77.5 & 98.2 & 99.9 & 100.0 & 100.0 \\
& $\tau$=0.95 & 4.2 & 57.0 & 90.3 & 98.7 & 99.8 & 100.0 \\
\hline
\multirow{5}[0]{*}{$H_0:\beta_{5\tau}=0$} & $\tau$=0.05 & 4.5 & 16.8 & 44.1 & 72.3 & 90.3 & 97.3 \\
& $\tau$=0.25 & 5.0 & 28.9 & 73.0 & 95.2 & 99.5 & 100.0 \\
& $\tau$=0.5 & 5.0 & 33.7 & 79.5 & 97.1 & 99.8 & 100.0 \\
& $\tau$=0.75 & 4.8 & 30.3 & 73.1 & 95.0 & 99.4 & 100.0 \\
& $\tau$=0.95 & 4.6 & 18.1 & 43.6 & 72.8 & 90.5 & 97.3 \\
\hline
\multirow{5}[0]{*}{$H_0:\beta_{6\tau}=0$} & $\tau$=0.05 & 4.3 & 15.9 & 35.6 & 61.1 & 82.7 & 93.6 \\
& $\tau$=0.25 & 4.4 & 25.0 & 61.8 & 88.9 & 98.0 & 99.7 \\
& $\tau$=0.5 & 4.4 & 27.6 & 67.9 & 92.7 & 99.0 & 99.8 \\
& $\tau$=0.75 & 4.4 & 24.7 & 61.2 & 89.8 & 97.9 & 99.8 \\
& $\tau$=0.95 & 4.9 & 15.2 & 35.9 & 60.9 & 82.3 & 93.3 \\
\hline
& $\tau$=0.05 & 4.6 & 11.7 & 25.8 & 44.0 & 63.1 & 79.1 \\
& $\tau$=0.25 & 4.3 & 18.1 & 43.8 & 71.6 & 89.8 & 97.4 \\
$H_0:\beta_{7\tau}=0$ & $\tau$=0.5 & 4.3 & 20.0 & 49.0 & 77.3 & 93.6 & 98.8 \\
& $\tau$=0.75 & 4.4 & 17.7 & 43.9 & 71.2 & 90.0 & 97.5 \\
& $\tau$=0.95 & 4.1 & 12.3 & 25.3 & 43.2 & 62.2 & 78.6 \\
\hline
\multirow{5}[0]{*}{$H_0:\beta_{8\tau}=0$} & $\tau$=0.05 & 4.5 & 13.7 & 27.2 & 46.4 & 66.9 & 82.9 \\
& $\tau$=0.25 & 4.6 & 19.3 & 48.2 & 76.8 & 93.0 & 98.5 \\
& $\tau$=0.5 & 4.6 & 21.9 & 53.7 & 81.9 & 95.6 & 99.2 \\
& $\tau$=0.75 & 4.7 & 19.7 & 48.2 & 77.2 & 93.1 & 98.4 \\
& $\tau$=0.95 & 4.9 & 13.3 & 26.1 & 46.9 & 66.9 & 83.2 \\
\end{longtable}
\end{center}
\section{Additional Theoretical Results}\label{app:C}
\subsection{Eliminate One Higher-order Term by Sample Splitting Method}
We eliminate the first one of two higher-order terms by applying a new instrumental variable which is constructed by the sample splitting method and the instrumental variable $z_{t-1}$.
The size distortion induced by the higher-order terms ${C_T}$ in Proposition \ref{mulpropfdie3} could be eliminated by removing the term $\frac{1}{T} \sum\nolimits_{t=1}^{T} z_{t-1}$ in the test statistics. We split the full sample into two sub-samples $\{(y_t,x_{t-1})\}_{t=1}^{T_0}$ and $\{(y_t,x_{t-1})\}_{t=T_0}^{T}$, where $T_0=\lfloor \lambda T\rfloor$ and $0<\lambda<1$ is a constant. In practice, we follow \cite{liao2024robust} to evenly split the full sample from the middle time point, i.e., $\lambda=0.5$.
First, we apply the two-step regression introduced in subsection \ref{subsection3.1} of the main text to the two sub-samples. The first step based on the fist sub-sample is as follows.
\begin{align}\label{firstep1a}
(\hat{\mu}_x^a,\hat{\theta}^a) = \arg \; \min_{\mu_x,\theta} \sum_{t=1}^{T_0} \left(x_{t-1}- \mu_x - \theta z_{t-1} \right)^\top \left(x_{t-1}- \mu_x - \theta z_{t-1} \right).
\end{align}
The fitted value of equation (\ref{firstep1a}) is $\tilde{x}_{t-1}^a =\hat{\mu}_x^a+\hat{\theta}^a z_{t-1}$, and the residual is $\tilde{v}_{t-1}^a = x_{t-1}- \tilde{x}_{t-1}^a$ and $t=1,2,\cdots,T_0$.
The second step based on the first sub-sample is as follows.
\begin{align}\label{firstep2a}
\left[\hat\mu_\tau^a,(\hat\beta_\tau^a)^\top,(\hat\gamma_\tau^a)^\top\right]^\top =\arg \,\min_{\mu_\tau, \beta_\tau,\gamma_\tau} \sum_{t=1}^{T_0} \rho_{\tau}\left(y_{t}-\mu_\tau-\beta_\tau^\top \tilde{x}_{t-1}^a -\gamma_\tau^\top \tilde{v}_{t-1}^a\right).
\end{align}
Additionally, the first step based on the second sub-sample is as follows.
\begin{align}\label{firstep1b}
(\hat{\mu}_x^b,\hat{\theta}^b) = \arg \; \min_{\mu_x,\theta} \sum_{t=T_0+1}^T \left(x_{t-1}- \mu_x - \theta z_{t-1} \right)^\top \left(x_{t-1}- \mu_x - \theta z_{t-1} \right).
\end{align}
The fitted value of equation (\ref{firstep1b}) is $\tilde{x}_{t-1}^b = \hat{\mu}_x^b + \hat{\theta}^b z_{t-1}$, and the residual is $\tilde{v}_{t-1}^b = x_{t-1}- \tilde{x}_{t-1}^b$ and $t=1,2,\cdots,T_0$.
The second step based on the second sub-sample is as follows.
\begin{align}\label{firstep2b}
\left[\hat\mu_\tau^b,(\hat\beta_\tau^b)^\top,(\hat\gamma_\tau^b)^\top\right]^\top =\arg \,\min_{\mu_\tau, \beta_\tau,\gamma_\tau} \sum_{t=T_0+1}^T \rho_{\tau}\left(y_{t}-\mu_\tau-\beta_\tau^\top \tilde{x}_{t-1}^b -\gamma_\tau^\top \tilde{v}_{t-1}^b\right).
\end{align}
Similar to Theorem \ref{thm2}, the following equations hold.
\begin{align}\label{muldef2new}
D_T\left(\hat{\beta}_\tau^a -\beta_\tau\right)&=\frac{1}{f_{u_\tau}(0)} \left[ D_T^{-2} \sum\limits_{t=1}^{T_0} \bar{z}_{t-1}^a x_{t-1}^\top \right]^{-1} D_T^{-1} \sum\limits_{t=1}^{T_0} \bar{z}_{t-1}^a \psi_\tau (u_{t\tau}) + o_p(1);\\
\label{muldef3new}
D_T \left(\hat{\beta}_\tau^b -\beta_\tau\right)&=\frac{1}{f_{u_\tau}(0)} \left[ D_T^{-2} \sum\limits_{t=T_0+1}^T \bar{z}_{t-1}^b x_{t-1}^\top \right]^{-1} D_T^{-1} \sum\limits_{t=T_0+1}^T \bar{z}_{t-1}^b \psi_\tau (u_{t\tau}) + o_p(1),
\end{align}
where $\bar{z}_{t-1}^a= z_{t-1}- \frac{1}{T_0}\sum\nolimits_{t=1}^{T_0} z_{t-1}$ and $\bar{z}_{t-1}^b= z_{t-1}- \frac{1}{T-T_0}\sum\nolimits_{t=T_0+1}^{T} z_{t-1}$.
Second, we apply equations (\ref{eq1}), (\ref{mulgtuA1}), (\ref{muldef2new}) and (\ref{muldef3new}) to obtain the following equations. 
\begin{align}\label{muldeftwo2}
& \sum\limits_{t=1}^{T} \bar{z}_{t-1} x_{t-1}^\top (\hat{\beta}_\tau-\beta) = \frac{1}{f_{u_\tau}(0)}\sum\limits_{t=1}^{T} z_{t-1}\psi_\tau (u_{t\tau}) - \frac{\frac{1}{T}\sum\limits_{t=1}^{T} z_{t-1} }{f_{u_\tau}(0)}\sum\limits_{t=1}^{T} \psi_\tau (u_{t\tau}) +o_p(D_T^{-1}), \\
\label{muldeftwo3}
& \sum\limits_{t=1}^{T_0} \bar{z}_{t-1}^a x_{t-1}^\top (\hat{\beta}_\tau^a-\beta) = \frac{1}{f_{u_\tau}(0)}\sum\limits_{t=1}^{T_0} z_{t-1} \psi_\tau (u_{t\tau})- \frac{ \frac{1}{T_0}\sum\limits_{t=1}^{T_0} z_{t-1}}{f_{u_\tau}(0)} \sum\limits_{t=1}^{T_0} \psi_\tau (u_{t\tau})+o_p(D_T^{-1}) , \\ \label{muldeftwo4}
& \sum\limits_{t=T_0+1}^{T} \bar{z}_{t-1}^b x_{t-1}^\top (\hat{\beta}_\tau^b-\beta) = \frac{\sum\limits_{t=T_0+1}^{T} z_{t-1} \psi_\tau (u_{t\tau})}{f_{u_\tau}(0)} - \frac{\frac{\sum\limits_{t=T_0+1}^{T} z_{t-1} }{T-T_0}}{f_{u_\tau}(0)} \sum\limits_{t=T_0+1}^{T} \psi_\tau (u_{t\tau})+o_p(D_T^{-1}). 
\end{align} 
By equations (\ref{muldeftwo3}) and (\ref{muldeftwo4}), the following equations hold.
\begin{small}
\begin{align}
\label{mul8new}
&S_a \sum\limits_{t=1}^{T_0} \bar{z}_{t-1}^a{x}_{t-1}^\top (\hat{\beta}_\tau^a-\beta) = \frac{S_a}{f_{u_\tau}(0)} \sum\limits_{t=1}^{T_0} z_{t-1} \psi_\tau (u_{t\tau})- \frac{\frac{1}{T}\sum\limits_{t=1}^{T} z_{t-1}}{ f_{u_\tau}(0)} \sum\limits_{t=1}^{T_0} \psi_\tau (u_{t\tau}) +o_p(D_T^{-1}) , \\ \label{mul9new}
&S_b\sum\limits_{t=1}^{T_0} \bar{z}_{t-1}^b{x}_{t-1}^\top (\hat{\beta}_\tau^b-\beta) = \frac{S_b}{f_{u_\tau}(0)} \sum\limits_{t=T_0+1}^{T} z_{t-1} \psi_\tau (u_{t\tau}) -
\frac{\frac{1}{T}\sum\limits_{t=1}^{T} z_{t-1} }{f_{u_\tau}(0)}\sum\limits_{t=T_0+1}^{T} \psi_\tau (u_{t\tau})+o_p(D_T^{-1}). 
\end{align} 
\end{small}
By subtracting the sum of equations (\ref{mul8new}) and (\ref{mul9new}) from equation (\ref{muldeftwo2}), the term $ \sum\nolimits_{t=1}^{T} \psi_\tau (u_{t\tau})$ of $\hat{\beta}_\tau$ is removed as follows. 
\begin{align}\label{mulkde2con}
& W_1 \hat{\beta}_\tau-W_2\hat{\beta}_\tau^a -W_3 \hat{\beta}_\tau^b - (W_1-W_2-W_3)\beta \\
& = \frac{1}{f_{u_\tau}(0)}(\operatorname{I_K}-S_a) \sum\limits_{t=1}^{T_0} z_{t-1}\psi_\tau (u_{t\tau}) + \frac{1}{f_{u_\tau}(0)} (\operatorname{I_K}-S_b) \sum\limits_{t=T_0+1}^{T} z_{t-1}\psi_\tau (u_{t\tau}) + o_p(D_T^{-1}) \nonumber\\
&= \frac{1}{f_{u_\tau}(0)}\sum\limits_{t=1}^{T} \tilde{z}_{t-1}\psi_\tau (u_{t\tau})+ o_p(D_T^{-1}), \nonumber
\end{align}
where $W_1=\sum\nolimits_{t=1}^{T} \bar{z}_{t-1} x_{t-1}^\top $, $W_2=S_a \sum\nolimits_{t=1}^{T_0} \bar{z}_{t-1}^a x_{t-1}^\top$,
$W_3=S_b \sum\nolimits_{t=T_0+1}^{T} \bar{z}_{t-1}^b x_{t-1}^\top$.
Third, we could utilize another instrumental variable estimator with the IV $\tilde{z}_{t-1}$. Define the IV estimator $\hat{\beta}_\tau^{l_0} $ as follows.
\begin{align}\label{mulkde2con2}
\hat{\beta}_\tau^{l_0} \equiv (W_1-W_2-W_3)^{-1}(W_1 \hat{\beta}_\tau-W_2\hat{\beta}_\tau^a -W_3 \hat{\beta}_\tau^b).
\end{align}
By equations (\ref{mulkde2con}) and (\ref{mulkde2con2}) and the fact $W_1-W_2-W_3 = \sum\limits_{t=1}^{T} \tilde{z}_{t-1}x_{t-1}$, it follows that
\begin{align}\label{kdad53g}
\hat{\beta}_\tau^{l_0}- \beta &=\frac{1}{f_{u_\tau}(0)} \left(\sum\limits_{t=1}^{T} \tilde{z}_{t-1}x_{t-1}^\top \right)^{-1} \sum\limits_{t=1}^{T} \tilde{z}_{t-1} \psi_\tau (u_{t\tau})+ o_p(D_T^{-1}).
\end{align}
The key and desirable property of the new instrumental variable $\tilde{z}_{t-1}$ is
$\sum\nolimits_{t=1}^{T} \tilde{z}_{t-1}=0$.
Thus it follows that
\begin{align*}
&\hat{\beta}_\tau^{l_0} - \beta \\
&= \frac{1}{f_{u_\tau}(0)}\left[\sum\limits_{t=1}^{T} \left( \tilde{z}_{t-1}-\frac{1}{T} \sum\limits_{t=1}^{T} \tilde{z}_{t-1} \right)x_{t-1}^\top \right]^{-1} \sum\limits_{t=1}^{T} \left( \tilde{z}_{t-1}-\frac{1}{T} \sum\limits_{t=1}^{T} \tilde{z}_{t-1} \right) \psi_\tau (u_{t\tau}) + o_p(D_T^{-1}) \\
&= \frac{1}{f_{u_\tau}(0)} \left(\sum\limits_{t=1}^{T} \tilde{z}_{t-1}x_{t-1}^\top \right)^{-1} \sum\limits_{t=1}^{T} \tilde{z}_{t-1} \psi_\tau (u_{t\tau})+ o_p(D_T^{-1}).
\end{align*}
By this way, the higher-order term ${C_T}$ vanishes.
\subsubsection*{Higher-order Term of $\check{Q}_{l-qr}$ and $Q_{l-qr}$}\label{hiorte1}
Although the higher-order term $C_T$ disappear in $\check{Q}_{l-qr}$ and $Q_{l-qr}$, the higher-order term $B_T^l$ which arise by the same reason of $B_T$ still exists.
To analyze these distortion effects of $\check{Q}_{l-qr}$ and $Q_{l-qr}$ intuitively, we show their higher-order terms with $J=K=1$ here. Following the Proposition 2 of \cite{liao2024robust}, we have the proposition as follows.
\begin{prop}\label{mulpropp2}
Under Assumptions \ref{Assumption A.1} and \ref{Assumption A.2} and the null hypothesis $H_0:\beta_\tau=0$ and $J=K=1$, for SD predictors, the following equation holds as $T \rightarrow \infty$ that
$$
\check{Q}_{l-qr}= Z_T^l+{B_T^l}+o_p\left(T^{\delta / 2-1 / 2}\right),
$$
where
$
{Z_T^l} = (\Sigma_{zz})^{-1/2}\sum\nolimits_{t=1}^T \tilde{z}_{t-1} \psi_\tau (u_{t\tau}) \stackrel{d}{\rightarrow} \operatorname{N}(0,1) \quad \text { with } 
$
${B_T^l }\rightarrow 0 \quad$ and
\begin{align}\label{mulpop1th12}
{B_T^l} = \varpi_l{Z_T^l},\;\varpi_l = -\frac{1}{2}\left\{ \tau(1-\tau)\sum\nolimits_{t=1}^T \tilde{z}_{t-1}^* (\tilde{z}_{t-1}^*)^\top -1\right\}
\end{align}
where $\tilde{z}_{t-1}^*= (\Sigma_{zz})^{-1/2}\tilde{z}_{t-1}$ and
\begin{align}\label{mulpotth3qhe1}
T^{(1 -\delta) / 2}{B_T^l} = R_T^l - W_l \rho_{ v\psi} / \sqrt{-2 c_z} 
\end{align}
and ${R_T^l}= W_a{R_{1,T}^l}+W_b{ R_{2,T}^l}$,
$
\operatorname{E}\left({R_{1,T}^l} \right) = \operatorname{E}\left({R_{2,T}^l} \right) =0
$
and
$
W_l = W_a/\sqrt{\lambda} + W_b/\sqrt{1-\lambda}
$ and
$
W_a = [ T^{-(1+\delta)} \sum_{t=1}^{T} \tilde{z}_{t-1} \tilde{z}_{t-1}^\top \hat{u}_t^2 ]^{-1/2}
(1-S_a)
[ T^{-(1+\delta)} \sum_{t=1}^{T_0} z_{t-1} z_{t-1}^\top \hat{u}_t^2 ]^{1/2}
$
and 
$
W_b = [ T^{-(1+\delta)} $ \quad $ \sum_{t=1}^{T} \tilde{z}_{t-1} \tilde{z}_{t-1}^\top \hat{u}_t^2 ]^{-1/2}
(1-S_b)
[ T^{-(1+\delta)} \sum_{t=T_0+1}^{T} z_{t-1} z_{t-1}^\top \hat{u}_t^2 ]^{1/2}
$. More details about $R_{1,T}^l$ and $R_{2,T}^l$ are discussed in Proposition 2 of \cite{liao2024robust}. 
\end{prop}
\begin{remark}\label{red8dma}
As specified by \cite{liao2024robust}, the term ${R_T^l}$ could be regarded as the ``residual term" of $B_T^l$ and thus only $-W_l \rho_{ v\psi} / \sqrt{-2 c_z}$ is the main source of size distortion by $B_T^l$ for $\check{Q}_{l-qr}$. So it is evident that the bigger $|c_z|=-c_z$ is, the smaller the size distortion by $B_T^l$ is.
\end{remark}
\section{Additional Empirical Results}\label{app:D}
\subsection{Empirical Results for One-sided Test}\label{onesidetest1}
We set $\check{Q}_{m-qr}^t = \operatorname{sign}(\check{Q}_{m-qr}) \, \max(|\check{Q}_{m-qr}|,4)$ to avoid the absolute value of $\check{Q}_{m-qr}$ being too large to distort the scale of the figure. The significance induced by $\check{Q}_{m-qr}^t$ is the same as that of $\check{Q}_{m-qr}$. However, the graphical representation in its figure distinctly indicates whether $\check{Q}_{m-qr}^t$ is positive or negative.
\begin{figure}[H]
\centering
\includegraphics[width=1\linewidth]{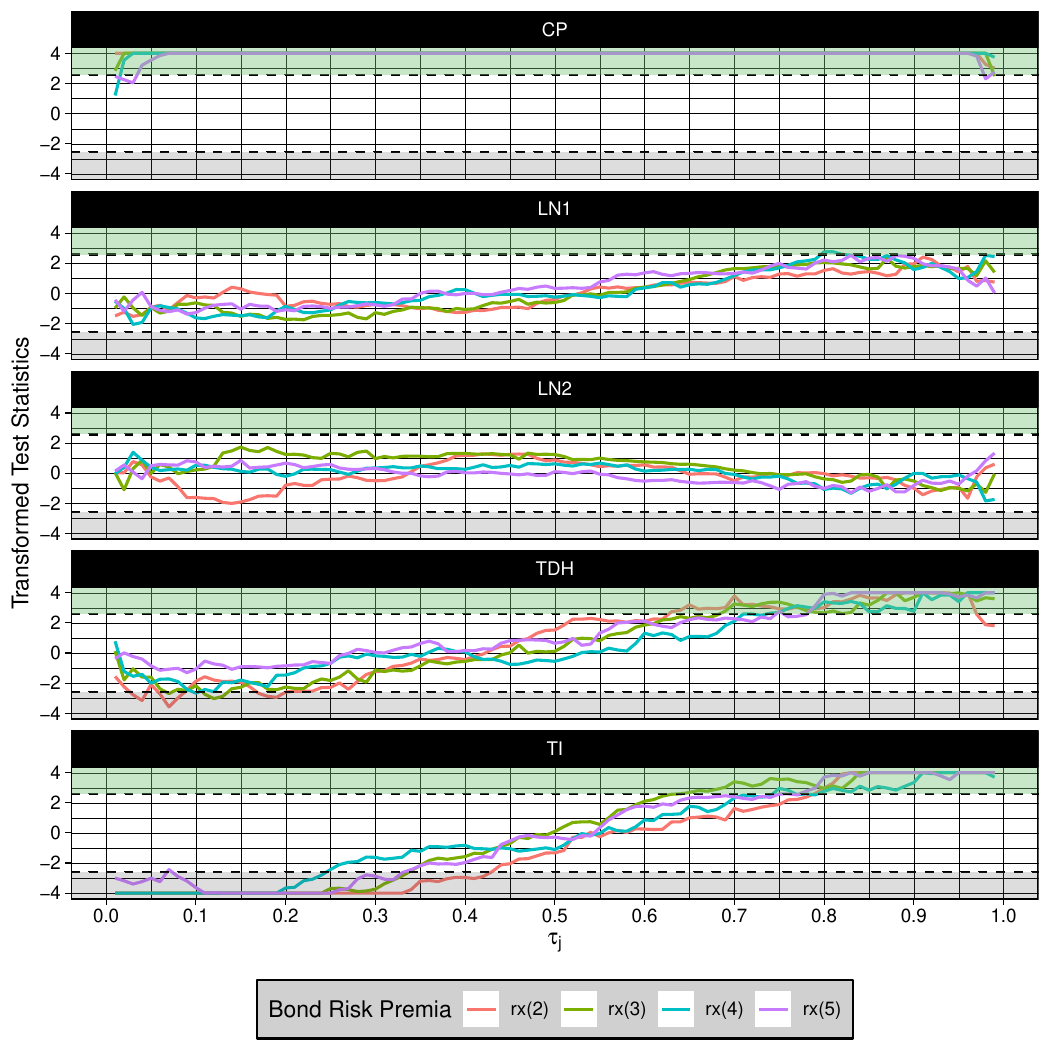}%
\caption{Transformed Test Statistic $\check{Q}_{m-qr}^t$ for One-sided Test}\label{onesidetest}
\end{figure} 
In Figure \ref{onesidetest}, the test statistic $\check{Q}_{m-qr}^t$is significantly positive at the 0.01 level in the light green area, while it is significantly negative at level 0.01 in the gray area. These results are consistent with those of $Q_{m-qr}$ and the economic explanations in subsection \ref{poexpr} of the main text. 
\subsection{Tail Risk Indicator for U.S. Treasury Bonds}\label{tairisin1}
to obtain the real time coefficient estimator $ 
\big[\hat{\dot\mu}_{\tau_j}(\tilde{T}), \hat{\dot\beta}_{\tau_j}(\tilde{T})^\top \big]^\top =\arg \,\min_{\mu_{\tau_j}, \beta_{\tau_j}} \sum_{t=1}^{\tilde{T}-1} \rho_{\tau_j} (y_{t}-\mu_{\tau_j}-\dot\beta_{\tau_j}^\top \dot{x}_{t-1}^j )$, where $\dot{x}_{t-1}^j$ is the one used in (\ref{cer1ev}). 
Then the out-of-sample predicted value of bond risk premia at quantiles $\tau_j$ is $\hat{y}_{\tau_j}^o(\tilde{T})=\hat{\dot\mu}_{\tau_j}(\tilde{T})+ \hat{\dot\beta}_{\tau_j}(\tilde{T})^\top \dot{x}_{t-1}^j$ and the out-of-sample quantile regression error is $\hat{u}_{\tau_j}^o (\tilde{T})=y_{\tilde{T}}-\hat{y}_{\tau_j}^o (\tilde{T})$. The out-of-sample performance in the period from $T_0$ to T is evaluated by qw-CRPS. 
Using qw-CRPS and the out-of-sample prediction results, we construct the tail-risk indicator by the following procedure. The basic idea to construct tail risk indicator is as follows. Since $W\left(\tau_j\right) >0$, then qw-CRPS could be written as
\begin{align}\label{djjhf1}
\operatorname{qw}-\operatorname{CRPS}_{\tilde{T}} =\frac{2}{J-1} \sum_{j=1}^{J-1} W\left(\tau_j\right) \rho_{\tau_j}[\hat{u}_{\tau_j}^o(\tilde{T})] = \frac{2}{J-1} \sum_{j=1}^{J-1} \rho_{\tau_j} [\hat{u}_{\tau_j}^w (\tilde{T})], 
\end{align}
where $\hat{u}_{\tau_j}^w (\tilde{T})= W\left(\tau_j\right) \hat{u}_{\tau_j}^o (\tilde{T})$. Moreover, since $y_{\tilde{T}} =\hat{y}_{\tau_j}^o(\tilde{T}) + \hat{u}_{\tau_j}^o (\tilde{T})$, we have $y_{\tilde{T}}^w = \hat{y}_{\tau_j}^w(\tilde{T}) + \hat{u}_{\tau_j}^w(\tilde{T})$,
where $j=1,2,\cdots J$, $y_{\tilde{T}}^w= W\left(\tau_j\right) y_{\tilde{T}}$ and $\hat{y}_{\tau_j}^w(\tilde{T})= W\left(\tau_j\right) \hat{y}_{\tau_j}^o(\tilde{T})$. As per equation (\ref{djjhf1}), $\operatorname{qw}-\operatorname{CRPS}_{\tilde{T}}$ could be viewed as the quantile loss function for the aforementioned weighted quantile regression at quantile $\tau_j$, where $j=1,2,\cdots,J$. Therefore, it is reasonable to use the out-of-sample prediction value of weighted quantile regression to construct the left and the right tail risk indicators for U.S. treasury bonds as 
\begin{align*}
\operatorname{R(tail)_{\tilde{T}}} = \sum_{t=1}^J \hat{y}_{\tau_j}^w(\tilde{T}) = \sum_{t=1}^J W\left(\tau_j\right) \hat{y}_{t\tau_j}^o, \quad \tilde{T}=T_m,T_m+1,\cdots,T.
\end{align*}
To model the right tail risk indicator, 
we set $W\left(\tau_j\right)=W_r\equiv \left(\tau_j\right)=\tau_j^2(\tau_1^2+ \tau_2^2+ \cdots+\tau_J^2)^{-1}$ and
$W\left(\tau_j\right)=W_l \left(\tau_j\right)=-(1-\tau_j)^2[(1-\tau_1)^2+ (1-\tau_2)^2+ \cdots+(1-\tau_J)^2]^{-1}$ to model the left tail risk indicator. \footnote{$W_r\left(\tau_j\right)$ and $W_l\left(\tau_j\right)$ are the unitized aforementioned weight $\tau_j^2$ evaluating the performance at right tail quantiles and $(1-\tau_j)^2$ evaluating the performance at left tail quantiles respectively such that $W_r(\tau_1)+W_r(\tau_2)+\cdots+W_r(\tau_J)=1$ and $W_l(\tau_1)+W_l(\tau_2)+\cdots+W_l(\tau_J)=-1$. } Therefore, the larger the values of the left and right tail risk indicators, $\operatorname{R(tail)_t}$, the greater the risk associated with 2- and 3-year bonds and 4- and 5-year bonds, respectively. 
\begin{figure}[H]
\centering
\includegraphics[width=1\linewidth]{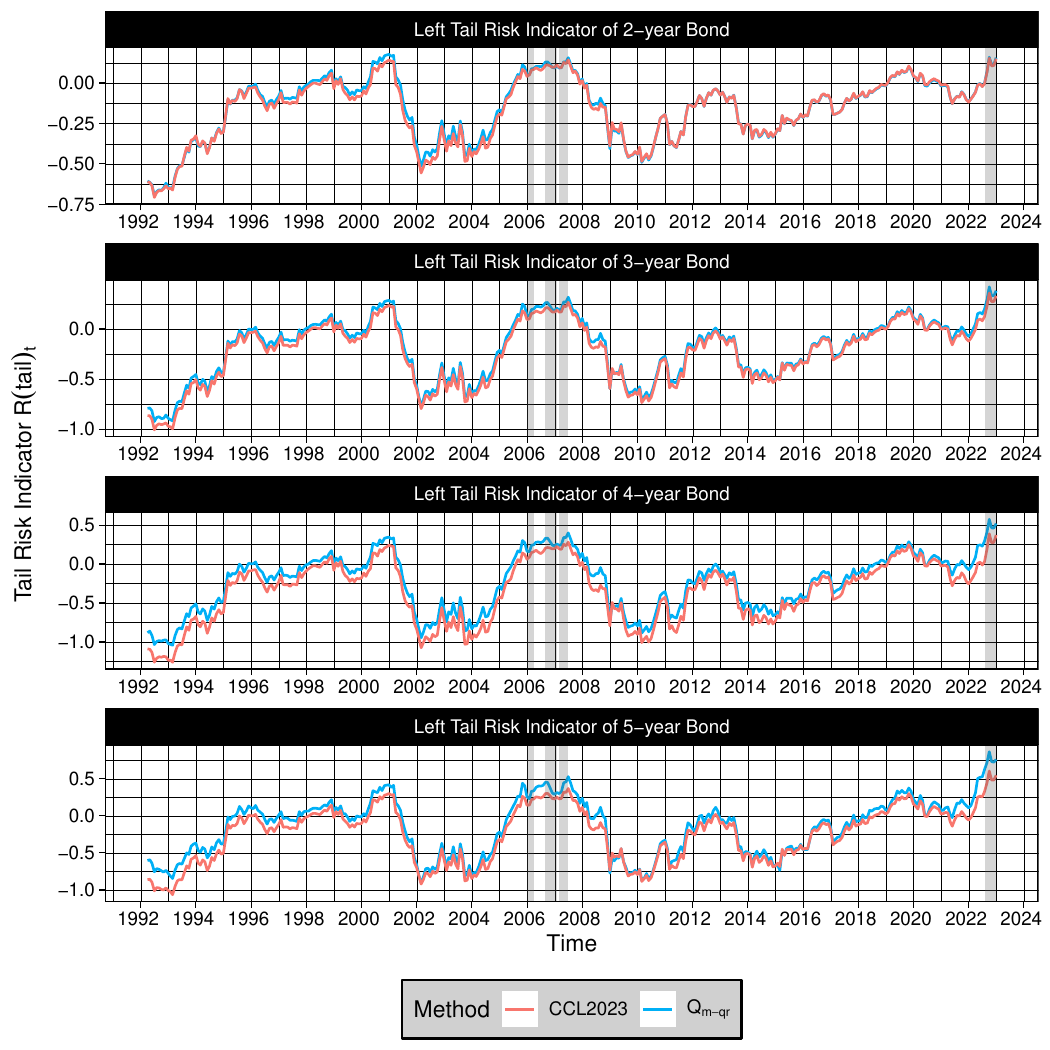}
\caption{Left Tail Risk Indicator in 1992--2022}
\label{lefttailrisk}
\end{figure} 
\begin{figure}[H]
\centering
\includegraphics[width=1\linewidth]{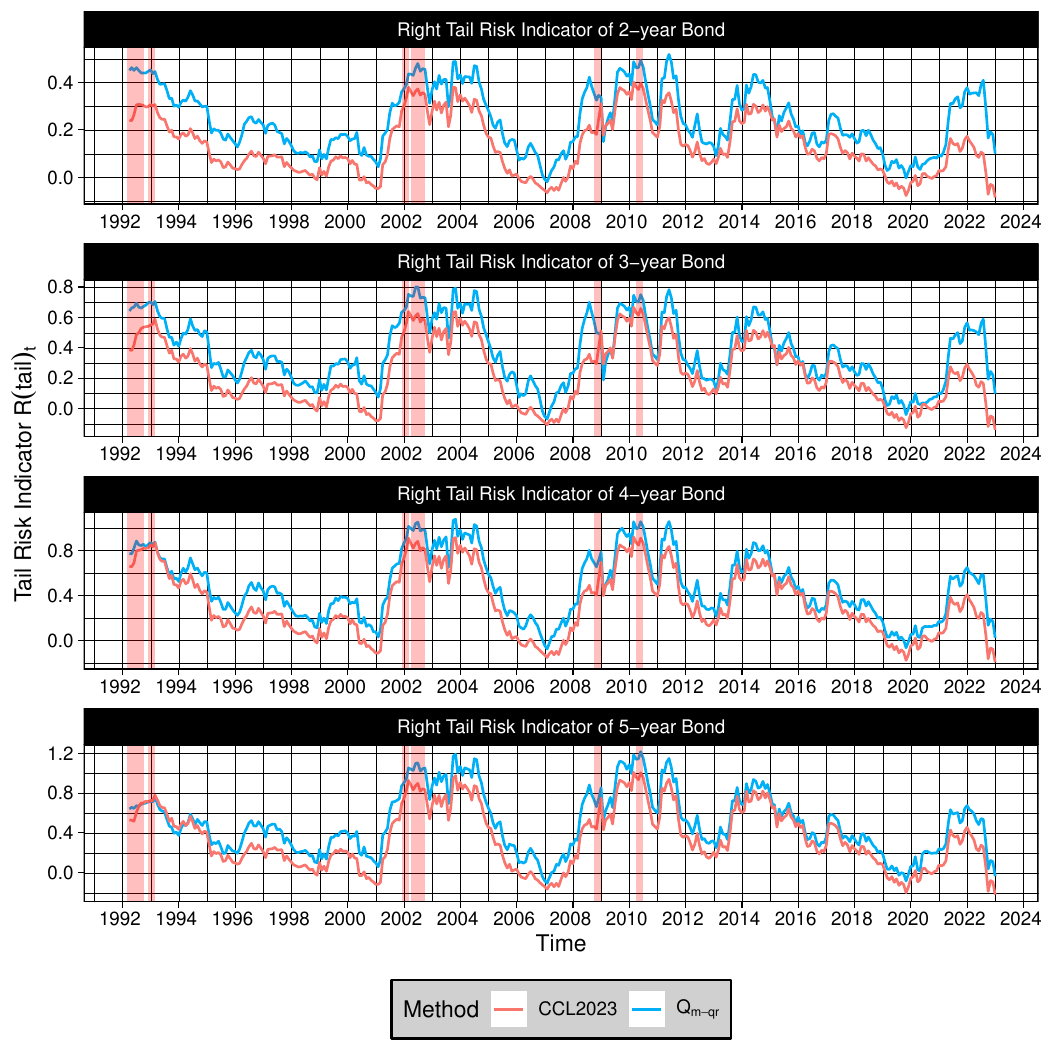}
\caption{Right Tail Risk Indicator in 1992--2022}
\label{righttailrisk}
\end{figure} 
In the following part, the left and the right tail indicator constructed based on $Q_{m-qr}$ and CCL2023 in 1992--2022 is shown in Figures \ref{lefttailrisk} and \ref{righttailrisk}. In Figure \ref{lefttailrisk}, the left tail risk indicator successfully highlights two peaks during the 2008 global financial crisis and the COVID-19 pandemic, indicating increased risk for 2- and 3-year bonds during these periods. Meanwhile, in Figure \ref{righttailrisk}, the right tail risk indicator successfully highlights four peaks during the periods of 1990-1992 (third oil crisis), 2002-2003 (dotcom bubble crisis), and twice for the 2008 global financial crisis, indicating elevated risk for 4- and 5-year bonds during these times. Although the left and right tail risk indicators based on CCL2023 have a similar shape to those based on $Q_{m-qr}$ and also capture these peaks, our approach is ahead of CCL2023. As a result, our tail risk indicators are capable of detecting bond risk more accurately and at an earlier stage than those derived from CCL2023. This superiority is attributed to several key factors: the test conducted in CCL2023 (and thus its tail risk indicators) did not find significant predictive power for TDH across all bond maturities at both tails quantiles, nor for TI at the right tail quantiles of 4- and 5-year bonds. This is in contrast to the findings associated with our $Q_{m-qr}$ test, which identified significant predictive power of those predictors under the same conditions. 
\section{Selected Proof}\label{app:E}
Theorems with WD predictors are easy to show using conventional technics, so we focus on the proof with SD predictors in this section. Equation numbers without the appendix section prefix refer to those in the main text.
Define $X_{t-1}= \left(T^{-1/2}, D_T^{-1}\tilde{x}_{t-1}^\top , \dot{D}_T^{-1} \tilde{v}_{t-1}^\top \right)^\top $.
\begin{lemma}\label{lem.ap1}
Under Assumption \ref{Assumption A.1} and \ref{Assumption A.2}, the following equation holds for SD predictors.
\begin{align}
\sum\limits_{t=1}^T X_{t-1} X_{t-1}^\top
\Rightarrow \operatorname{diag} [1,\Omega_{zx}^\top \Omega_{zz}^{-1} \Omega_{zx}, \int_0^1 \bar{J}_x^c (r) \bar{J}_x^c (r)^\top dr ] \nonumber
\end{align}
\end{lemma}
\begin{proof}[Proof of Lemma \ref{lem.ap1}]
By Lemma B3 and proof of Theorem A of online appendix of \cite{Kostakisetal2015}, for SD predictors, it follows that
\begin{align}\label{appe1}
D_T^{-1}\sum\limits_{t=1}^T \bar{z}_{t-1}\bar{z}_{t-1}^\top D_T^{-1}
\xrightarrow{P} \Omega_{zz}, \quad
D_T^{-1} \sum\limits_{t=1}^T \bar{z}_{t-1} x_{t-1}^\top D_T^{-1}
\xrightarrow{P} \Omega_{zx}.
\end{align}
By equation (\ref{firstep}) and (\ref{appe1}) and continuous mapping theorem, it follows that
\begin{align}\label{appe2}
\hat{\theta}&= D_T^{-2} \sum\limits_{t=1}^T x_{t-1}\bar{z}_{t-1}^\top \left(D_T^{-2}\sum\limits_{t=1}^T \bar{z}_{t-1}\bar{z}_{t-1}^\top \right)^{-1}
\Rightarrow \Omega_{zx}^\top\Omega_{zz}^{-1}.
\end{align}
By equation (\ref{appe1}) and (\ref{appe2}), definition of $\tilde{x}_{t-1}$ and continuous mapping theorem,
\begin{align}\label{appe4}
T^{-1/2} D_T^{-1}\sum\limits_{t=1}^T \tilde{x}_{t-1} = \hat{\theta} D_T^{-1} T^{-1/2}\sum\limits_{t=1}^T z_{t-1}
= O_p(1) O_p[T^{(\delta-1)/2}] = o_p(1).
\end{align}
Moreover, by the property of OLS, $\tilde{x}_{t-1}^\top$ and $\tilde{v}_{t-1}$ are orthogonal, i.e.,
\begin{align}\label{appe8}
\sum\limits_{t=1}^T \tilde{x}_{t-1}\tilde{v}_{t-1}^\top=0,\quad \sum\limits_{t=1}^T \tilde{v}_{t-1}=0.
\end{align}
By the equation $x_{t-1}= \hat{\mu}_x +\hat{\theta} z_{t-1}+ \tilde{v}_{t-1}$ and equation (\ref{appe8}), we have
\begin{align}\label{appe9}
\hat{\mu}_x = \frac{1}{T} \sum\limits_{t=1}^T x_{t-1}- \hat{\theta} \frac{1}{T}\sum\limits_{t=1}^T z_{t-1}.
\end{align}
By equations (\ref{appe2}) and (\ref{appe9}) and the fact that $x_{\lfloor rT \rfloor}/{\sqrt{T}} \Rightarrow J_x^c(r)$ for $0\le r\le 1$, it follows that
\begin{align}\label{appe10}
T^{-1/2}\hat{\mu}_x = \frac{1}{T^{3/2}} \sum\limits_{t=1}^T x_{t-1}+ o_p(1)\Rightarrow \int_0^1 J_x^c (r)dr.
\end{align}
By equation (\ref{appe1}), (\ref{appe2}), (\ref{appe4}) and (\ref{appe10}), definition of $\tilde{x}_{t-1}$ and continuous mapping theorem, 
\begin{align}\label{appe3}
D_T^{-1}\sum\limits_{t=1}^T \tilde{x}_{t-1} \tilde{x}_{t-1}^\top D_T^{-1} = \hat{\theta} D_T^{-1}\sum\limits_{t=1}^T z_{t-1} z_{t-1}^\top D_T^{-1} \hat{\theta}^\top + o_p(1)
\Rightarrow \Omega_{zx}^\top \Omega_{zz}^{-1} \Omega_{zx}. 
\end{align}
As per the equation $x_{t-1}= \hat{\mu}_x +\hat{\theta} z_{t-1}+ \tilde{v}_{t-1}$ and equation (\ref{appe10}), it becomes evident that
\begin{align}\label{appe5}
T^{-1/2} \tilde{v}_{\lfloor rT\rfloor} = T^{-1/2} x_{\lfloor rT\rfloor} - T^{-1/2}\hat{\mu}_x + o_p(1)\Rightarrow \bar{J}_x^c(r),\quad 0\leq r \leq 1.
\end{align}
Thus equation (\ref{appe5}) implies that
\begin{align}\label{appe6}
T^{-2}\sum\limits_{t=1}^T \tilde{v}_{t-1} \tilde{v}_{t-1}^\top
\Rightarrow \int_0^1 \bar{J}_x^c (r) \bar{J}_x^c (r)^\top dr.
\end{align}
By equations(\ref{appe4}), (\ref{appe8}), (\ref{appe3}) and (\ref{appe6}), it follows that
\begin{align}
\sum\limits_{t=1}^T X_{t-1} X_{t-1}^\top
\Rightarrow \operatorname{diag} [1,\Omega_{zx}^\top \Omega_{zz}^{-1} \Omega_{zx}, \int_0^1 \bar{J}_x^c (r) \bar{J}_x^c (r)^\top dr ]. \nonumber
\end{align}
\end{proof}
\begin{lemma}\label{lem.ap2}
Under Assumption \ref{Assumption A.1} and \ref{Assumption A.2}, the following equation holds for SD predictors.
\begin{align}
\sum\limits_{t=1}^T X_{t-1} \psi_\tau (u_{t\tau})
\Rightarrow \left\{ \operatorname{MN}\left[0,\operatorname{diag} (1,\Omega_{zx}^\top \Omega_{zz}^{-1} \Omega_{zx}) \right]^\top, \int_0^1 \bar{J}_x^c (r)^\top dB_{\psi_\tau}(r) \right\}^\top.
\end{align}
\end{lemma}
\begin{proof}[Proof of Lemma \ref{lem.ap1}]
First, the fact
\begin{align}\label{appe11}
\sum\limits_{t=1}^T \left(T^{-1/2}, D_T^{-1}z_{t-1}^\top \right)^\top \psi_\tau (u_{t\tau})\Rightarrow \operatorname{MN}\left\{0,\operatorname{diag} [1, \tau(1-\tau)\Omega_{zz} ] \right\}
\end{align}
holds by \cite{Lee2016}. In light of equation (\ref{appe2}) and (\ref{appe11}) and the continuous mapping theorem, it is clear that
\begin{align}\label{appe12}
&\sum\limits_{t=1}^T \left(T^{-1/2}, D_T^{-1}\tilde{x}_{t-1}^\top \right)^\top \psi_\tau (u_{t\tau}) =
\operatorname{diag}(1,\hat{\theta})\sum\limits_{t=1}^T \left(T^{-1/2}, D_T^{-1}z_{t-1}^\top \right)^\top \psi_\tau (u_{t\tau})\\
&\Rightarrow \operatorname{MN}\left[0,\operatorname{diag} (1,\tau(1-\tau)\Omega_{zx}^\top \Omega_{zz}^{-1} \Omega_{zx} ) \right]. \nonumber
\end{align}
FCLT shown in equation (\ref{eq3}) and the fact $x_{\lfloor rT \rfloor}/{\sqrt{T}} \Rightarrow J_x^c(r)$ for $0\le r\le 1$ imply the following equation by the continuous mapping theorem.
\begin{align}\label{appe13}
\frac{1}{T}\sum\limits_{t=1}^T \tilde{v}_{t-1} \psi_\tau (u_{t\tau})
\Rightarrow \int_0^1 \bar{J}_x^c (r) dB_{\psi_\tau}(r).
\end{align}
Using equations (\ref{appe12}) and (\ref{appe13}) and the joint convergence shown by equation (\ref{eq3}), we can determine that
\begin{align}
\sum\limits_{t=1}^T X_{t-1} \psi_\tau (u_{t\tau})
\Rightarrow \left\{ \operatorname{MN}\left[0,\operatorname{diag} (1,\tau(1-\tau)\Omega_{zx}^\top \Omega_{zz}^{-1} \Omega_{zx} ) \right]^\top, \int_0^1 \bar{J}_x^c (r)^\top dB_{\psi_\tau}(r) \right\}^\top.
\end{align}
\end{proof}
\begin{lemma}[Lemma A.1 of online appendix of CCL2023]\label{keylemma2}
Let $V_T(v)$ be a vector function that satisfies: (i) $-v^\top V_T(\lambda v)\geq -v^\top V_T(v)$ for any $\lambda\geq 1$. (ii) $\mathop{\sup}\limits_{\|v\| < M} \|V_T(v)+Nv-A_T\|=o_p(1)$ where $\|A_T\|=O_p(1)$, $0<M<\infty$, and $N$ is a positive-definite random matrix. Suppose that $v_T$ is a vector such that $\|V_T(v_T)\|=o_p(1)$, then $\|v_T\|=O_p(1)$ and $ v_T =N^{-1} A_T +o_p(1)$.
\end{lemma}
\begin{proof}[Proof of Proposition \ref{thm1}]
Following the same procedure of the proof of Theorem 1 in online appendix of CCL2023, Proposition \ref{thm1} holds by Lemma \ref{lem.ap1}, \ref{lem.ap2} and \ref{keylemma2}.
\end{proof}
\begin{proof}[Proof of Theorem \ref{thm2}]
From Proposition \ref{thm1} and equation (\ref{appe2}) and the orthogonality shown in equation (\ref{appe8}), we can deduce the following equation.
\begin{align}\label{appe14}
D_T\left(\hat{\beta}_\tau -\beta_\tau\right)&=\frac{1}{f_{u_\tau}(0)} \left( D_T^{-2} \sum\limits_{t=1}^T \overline{\tilde{x}}_{t-1} \overline{\tilde{x}}_{t-1}^\top \right)^{-1} D_T^{-1} \sum\limits_{t=1}^T \overline{\tilde{x}}_{t-1} \psi_\tau (u_{t\tau}) + o_p(1) \\
&=\frac{1}{f_{u_\tau}(0)} \left[ D_T^{-2} \sum\limits_{t=1}^T \hat{\theta}\overline{z}_{t-1} (\hat{\theta}\overline{z}_{t-1})^\top \right]^{-1} D_T^{-1} \sum\limits_{t=1}^T \overline{\tilde{x}}_{t-1} \psi_\tau (u_{t\tau}) + o_p(1) \nonumber \\
&=\frac{1}{f_{u_\tau}(0)} \left(\hat{\theta} D_T^{-2} \sum\limits_{t=1}^T \overline{z}_{t-1} \overline{z}_{t-1}^\top \hat{\theta}^\top \right)^{-1} \hat{\theta} D_T^{-1} \sum\limits_{t=1}^T \overline{z}_{t-1} \psi_\tau (u_{t\tau}) + o_p(1) \nonumber \\
&=\frac{1}{f_{u_\tau}(0)} (\hat{\theta}^\top)^{-1} \left( D_T^{-2} \sum\limits_{t=1}^T \overline{z}_{t-1} \overline{z}_{t-1}^\top \right)^{-1} \hat{\theta}^{-1}\hat{\theta} D_T^{-1} \sum\limits_{t=1}^T \overline{z}_{t-1} \psi_\tau (u_{t\tau}) + o_p(1) \nonumber \\
&=\frac{1}{f_{u_\tau}(0)} (\hat{\theta}^\top)^{-1} \left( D_T^{-2} \sum\limits_{t=1}^T \overline{z}_{t-1} \overline{z}_{t-1}^\top \right)^{-1} D_T^{-1} \sum\limits_{t=1}^T \overline{z}_{t-1} \psi_\tau (u_{t\tau}) + o_p(1) \nonumber \\
&=\frac{1}{f_{u_\tau}(0)} \left( D_T^{-2} \sum\limits_{t=1}^T \bar{z}_{t-1} x_{t-1}^\top \right)^{-1} D_T^{-1} \sum\limits_{t=1}^T \bar{z}_{t-1} \psi_\tau (u_{t\tau}) + o_p(1) \nonumber
\end{align}
Equation (\ref{appe1}), (\ref{appe11}) and (\ref{appe14}) and Slutsky's Theorem imply that
\begin{align}\label{appe15}
D_T\left(\hat{\beta}_\tau -\beta_\tau\right) \xrightarrow{d}
\begin{cases}
\operatorname{MN}\left[0_K,\frac{1}{f_{u_\tau}(0)^2 } \Omega_{zx}^{-1}\Omega_{zz}\left( \Omega_{zx}^{-1}\right)^\top\right],\quad SD;\\
\operatorname{N}\left[0_K,\frac{1}{f_{u_\tau}(0)^2 } \Omega_{zx}^{-1}\Omega_{zz}\left( \Omega_{zx}^{-1}\right)^\top\right],\quad WD;
\end{cases}.
\end{align}
By equations (\ref{appe14}) and (\ref{appe15}), Theorem \ref{thm2} holds.
\end{proof}
\begin{proof}[Proof of Proposition \ref{thm3}]
Using equations (\ref{dgjh3e1}), (\ref{dgjh3e2}), (\ref{dj74e6}) and (\ref{appe1}), Theorem \ref{thm2} and Slutsky's theorem, Proposition \ref{thm3} is proved.
\end{proof}
\begin{proof}[Proof of Proposition \ref{mulpropfdie3}]
Following the procedure of proof of Proposition 1 of \cite{liao2024robust}, Proposition \ref{mulpropfdie3} is proved by the definition shown in equation (\ref{dgjh3e2}).
\end{proof}
The proof of equation (\ref{muldef2new}) and (\ref{muldef3new}) closely mirrors the proof of Theorem \ref{thm2}, therefore, we omit it for brevity.
\begin{proof}[Proof of Equation (\ref{dj74e5})]
Equation (\ref{kdj38f3}) and the fact that $\overrightarrow{\xi_{t-1}}$ follows i.i.d. $\operatorname{N}(0_{M},I_M)$ imply that
$ \frac{1}{\sqrt{T}} \sum\limits_{t=1}^T \xi_{j,t-1}^{(i)} \psi_\tau (u_{t\tau}) \bigg | \mathcal{F}_T \sim
\operatorname{N}\left[0, \frac{1}{T} \sum\limits_{t=1}^T \psi_\tau (u_{t\tau})^2 \right]$ 
and thus 
\begin{align}\label{kdj48f4} 
\left[ \frac{1}{T} \sum\limits_{t=1}^T \psi_\tau (u_{t\tau})^2\right]^{-1/2} \frac{1}{\sqrt{T}} \sum\limits_{t=1}^T \xi_{j,t-1}^{(i)} \psi_\tau (u_{t\tau}) \bigg | \mathcal{F}_T \sim
\operatorname{N} (0,1), 
\end{align}
where $i=1,\cdots,M_1$ and $j=1,\cdots,M_2$ and $\hat{l}_{\tau}^{(i)}=\left[\hat{l}_{1,\tau}^{(i)},\hat{l}_{2,\tau}^{(i)},\cdots,\hat{l}_{M_2,\tau}^{(i)}\right]^\top$.
By the independence of $\xi_{j,t-1}^{(i)}$ for $j=1,\cdots,M_2$ and $i=1,\cdots,M_1$, it can be verified that conditional on the whole sample $\mathcal{F}_T$, $\left[ \frac{1}{T} \sum\nolimits_{t=1}^T \psi_\tau (u_{t\tau})^2\right]^{-1/2} \frac{1}{\sqrt{T}} \sum\nolimits_{t=1}^T \xi_{j,t-1}^{(i)} \psi_\tau (u_{t\tau}) $ are uncorrelated for all i and j.
Therefore, equation \eqref{kdj48f4} and the uncorrelation between the (i,j)th items of $\left[ \frac{1}{T} \sum\nolimits_{t=1}^T \psi_\tau (u_{t\tau})^2\right]^{-1/2} \frac{1}{\sqrt{T}} \sum\nolimits_{t=1}^T \xi_{j,t-1}^{(i)} \psi_\tau (u_{t\tau}) $ lead to the following joint convergence.
\begin{align}\label{dj74e4}
&\left[ \frac{1}{T} \sum_{t=1}^T \psi_\tau (u_{t\tau})^2\right]^{-1/2} \Bigg\{\frac{1}{\sqrt{T}} \sum_{t=1}^T (\xi_{t-1}^{(1)})^\top \psi_\tau (u_{t\tau}), \frac{1}{\sqrt{T}} \sum_{t=1}^T (\xi_{2,t-1}^{(2)})^\top \psi_\tau (u_{t\tau}),\cdots, \\
&\frac{1}{\sqrt{T}} \sum_{t=1}^T (\xi_{t-1}^{(M_1)})^\top \psi_\tau (u_{t\tau})\Bigg\} \Bigg | \mathcal{F}_T \sim
\operatorname{N}\left[\operatorname{0_M}, \operatorname{I_M}\right],\nonumber
\end{align}
where $M=M_1M_2$. Hence,
\begin{align}\label{kdj68f7} 
\left[ \frac{1}{T} \sum\limits_{t=1}^T \psi_\tau (u_{t\tau})^2\right]^{-1} \sum_{i=1}^{M_1}\sum_{j=1}^{M_2} \left[\frac{1}{\sqrt{T}} \sum\limits_{t=1}^T \xi_{j,t-1}^{(i)} \psi_\tau (u_{t\tau})\right]^2 \bigg | \mathcal{F}_T \sim \chi_M^2, 
\end{align}
Using equation (\ref{kdj68f7}), we can derive that
\begin{align}\label{kuy1} 
\operatorname{E} \left\{ \left[ \frac{1}{T} \sum\limits_{t=1}^T \psi_\tau (u_{t\tau})^2\right]^{-1} \sum_{i=1}^{M_1}\sum_{j=1}^{M_2} \left[\frac{1}{\sqrt{T}} \sum\limits_{t=1}^T \xi_{j,t-1}^{(i)} \psi_\tau (u_{t\tau})\right]^2 \bigg | \mathcal{F}_T \right\} = M,\\ \label{k1uy1} 
\operatorname{Var} \left\{ \left[ \frac{1}{T} \sum\limits_{t=1}^T \psi_\tau (u_{t\tau})^2\right]^{-1} \sum_{i=1}^{M_1}\sum_{j=1}^{M_2} \left[\frac{1}{\sqrt{T}} \sum\limits_{t=1}^T \xi_{j,t-1}^{(i)} \psi_\tau (u_{t\tau})\right]^2 \bigg | \mathcal{F}_T \right\} = 2M.
\end{align}
From equations (\ref{kdj38f3}) and (\ref{k1uy1}),
we can deduce that
\begin{align}\label{kuy33}
&\operatorname{Var}\left\{\frac{1}{M} \sum_{i=1}^{M_1}\sum_{j=1}^{M_2} \left[\sqrt{T}\hat{l}_{j,\tau}^{(i)}\right]^2 \bigg| \mathcal{F}_T \right\} \\
& = \frac{1}{f_{u_\tau}(0)^4} \frac{1}{M^2} \operatorname{Var}\left\{\sum_{i=1}^{M_1}\sum_{j=1}^{M_2} \left[\frac{1}{\sqrt{T}} \sum\limits_{t=1}^T \xi_{j,t-1}^{(i)} \psi_\tau (u_{t\tau})\right]^2 \bigg| \mathcal{F}_T \right\}+o_p(1) \nonumber\\
& = \frac{1}{f_{u_\tau}(0)^4} \frac{1}{M^2} \operatorname{Var}\left\{\tau(1-\tau)\left[ \frac{1}{T} \sum\limits_{t=1}^T \psi_\tau (u_{t\tau})^2\right]^{-1} \sum_{i=1}^{M_1}\sum_{j=1}^{M_2} \left[\frac{1}{\sqrt{T}} \sum\limits_{t=1}^T \xi_{j,t-1}^{(i)} \psi_\tau (u_{t\tau})\right]^2 
+ o_p(1)\bigg| \mathcal{F}_T \right\}+o_p(1) \nonumber\\ 
& = \frac{[\tau(1-\tau)]^2 }{f_{u_\tau}(0)^4} \frac{1}{M^2} \operatorname{Var}\left\{\left[ \frac{1}{T} \sum\limits_{t=1}^T \psi_\tau (u_{t\tau})^2\right]^{-1} \sum_{i=1}^{M_1}\sum_{j=1}^{M_2} \left[\frac{1}{\sqrt{T}} \sum\limits_{t=1}^T \xi_{j,t-1}^{(i)} \psi_\tau (u_{t\tau})\right]^2 
\bigg| \mathcal{F}_T \right\}+o_p(1) \nonumber\\ 
& = \frac{[\tau(1-\tau)]^2 }{f_{u_\tau}(0)^4} \frac{1}{M^2} 2M+o_p(1) \nonumber\\ 
& = \frac{[\tau(1-\tau)]^2 }{f_{u_\tau}(0)^4} \frac{2}{M} +o_p(1)\rightarrow 0, \nonumber
\end{align} 
as $M\rightarrow\infty$ and $T\rightarrow\infty$. Similarly, equations (\ref{kdj38f3}) and (\ref{kuy1}) imply that
\begin{align}\label{kuy3}
&\operatorname{E}\left\{\frac{1}{M} \sum_{i=1}^{M_1}\sum_{j=1}^{M_2} \left[\sqrt{T}\hat{l}_{j,\tau}^{(i)}\right]^2 \bigg| \mathcal{F}_T \right\} \\
&= \frac{1}{f_{u_\tau}(0)^2} \frac{1}{M} \operatorname{E}\left\{\sum_{i=1}^{M_1}\sum_{j=1}^{M_2} \left[\frac{1}{\sqrt{T}} \sum\limits_{t=1}^T \xi_{j,t-1}^{(i)} \psi_\tau (u_{t\tau})\right]^2 \bigg| \mathcal{F}_T\right\} +o_p(1) \nonumber\\
&= \frac{\tau(1-\tau)}{f_{u_\tau}(0)^2} \frac{1}{M} \operatorname{E}\left\{\left[ \frac{1}{T} \sum\limits_{t=1}^T \psi_\tau (u_{t\tau})^2\right]^{-1} \sum_{i=1}^{M_1}\sum_{j=1}^{M_2} \left[\frac{1}{\sqrt{T}} \sum\limits_{t=1}^T \xi_{j,t-1}^{(i)} \psi_\tau (u_{t\tau})\right]^2 +o_p(1) \bigg| \mathcal{F}_T \right\} +o_p(1) \nonumber \\
&= \frac{\tau(1-\tau)}{f_{u_\tau}(0)^2} \frac{1}{M} \operatorname{E}\left\{\left[ \frac{1}{T} \sum\limits_{t=1}^T \psi_\tau (u_{t\tau})^2\right]^{-1} \sum_{i=1}^{M_1}\sum_{j=1}^{M_2} \left[\frac{1}{\sqrt{T}} \sum\limits_{t=1}^T \xi_{j,t-1}^{(i)} \psi_\tau (u_{t\tau})\right]^2 \bigg| \mathcal{F}_T \right\} +o_p(1) \nonumber \\
&= \frac{\tau(1-\tau)}{f_{u_\tau}(0)^2} \frac{1}{M} M +o_p(1)= \frac{\tau(1-\tau)}{f_{u_\tau}(0)^2} +o_p(1), \nonumber 
\end{align} 
Equations (\ref{kuy33}) and (\ref{kuy3}) and the continuous mapping theorem imply that
\begin{align} 
\frac{1}{M} \sum_{i=1}^{M_1}\sum_{j=1}^{M_2} \left[\sqrt{T}\hat{l}_{j,\tau}^{(i)}\right]^2 \bigg| \mathcal{F}_T & \xrightarrow{P} \frac{\tau(1-\tau)}{f_{u_\tau}(0)^2} \frac{1}{M} M=\frac{1}{f_{u_\tau}(0)^2 } \tau(1-\tau), 
\end{align}
as $M\rightarrow\infty$ and $T\rightarrow\infty$. 
By the iterative law of expectation and equations (\ref{kuy33}) and (\ref{kuy3}), we have
\begin{align}\label{kuy4}
&\operatorname{E}\left\{\frac{1}{M} \sum_{i=1}^{M_1}\sum_{j=1}^{M_2} \left[\sqrt{T}\hat{l}_{j,\tau}^{(i)}\right]^2 \right\} = \frac{\tau(1-\tau)}{f_{u_\tau}(0)^2} +o_p(1), \\ \label{kuy5}
& \operatorname{Var}\left\{\frac{1}{M} \sum_{i=1}^{M_1}\sum_{j=1}^{M_2} \left[\sqrt{T}\hat{l}_{j,\tau}^{(i)}\right]^2 \right\} \rightarrow 0, 
\end{align} 
By equations (\ref{kuy4}) and (\ref{kuy5}), it follows that
\begin{align} 
\frac{1}{M} \sum_{i=1}^{M_1}\sum_{j=1}^{M_2} \left[\sqrt{T}\hat{l}_{j,\tau}^{(i)}\right]^2 & \xrightarrow{P} \frac{1}{f_{u_\tau}(0)^2 } \tau(1-\tau), 
\end{align}
as $M\rightarrow\infty$ and $T\rightarrow\infty$. 
\end{proof}
\begin{lemma}\label{sasb69}
Under Assumption \ref{Assumption A.1} and \ref{Assumption A.2}, it follows that
\begin{align}
&S_a \Rightarrow \tilde{S}_a =
\begin{cases}
J_x^c(1) J_x^c(\lambda)^\top \left[ J_x^c(\lambda)^\top J_x^c(\lambda) \right]^{-1},\quad \text{SD}; \\
B_v(1) B_v(\lambda)^\top \left[ B_v(\lambda)^\top B_v(\lambda) \right]^{-1},\quad \text{WD};
\end{cases}\\
&S_b \Rightarrow \tilde{S}_b =
\begin{cases}
J_x^c(1) \left[J_x^c(1) - J_x^c(\lambda) \right]^\top \left\{ \left[J_x^c(1) - J_x^c(\lambda) \right]^\top \left[J_x^c(1) - J_x^c(\lambda) \right] \right\}^{-1},\quad \text{SD}; \\
B_v(1) \left[B_v(1) - B_v(\lambda) \right]^\top \left\{ \left[B_v(1) - B_v(\lambda) \right]^\top \left[B_v(1) - B_v(\lambda) \right] \right\}^{-1},\quad \text{WD};
\end{cases}.\nonumber
\end{align}
\end{lemma}
\begin{proof}[Proof of Lemma \ref{sasb69}, Theorem \ref{multh1m}]
The proofs are very similar to the Theorem 2 of \cite{liao2024robust}, and thus are omitted for brevity.
\end{proof}
\begin{proof}[Proof of Proposition \ref{t7sf8f1}]
The proof closely follows the proof of Theorem 3 of \cite{liao2024robust}, and is therefore omitted.
\end{proof}
\begin{proof}[Proof of Proposition \ref{mulpropp2}]
The proof directly follows from the proof of Proposition 2 in \cite{liao2024robust}.
\end{proof}
\begin{lemma}\label{app1djgh}
Under Assumption \ref{Assumption A.1} and \ref{Assumption A.2}, it follows that
$D_T^{-2}\operatorname{\widehat{Avar}}(\hat{\beta}_\tau^l )\Rightarrow \operatorname{ Avar}(\hat{\beta}_\tau^l )$.
\end{lemma}
\begin{proof}[Proof of Lemma \ref{app1djgh}]
Lemma \ref{lem.ap1} and \ref{lem.ap2}, equation (\ref{dj74e6}) imply Lemma \ref{app1djgh}.
\end{proof}
\begin{proof}[Proof of Proposition \ref{mulkeythe2}]
By equation (\ref{defqc1}), it follows that
\begin{align}
{\check{Q}_{l-qr}}&= \frac{ R \hat{\beta}_\tau^l -r_\tau }{ \left[ R\operatorname{\widehat{Avar}}(\hat{\beta}_\tau^l )R^\top\right]^{1/2}} = \frac{ R \hat{\beta}_\tau^l- R\beta_\tau + R\beta_\tau -r_\tau }{ \left[ R\operatorname{\widehat{Avar}}(\hat{\beta}_\tau^l )R^\top\right]^{1/2}} \\
&= \frac{ R (\hat{\beta}_\tau^l- \beta_\tau) }{ \left[ R\operatorname{\widehat{Avar}}(\hat{\beta}_\tau^l )R^\top\right]^{1/2}} + \frac{ R\beta_\tau -r_\tau }{ \left[ R\operatorname{\widehat{Avar}}(\hat{\beta}_\tau^l )R^\top\right]^{1/2}}\nonumber
\end{align}
As per Lemma \ref{app1djgh} and Theorem \ref{multh1m} and the continuous mapping theorem, it become evident that
\begin{align}\label{appe16}
\frac{ R \hat{\beta}_\tau^l- \beta_\tau }{ \left[ R\operatorname{\widehat{Avar}}(\hat{\beta}_\tau^l )R^\top\right]^{1/2}} \xrightarrow{d} \operatorname{N}(0,1).
\end{align}
And Lemma \ref{app1djgh} implies that
\begin{align}\label{appe17}
\frac{ R\beta_\tau -r_\tau }{ \left[ R\operatorname{\widehat{Avar}}(\hat{\beta}_\tau^l )R^\top\right]^{1/2}}\Rightarrow \left[ R\operatorname{Avar}(\hat{\beta}_\tau^l)R^\top \right]^{-1/2} b_\tau
\end{align}
Using equations (\ref{appe16}) and (\ref{appe17}) and the continuous mapping theorem, we have the following result that
\begin{align}
\check{Q}_{l-qr} \xrightarrow{d} \check{Q}_{l-qr}^* \overset{d}{=} \operatorname{N}(0,1)+\left[ R\operatorname{Avar}(\hat{\beta}_\tau^l)R^\top \right]^{-1/2} b_\tau
\end{align}
when J=1.
Moreover, by equation (\ref{walde2tes}), it follows that
\begin{align}\label{jdj767g}
{Q_{l-qr}} &= \left(R \hat{\beta}_\tau^l -r_\tau \right)^\top \left\{R \operatorname{\widehat{Avar}}(\hat{\beta}_\tau^l ) R^\top \right\}^{-1} \left(R \hat{\beta}_\tau^l -r_\tau \right) \\
&= \left(R \hat{\beta}_\tau^l- R\beta_\tau + R\beta_\tau -r_\tau \right)^\top \left\{R \operatorname{\widehat{Avar}}(\hat{\beta}_\tau^l ) R^\top \right\}^{-1} \left(R \hat{\beta}_\tau^l - R\beta_\tau + R\beta_\tau-r_\tau \right) \nonumber\\
&= \left(R \hat{\beta}_\tau^l- R\beta_\tau \right)^\top \left\{R \operatorname{\widehat{Avar}}(\hat{\beta}_\tau^l ) R^\top \right\}^{-1} \left(R \hat{\beta}_\tau^l - R\beta_\tau \right) \nonumber\\
&+ \left(R \hat{\beta}_\tau^l- R\beta_\tau \right)^\top \left\{R \operatorname{\widehat{Avar}}(\hat{\beta}_\tau^l ) R^\top \right\}^{-1} \left( R\beta_\tau-r_\tau \right) \nonumber\\
&+ \left( R\beta_\tau -r_\tau \right)^\top \left\{R \operatorname{\widehat{Avar}}(\hat{\beta}_\tau^l ) R^\top \right\}^{-1} \left(R \hat{\beta}_\tau^l - R\beta_\tau \right) \nonumber\\
&+ \left( R\beta_\tau -r_\tau \right)^\top \left\{R \operatorname{\widehat{Avar}}(\hat{\beta}_\tau^l ) R^\top \right\}^{-1} \left( R\beta_\tau-r_\tau \right) \nonumber
\end{align}
As per Lemma \ref{app1djgh} and Theorem \ref{multh1m} and the continuous mapping theorem, it becomes evident that
\begin{align}\label{appe18}
\left(R \hat{\beta}_\tau^l- R\beta_\tau \right)^\top \left\{R \operatorname{\widehat{Avar}}(\hat{\beta}_\tau^l ) R^\top \right\}^{-1} \left(R \hat{\beta}_\tau^l - R\beta_\tau \right)
\Rightarrow \chi_J^2,
\end{align}
and
\begin{align}\label{appe19}
\left(R \hat{\beta}_\tau^l- R\beta_\tau \right)^\top \left\{R \operatorname{\widehat{Avar}}(\hat{\beta}_\tau^l ) R^\top \right\}^{-1} \left( R\beta_\tau-r_\tau \right) \xrightarrow{d} Z_Q=Z_\beta^\top \left[ R\operatorname{Avar}(\hat{\beta}_\tau^l)R^\top \right]^{-1} b_\tau,
\end{align}
and
\begin{align}\label{appe20}
\left( R\beta_\tau-r_\tau \right) ^\top \left\{R \operatorname{\widehat{Avar}}(\hat{\beta}_\tau^l ) R^\top \right\}^{-1} \left(R \hat{\beta}_\tau^l- R\beta_\tau \right) \xrightarrow{d} Z_Q^\top= Z_Q.
\end{align}
And Lemma \ref{app1djgh} implies that
\begin{align}\label{appe21}
\left( R\beta_\tau -r_\tau \right)^\top \left\{R \operatorname{\widehat{Avar}}(\hat{\beta}_\tau^l ) R^\top \right\}^{-1} \left( R\beta_\tau-r_\tau \right) \xrightarrow{d} b_\tau^\top \left[ R\operatorname{Avar}(\hat{\beta}_\tau^l)R^\top \right]^{-1} b_\tau.
\end{align}
By equations (\ref{jdj767g}), (\ref{appe18}), (\ref{appe19}), (\ref{appe20}) and (\ref{appe21}) and Slutsky's Theorem, Proposition \ref{mulkeythe2} holds.
\end{proof}
The proof of equation (\ref{dk4r2jgj1}) is very similar to the proof of Proposition \ref{mulkeythe2}, so we omit it.
\begin{proof}[Proof of Theorem \ref{jg8jgh1}]
By Proposition \ref{mulkeythe2} and equations (\ref{oraqu2np8}) and (\ref{dk4r2jgj1}) and the continuous mapping theorem, Theorem \ref{jg8jgh1} holds.
\end{proof}
\begin{proof}[Proof of Theorem \ref{jg8jghh1}]
By Proposition \ref{mulkeythe2} and equations (\ref{nocoraqu2np8}) and (\ref{2dk4r2jgj1}) and the continuous mapping theorem, Theorem \ref{jg8jghh1} holds.
\end{proof}
\end{document}